\documentclass[letterpaper,12pt,titlepage,oneside,final]{book}
 
\newcommand{\href}[1]{#1} 

\usepackage{ifthen}
\newboolean{PrintVersion}
\setboolean{PrintVersion}{false} 

\usepackage{amsmath,amssymb,amstext,enumitem} 
\usepackage[pdftex]{graphicx} 

\usepackage[numbers, sort&compress]{natbib}
\usepackage[all]{xy}
\usepackage[pdftex,letterpaper=true,pagebackref=false]{hyperref} 
\hypersetup{
    plainpages=false,       
    pdfpagelabels=true,     
    bookmarks=true,         
    unicode=false,          
    pdftoolbar=true,        
    pdfmenubar=true,        
    pdffitwindow=false,     
    pdfstartview={FitH},    
    pdftitle={A Classification of (2+1)D Topological Phases with Symmetries},    
    pdfauthor={Tian Lan},    
    pdfnewwindow=true,      
    colorlinks=true,        
    linkcolor=blue,         
    citecolor=green,        
    filecolor=magenta,      
    urlcolor=cyan           
}
\ifthenelse{\boolean{PrintVersion}}{   
\hypersetup{	
    citecolor=black,%
    filecolor=black,%
    linkcolor=black,%
    urlcolor=black}
}{} 

\usepackage{amsthm,bm}
\usepackage{pdfpages}

\newcommand{\stck}[2]{{#1\atop #2}}

\DeclareMathOperator{\Tr}{Tr}

\DeclareMathOperator{\Hom}{Hom}
\DeclareMathOperator{\Rep}{Rep}
\DeclareMathOperator{\sRep}{sRep}
\DeclareMathOperator{\Rp}{Rep}
\DeclareMathOperator{\sRp}{sRep}
\DeclareMathOperator{\id}{id}

\newcommand{\ob}{\mathrm{Ob}}
\newcommand{\xt}{\times}
\newcommand{\ot}{\otimes}
\newcommand{\bt}{\boxtimes}

\newcommand{\ket}{\rangle}
\newcommand{\inj}{\hookrightarrow}

\newcommand{\Isb}{\overline{\text{Ising}}}
\newcommand{\Is}{{\text{Ising}}}
\newcommand{\cen}[2]{{#1}^\mathrm{cen}_{#2}}
\newcommand{\Ab}[1]{{#1}_{Ab}}

\newcommand{\one}{\mathbf{1}}

\newcommand{\Z}{\mathbb{Z}}

\newcommand{\C}{\mathbb{C}}
\newcommand{\Q}{\mathbb{Q}}

\newcommand{\ii}{\mathrm{i}}
\newcommand{\ee}{\mathrm{e}}

\newcommand{\<}{\langle} 
\renewcommand{\>}{\rangle} 

\newcommand{\Ref}[1]{Ref.~\cite{#1}}

\newcommand{\sgn}{{\rm sgn}}

\newcommand{\ie}{{i.e., }} 

\newcommand{\cA}{ {\cal A} } 
\newcommand{\cB}{ {\cal B} }
\newcommand{\cC}{ {\cal C} } 
\newcommand{\cD}{ {\cal D} } 
\newcommand{\cE}{ {\cal E} } 
\newcommand{\cF}{ {\cal F} }

\newcommand{\cI}{{\cal I}}

\newcommand{\cK}{ {\cal K} } 
 
\newcommand{\cM}{ {\cal M} } 
\newcommand{\cN}{ {\cal N} }

\newcommand{\cV}{ {\cal V} } 
\newcommand{\cW}{ {\cal W} } 
\newcommand{\cZ}{ {\cal Z} }

\newcommand{\Th}{\Theta}

\newcommand{\bpm}{\begin{pmatrix}}
\newcommand{\epm}{\end{pmatrix}}
\newcommand{\bmm}{\begin{matrix}}
\newcommand{\emm}{\end{matrix}}

\newcommand{\mce}[1]{\text{UMTC}_{/#1}}
\newcommand{\aute}{\mathrm{Aut}}

\theoremstyle{definition}
\newtheorem{thm}{Theorem}
\newtheorem{cor}[thm]{Corollary}

\newtheorem{dfn}{Definition}
\newtheorem{rmk}{Remark}
\newtheorem{lem}{Lemma}
\newtheorem{eg}{Example}
\newcommand{\mext}{\mathcal{M}_{ext}}
\newcommand{\ov}{\overline}
\newcommand{\Ve}{\mathbf{Vec}}
\newcommand{\Hilb}{\mathbf{Hilb}}
\newcommand{\sHilb}{\mathbf{sHilb}}
\newcommand{\Fun}{\mathrm{Fun}}

\let\origdoublepage\cleardoublepage
\newcommand{\clearemptydoublepage}{%
  \clearpage{\pagestyle{empty}\origdoublepage}}
\let\cleardoublepage\clearemptydoublepage

\begin{document}

\pagestyle{empty}
\pagenumbering{roman}

\begin{titlepage}
        \begin{center}
        \vspace*{1.0cm}

        \Huge
        {\bf A Classification of (2+1)D Topological Phases with Symmetries }

        \vspace*{1.0cm}

        \normalsize
        by \\

        \vspace*{1.0cm}

        \Large
        Tian Lan \\

        \vspace*{3.0cm}

        \normalsize
        A thesis \\
        presented to the University of Waterloo \\ 
        in fulfillment of the \\
        thesis requirement for the degree of \\
        Doctor of Philosophy \\
        in \\
        Physics \\

        \vspace*{2.0cm}

        Waterloo, Ontario, Canada, 2017 \\

        \vspace*{1.0cm}

        \copyright\ Tian Lan 2017 \\
        \end{center}
\end{titlepage}

\pagestyle{plain}
\setcounter{page}{2}

\cleardoublepage 
 
\includepdf{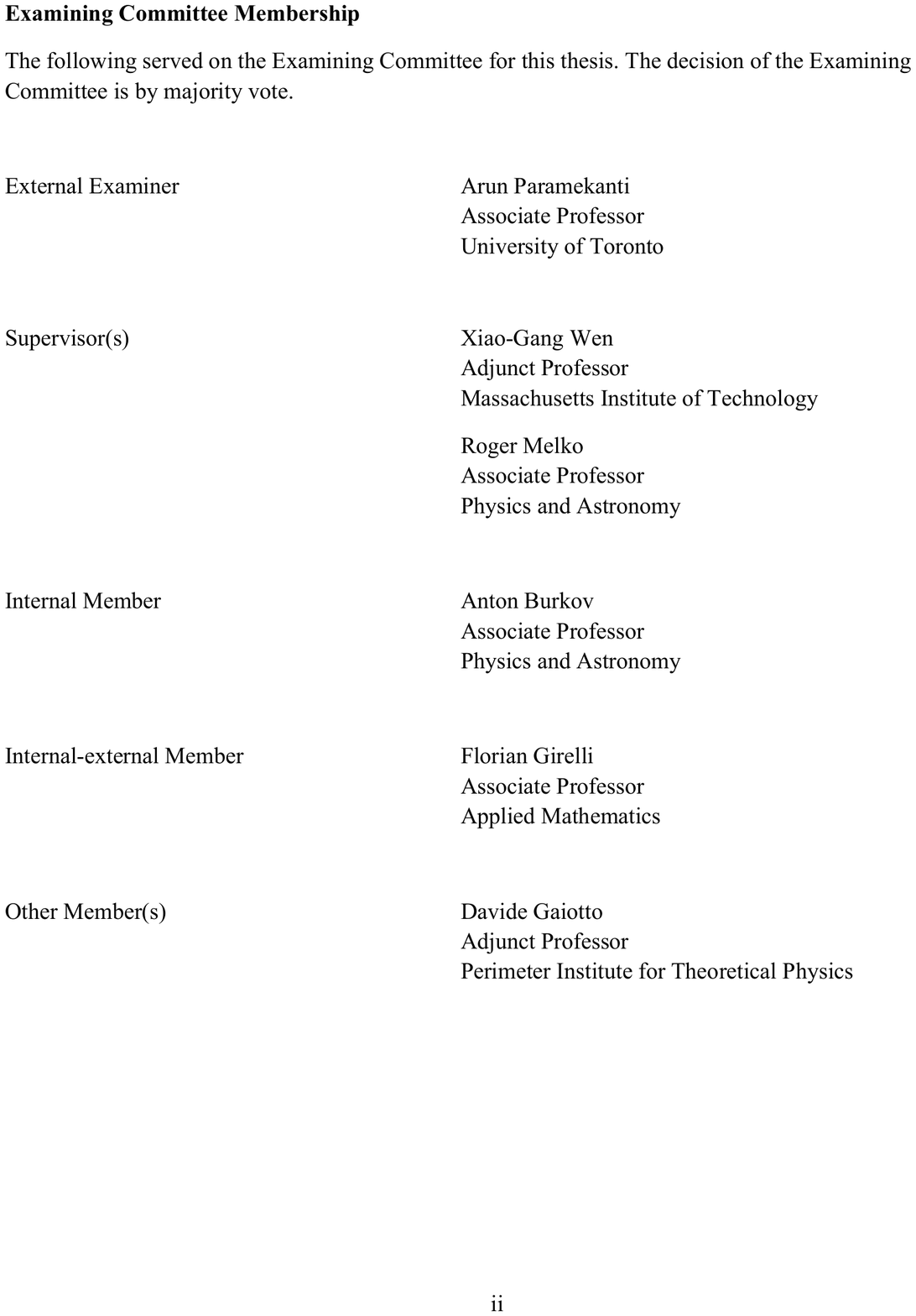}

  \noindent
I hereby declare that I am the sole author of this thesis. This is a true copy of the thesis, including any required final revisions, as accepted by my examiners.

  \bigskip
  
  \noindent
I understand that my thesis may be made electronically available to the public.

\cleardoublepage

\begin{center}\textbf{Abstract}\end{center}

This thesis aims at concluding the classification results for topological
phases with symmetry in 2+1 dimensions.
First, we know that ``trivial'' (i.e., not topological) phases with symmetry
can be classified by Landau symmetry breaking theory. If the Hamiltonian of the
system has symmetry group $G_H$, the symmetry of the ground state, however, can be
spontaneously broken and thus a smaller group $G$. In other words, different
symmetry breaking patterns are classified by $G\subset G_H$.

For topological phases, symmetry breaking is always a possibility. In this
thesis, for simplicity we assume that there is no symmetry breaking;
equivalently we always work with the symmetry group $G$ of the ground states.
We also restrict to the case that $G$ is finite and on-site.

The classification of topological phases is far beyond symmetry breaking
theory. There are two main exotic features in
topological phases: (1) protected chiral, or non-chiral but still gapless, edge
states; (2) fractional, or (even more wild) anyonic, quasiparticle
excitations that can have non-integer internal degrees of freedom,
fractional charges or spins and
non-Abelian braiding statistics.
In this thesis we achieved a full classification by studying the properties
of these exotic quasiparticle excitations. 

Firstly, we want to distinguish the exotic excitations with the ordinary ones.
Here the criteria is whether excitations can be created or annihilated by
\emph{local operators}. The ordinary ones can be created by local operators,
such as a spin flip in the Ising model, and will be referred to as {local
excitations}. The exotic ones can not be created by local operators, for
example a quasi-hole excitation with 1/3 charge in the $\nu=1/3$ Laughlin
state. Local operators can only create quasi-hole/quasi-electron pairs but never
a single quasi-hole. They will be referred to as \emph{topological excitations}.

Secondly, we know that local excitations always carries the representations of
the symmetry group $G$. This constitutes the first layer of our classification, a symmetric fusion category, $\cE=\Rep(G)$ for boson systems
or $\cE=\sRep(G^f)$ for fermion systems,
consisting of the representations of the symmetry group and describing the local
excitations with symmetry.

Thirdly, when we combine local excitations and topological excitations
together, all the excitations in the phase must form a consistent anyon
model. This constitutes the second layer of our classification, a unitary
braided fusion category $\cC$ describing all the
quasiparticle excitations in the bulk. It is clear that $\cE\subset\cC$. Due to braiding non-degeneracy, the
subset of 
excitations that have trivial mutual statistics with all excitations (namely the
M\"uger center) must coincide with the local excitations $\cE$. Thus, $\cC$ is a
non-degenerate unitary braided fusion category over $\cE$, or a $\mce{\cE}$.

However, it turns out that only the information of excitations in the original phase
is not enough. Most importantly, we miss the information of the protected edge
states. To fix this weak point, we consider the extrinsic symmetry defects, and
promote them to dynamical excitations, a.k.a., ``gauge the symmetry''.
We fully gauge the symmetry such that the gauged theory is a bosonic topological phase
with no symmetry, described by a unitary modular tensor category $\cM$, which
constitutes the third layer of our classification. It is clear that $\cM$
contains all excitations in the original phase, $\cC\subset\cM$, plus additional
excitations coming from symmetry defects. It is a minimal modular extension of $\cC$.
$\cM$ captures most information of the edge states and in particular fixes the chiral central
charge of the edge states modulo 8.

We believe that the only thing missing is the $E_8$ state which has no bulk
topological excitations but non-trivial edge states with chiral central charge
$c=8$. So in addition we add the central charge to complete the classification.
Thus, topological phases with symmetry are classified by
$(\cE\subset\cC\subset\cM,c)$.

We want to emphasize that, the UBFCs $\cE,\cC,\cM$ consist of large sets of
data describing the excitations, and large sets of consistent conditions
between these data. The data and conditions are complete and rigid in the sense that the solutions
are discrete and finite at a fixed rank.

As a first application, we use a subset of data (gauge-invariant physical
observables) and conditions between them to numerically search for possible
topological orders and tabulate them.

We also study the stacking of topological phases with symmetry based on such
classification. We recovered the known classification $H^3(G,U(1))$ for bosonic
SPT phases from a different perspective, via the stacking of modular extensions
of $\cE=\Rep(G)$. Moreover, we predict the classification of invertible
fermionic phases with symmetry, by the
modular extensions of $\cE=\sRep(G^f)$.
We also show that the $\mce{\cE}$ $\cC$ determines the topological phase with
symmetry up to invertible ones.

A special kind of anyon condensation is used in the study of stacking
operations. We then study other kinds of anyon condensations. They allow us to
group topological phases into equivalence classes and simplifies the
classification. More importantly, anyon condensations reveal more
relations between topological phases and correspond to certain topological phase
transitions.

\cleardoublepage


\begin{center}\textbf{Acknowledgements}\end{center}

I would like to thank Xiao-Gang Wen for his supervision and constant support. I
have benefited a lot from his deep physical intuitions, enthusiasm and
encouragement. Also many thanks to my undergraduate supervisor Liang
Kong who has brought me into this exciting research field. He is also a main
collaborator of this work and complemented many mathematical details.

I would like to thank Zheng-Cheng Gu, Pavel Etingof, Dmitri Nikshych, Chenjie Wang and
Zhenghan Wang for helpful discussions. Besides, I would like
to thank my committee members Roger Melko, Davide Giotto and Anton Burkov for
their valuable advice and feedback.

I am also grateful to Perimeter Institute for the wonderful work environment.
\cleardoublepage


\begin{center}\textbf{Dedication}\end{center}

\begin{center}To my parents.\end{center}
\cleardoublepage

\renewcommand\contentsname{Table of Contents}
\tableofcontents
\cleardoublepage
\phantomsection

\addcontentsline{toc}{chapter}{List of Tables}
\listoftables
\cleardoublepage
\phantomsection		


\chapter*{List of Abbreviations}
\addcontentsline{toc}{chapter}{List of Abbreviations}
\begin{description}[labelsep=*,
    align=right,labelwidth=!]
  \item[SPT] Symmetry protected topological phases.
  \item[SET] Symmetry enriched topological phases.
    Topological
    phases with symmetry.
  \item[UBFC] Unitary braided fusion category.
  \item[UMTC] Unitary modular tensor category.
    Non-degenerate unitary braided fusion category.
  \item[$\mce{\cE}$] UMTC over $\cE$. Non-degenerate unitary braided fusion category over $\cE$. 
\end{description}
\chapter*{List of Symbols}
\addcontentsline{toc}{chapter}{List of Symbols}
\begin{description}[labelsep=*,
    align=right,labelwidth=!]
  \item[$\cE$] Symmetric fusion category. $\cE=\Rep(G)$ or $\cE=\sRep(G^f)$.
  \item[$\Rep(G)$] Representation category of bosonic symmetry group $G$. 
  \item[$\sRep(G^f)$] Fermionic variant of $\Rep(G)$.
  \item[$\cC$] $\mce{\cE}$, the main label of topological phases with
    symmetry $\cE$.
  \item[$\cM$] Modular extension of $\cC$.
  \item[$c$] The chiral central charge of the edge state.
  \item[$(\cC,\cM,c)$] The complete label of topological phases with symmetry
    $\cE$.
  \item[$\bt$] Deligne tensor product.
  \item[$\bt_\cE$] Stacking of topological phases.
  \item[$\ot$] Fusion of anyons. Tensor product in UBFCs.
  \item[$\oplus$] Direct sum of anyons.
  \item[$\Z$] Integers.
  \item[$\C$] Complex numbers.
  \item[$\ot_\C$] Tensor product of vector spaces, linear operators or matrices
    over $\C$.
  \item[$N$] Rank of $\cC$. Number of simple anyon types.
  \item[$N^{ij}_k$] Fusion rules.
  \item[$s_i$] Topological spin.
  \item[$d_i$] Quantum dimension.
  \item[$c_{A,B}$] Braiding.
  \item[$X^*$]  Anti-particle of anyon $X$. Dual object of $X$.
  \item[$\ov\cC$] Mirror conjugate of $\cC$.
  \item[$\ov{S_{ij}}$] Complex conjugate of $S_{ij}$.
  \item[$H^n(G,M)$] $n$-th cohomology group of $G$ with coefficients in $M$. 
  \item[$\Hom_\cC(A,B)$] The morphisms from $A$ to $B$ in $\cC$. $\cC$ may be
    omitted.
  \item[$\cen{\cD}{\cC}$] The centralizer of $\cD$ in $\cC$.
  \item[$\mext(\cC)$]  The set of equivalence classes of
modular extensions of $\cC$.
\item[$\Ab{\cC}$] The full subcategory of Abelian anyons in $\cC$ (maximal
  pointed subcategory). 
\end{description}

\cleardoublepage

\pagenumbering{arabic}

\chapter{Introduction}
Gapped quantum phases, or ``insulators'', used to be considered boring, whose
classification was solved by Landau symmetry breaking theory~\cite{Landau37}. However,  over the last few
decades, many exotic phases with ``topological'' nature have be discovered,
which are beyond the scope of symmetry breaking theory.
There are two typical examples of such exotic topological phases:
\begin{itemize}
  \item Fractional quantum Hall (FQH) states. They have anyonic
    quasiparticle excitations with fractional charges and fractional
    statistics. Also their edge states are protected gapless, i.e.,
    not gappable by any boundary interactions. FQH states are
    considered to present intrinsic topological order~\cite{Wen89,Wen90},
    which leads to different phases of matter without requiring any symmetry. 
  \item Topological insulators. Similarly they have gapless edge
    states, which are protected by the symmetry, i.e., not
    gappable if the symmetry is not broken. Topological
    insulator is a special case of another large class of phases
    -- symmetry protected topological (SPT) phases~\cite{CGLW1106.4772}. They have no
    intrinsic topological orders.
\end{itemize}

Topological phases of matter have drawn more and more research interest.
A complete classification has been achieved in 1+1
dimensions~\cite{CGW1103.3323}. It was proved that in 1+1D there is no
topological order, thus the only two ingredients in the classification are
symmetry breaking and SPT. Mathematically, symmetry breaking is described by a
pair of groups $G\in G_H$, where $G_H$ is the symmetry group of the
Hamiltonian, and $G$ is the unbroken subgroup of the ground state; 1+1D SPT
phases are classified by projective representations of the symmetry group, or
by the second cohomology group $H^2(G,U(1))$. The combination of the two
ingredients in 1+1D is obvious.

Intrinsic topological order begins to appear from 2+1D. We need to combine
symmetry and topological order into a unified framework, such that SPT phases
and intrinsic topological orders with no symmetry are just two extreme
special cases. 
This thesis aims at giving such a 
classification of (2+1)D topological phases of matter. The final goal is to
create a ``table'' for all topological phases. There are two main
questions for a good classification:
\begin{enumerate}
  \item How to efficiently describe topological phases? 

    It turns out the universal properties of topological phases, such as
    symmetry, ground state degeneracy, quasiparticle statistics, can be well
    organized into the mathematical framework of unitary braided fusion
    category. Although category theory is quite abstract, obscure and
    unfamiliar to most physicists, it is the right and precise language for
    topological phase. Alternatively, we can use a subset of universal
    properties to name topological phases. We expect to use as few properties
    as possible, as long as they are enough to distinguish different phases.

  \item How to relate different topological phases?

    Apparently, one way to answer the above question is to study phase
    transitions between topological phases. However, a general theory for
    topological phase
    transitions is still beyond our scope.

    A simple construction to relate different phases may be stacking several
    layers of topological phases to obtain a new one. We will discuss such
    stacking operations when symmetry is taken into account.

    We will also discuss phase transitions that are driven by certain anyon
    condensations. Hopefully, the theory of anyon condensations can provide us
    a general framework for topological phase transitions.
\end{enumerate}

\section{Topological Phases and Symmetry}
We take topological phase as a synonym of gapped quantum phase. Consider a
physical system described by the Hilbert space $\cV$ and Hamiltonian $H$.
First, we require that the system admits local structures, namely, the
total Hilbert space is the tensor product of local Hilbert spaces on each
site, $\cV=\ot_i V_i$, and the Hamiltonian is the sum of local terms, $H=\sum_i
H_i$, where $H_i$ acts on only several neighbouring sites around $i$.
Second, we require a finite energy gap $\Delta>0$ under the thermodynamic
limit. States below the gap are the ground states of the system.
If the energy gap closes $\Delta\to 0$ when deforming the Hamiltonian, there is
a \emph{phase transition}. In other words, two systems belong to the same
gapped quantum phase, or topological phase, if they can be deformed into each
other without closing the energy gap. Equivalently, their ground states can be
related by \emph{local unitary transformations}~\cite{CGW1004.3835}.

A physical system may have some global symmetry, described by a group $G$.
In other words, $G$ has a faithful action on the Hilbert space $g\mapsto U_g\in
GL(\cV)$, and the Hamiltonian remains invariant under such action $U_g H
U_g^{-1}=H$. For rigorousness, in this thesis we will restrict to finite
on-site symmetries, which means that $G$ is finite and $U_g=\otimes_i U_{g,i}$,
where $U_{g,i}$ is a local operator acting around the site $i$.

When a symmetry is present, the definition of topological phases must be
modified a little bit. The deformation of Hamiltonian must respect the
symmetry. Topological phases with symmetry are thus equivalence classes under
\emph{symmetric} local unitary transformations.

For fermion systems, we consider the fermion number parity, denoted by $z$, as
a special element in the global symmetry group $G$, which is involutive and
central, $z^2=1,zg=gz, \forall g\in G$. To distinguish with purely bosonic
symmetries, we denote fermionic symmetries by $(G,z)$ or $G^f$ if there is no
ambiguity.\footnote{On the other hand, bosonic symmetry $G$ can be viewed
as $(G,z=1)$.} 

In this thesis, the categorical description of symmetry is more common. Roughly
speaking, we use the representation category, $\cE=\Rep(G)$ for boson systems,
or $\cE=\sRep(G^f)$\footnote{$\sRep(G^f)$ consists of the
  same representations as $\Rep(G)$, but the irreducible representations with $z$ acting as
$-1$ are regarded as fermions. Braiding two fermions has an extra phase factor
$-1$.} for fermion systems,
instead of the group, to describe the symmetry.

There are universal properties, which are invariant under symmetric local unitary
transformations, that can serve as ``labels'' of topological phases. 
In this chapter, topological phases will be denoted by $\cC,\cD,\dots$, which
can be regarded as the collection of universal properties.

\section{Stacking Topological Phases}
We can \emph{stack} two existing topological phases to obtain a
third phase, which is better visualized in (2+1)D by just constructing a
two-layer system. The stacking operation is the easiest way to construct new topological phases
from old ones, and is a main topic of the thesis. 

The most simple case is when there is no symmetry, and we allow any local interactions between layers.
We denote such stacking operation by $\bt$. Obviously, it is commutative and
associative,
$\cC\bt\cD=\cD\bt\cC,(\cC_1\bt\cC_2)\bt\cC_3=\cC_1\bt(\cC_2\bt\cC_3)$. The trivial phase $\cI$ (tensor product states) is the
identity, $\cC\bt\cI=\cI\bt\cC$. Therefore, topological phases form a
commutative monoid
(a ``group'' that requires the existence of identity but not inverse) under
stacking.

When stacking two systems $\cC,\cD$ with the symmetries $G_1,G_2$,
there is a choice for the new symmetry of the two-layer system, that puts
restrictions on what \emph{symmetric} interactions between layers can be added.
One natural choice is $G_1\xt G_2$, denoted by $\cC\bt\cD$, that is to preserve the symmetry of each
layer respectively.

When $G_1=G_2=G$, another natural choice for the new symmetry is $G$, denoted by $\cC\bt_\cE\cD$ (recall
that $\cE$ is the representation category of $G$) where $G$
is viewed as a subgroup of $G\xt G$ via the embedding $g\mapsto (g,g)$. In other
words, for the stacking $\bt_\cE$ we allow the inter-layer interactions that
preserve only the subgroup $G$. $\cC$, $\cD$ and $\cC\bt_\cE\cD$ share the same
symmetry $G$. Therefore, topological phases with symmetry $G$ again form a
commutative
monoid under the stacking $\bt_\cE$ which preserves the symmetry.

A topological phase $\cC$ with symmetry $\cE$ is called invertible if there exists another phase
$\cD$ such that $\cC\bt_\cE\cD=\cI$. In this case $\cC$ and $\cD$ are
time-reversal conjugates. All invertible topological phases form an Abelian
group $\mathbf{Inv}$
under stacking. The chiral central charges of the edge states add up under stacking,
so taking the central charge is a group homomorphism from invertible
phases $\mathbf{Inv}$ to $\Q$. Its image is $c_\text{min}\Z$, where $c_\text{min}$ is the smallest positive
central charge.
From this point of view, the non-chiral invertible phases (the kernel of the above group homomorphism)
are the symmetry protected
topological (SPT) phases:
\begin{align*}
  0\to \text{SPT}\to \mathbf{Inv} \to c_\text{min}\Z\to 0,
\end{align*}
Since $H^2(\Z,M)=0$ for any abelian group $M$, the above must be a trivial
extension, namely
\begin{align*}
  \text{Invertible topological phases with symmetry}\cong\text{SPT}\times
  c_\text{min}\Z,
\end{align*}
 For boson systems, $c_\text{min}=8$ corresponding to the $E_8$
state. For fermion systems with symmetry $G^f=G_b\times \Z_2^f$, $c_\text{min}=1/2$ corresponding to the $p+\ii p$
superconducting state. But for other fermionic symmetries it is not totally clear
what $c_\text{min}$ should be.

Invertible phases do not support any non-trivial quasiparticle statistics.
For the non-invertible topological phases, we have to seriously study
their quasiparticle excitations.
\section{Quasiparticle Excitations}

The properties of excitations play a central role in the study of topological
phases. In (2+1)D, excitations whose energy density is non-zero in a local area
can always be viewed as a particle-like excitation, and will be referred to as
quasiparticles or anyons.

We would like to group quasiparticles into several \emph{topological types} depending on whether
they can be related by local operators. It is easy to think about
the trivial topological type, \ie quasiparticles that can be
created from or annihilated to the ground state, by local operators. We also
call the trivial type \emph{local excitations}. At first glance, one may
wonder, ``are there exotic excitations beyond the local ones?'' The answer is
``Yes'', and the very
existence of non-trivial topological types is exactly the most important signature of
\emph{topological orders}~\cite{Wen89,Wen90}.
We call the non-trivial types \emph{topological excitations}. They can not
be created or annihilated by local operators. Different topological types may
also be referred to as carrying different topological charges or in different
superselection sectors.

When dealing with topological phases with symmetry, we need to regroup
quasiparticles into different types; each type is related by \emph{symmetric}
local operators. In particular, if the degeneracy of a quasiparticle can not
be lifted by any symmetric local operators, for example a local excitation
carrying an irreducible representation of the symmetry group, we say that it is of a
\emph{simple type}. A generic quasiparticle of a \emph{composite type}, for
example a local excitation carrying a reducible representation, is the direct
sum of several simple types. The total number of simple types in a
topological phase with symmetry is called its \emph{rank}, denoted by $N$.

It is natural to ask ``what properties of the excitations are invariant
under (symmetric) local operators''. It turns out these properties can be well
organized in terms of two kinds of processes, fusion and braiding.
Moreover, it seems that the properties of excitations can determine all the
universal properties of topological phases up to invertible ones, even those
non-local properties such as the ground state degeneracies on arbitrary
manifolds.
This
constitutes the main approach of this thesis: classify topological phases by
the fusion and braiding properties of quasiparticles, or mathematically, by
unitary braided fusion categories (UBFC). We find that topological phases with
symmetry are classified by a sequence of UBFCs, $\cE\subset\cC\subset\cM$, plus a central
charge $c$. They correspond to \{local excitations\}$\subset$\{all
excitations\}=\{local excitations plus topological
excitations\}$\subset$\{excitations in the gauged theory\}=\{all excitations plus gauged symmetry defects\}
respectively.

On the other hand, unitary braided fusion category is a very rigid structure.
Physically this means that the fusion of braiding of quasiparticles must
satisfy a series of consistent conditions, such that for a fixed rank, there
are only finite solutions~\cite{BNRW1310.7050}. This allows us to numerically
search for topological phases with symmetries, using a subset of these conditions, and
tabulate them~\cite{Wen1506.05768,LKW1507.04673,LKW1602.05946}.

We will also study the condensation of anyons which relates different
topological phases. This includes two variants:
\begin{enumerate}
  \item Condense a (self-)boson into the trivial state. This will produce a new
    topological phase with the same central charge, and a gapped domain wall
    between the old and the new phases. Symmetry breaking is a special case of
    such condensations.
  \item Condense Abelian anyons into a Laughlin state. This will produce a
    new topological phase with the same symmetry and similar non-Abelian content.
\end{enumerate}

\chapter{Categorical Description of Symmetry and Excitations}
To describe specific physical phenomena, it is important to identify the
corresponding mathematical languages. Newton invented calculus to describe the
gravity and dynamics of classical world. In later developments, physics seems
to fall a
little behind mathematics: linear algebra for quantum mechanics, Riemann
geometry for general relativity, group theory for symmetry, and so on. Now for
topological phases of matter, one of the most important mathematical languages
is the theory of tensor categories.

In this chapter we try to give an introduction to category theory. But, instead
of listing a series of mathematical definitions, we aim at building a mathematics-physics dictionary.

\section{The Categorical Viewpoint}
 Categories have two levels of structures: objects and morphisms.
 Morphisms are the ``arrows'', or operations, on the objects.
 It is not harmful to interpret objects as ``physical objects'', and morphisms
 as ``physical measurements''.
 The most
 important philosophy of category theory is to take objects as black boxes and focus
 on the morphisms, which is like measuring the unknown physical objects.

 We denote the set of morphisms from object $A$ to object $B$ by
 $\Hom(A,B)$. As most valuable intuitions come from the ``trivial''
 example, let's elaborate the idea in the simplest tensor category, the
 category of finite-dimensional vector spaces $\Ve$. In $\Ve$, objects are
vector spaces, while morphisms are linear operators. Now if we pretend to
know nothing about the objects in $\Ve$, how can we recover their vector space
structures via the morphisms, \ie the linear operators? It is an easy observation, that any
vector space $V$ is the same as the vector space of linear operators from
the base field\footnote{In this thesis we always take the base field to be
complex numbers $\C$} to itself, $V\cong \Hom(\C,V)$, by identifying a vector
$|x\rangle$ with the linear map $f(\alpha)=\alpha|x\rangle,\alpha\in\C$. Thus,
linear operators (morphisms) already tell us everything about $\Ve$; the vector space structures of the
objects are redundant information.

Why is this categorical point of view important to physicists? Because it is
exactly the physical point of view. In more general cases, objects may be
interpreted as particles or other physical systems, and morphisms as
evolutions, interactions or other physical operations. We have to admit that
all the physical systems are indeed black boxes to us; we can not know anything
about them unless we try to measure or observe them. All our information about
the physical world comes from the interactions between various systems and
ourselves, \ie from the ``morphisms''. This is just like the $\Ve$ example. We
are like the ``base field $\C$'', and all we know is what we observe, $\Hom(\C,V)$.
\section{Unitary Braided Fusion Category}

Now let us restrict to the case of topological phases of matter and try to give
a specific dictionary. We interpret as follows:
\begin{itemize}
  \item{\textbf{Object}} Quasiparticle excitation, anyon
  \item{\textbf{Morphism}} Operator acting on quasiparticles, but up to symmetric
    local operators
\end{itemize}
If we collect all possible quasiparticle excitations in a topological phase,
they naturally form a category with additional structures, which is a unitary
braided fusion category. Again we interpret the additional structures term by
term:
\begin{itemize}
  \item\textbf{Unitary} \ A structure inherited from the inner product of Hilbert
    spaces (physical measurements). For the morphisms, it means operators have
    Hermitian conjugates.
    As the simplest example, $\Hilb$, the category of Hilbert spaces, is just the
    unitary version of $\Ve$, the category of vector spaces.
  \item\textbf{Fusion} \ Bring several quasiparticles together and view them as one
    quasiparticle. In practice, we need to fuse quasiparticles one by one along
    certain direction; thus fusion is represented by a two-to-one tensor
    product.
     We denote the tensor product by $A\ot B$.
     The trivial anyon (vacuum), denoted by $\one$, is the unit of the tensor
     product, $\one\otimes A\cong A\cong A\otimes \one$. Also it must be
     associative $(A\otimes B)\otimes C\cong A\otimes (B\otimes C)$. This
     associative condition
     only holds up to some local operators, which can
     be represented by the $F$-matrix.
  \item\textbf{Braiding} \ In (2+1)D, the result of tensor product should not depend
    on the direction we choose. In particular, $A\ot B$ should be naturally
    isomorphic to $B\ot A$. The natural isomorphism can be generated by the
    braiding process: move $A$ to the other side of $B$. There are two
    different paths to do this, clockwise or counter-clockwise. We denote the
    clockwise path by $c_{A,B}:  A\ot B \cong B\ot A$, the counter-clockwise path
    is then given by $c_{B,A}^{-1}: A\ot B\cong B\ot A$. Braiding can be
    represented by the $R$-matrix. Detailed definition of
    $F,R$-matrices, also known as $6j$-symbols, $F,R$-symbols, can be found in,
    for example, Refs.~\cite{Kitaev0506438,Wang10,EGNO15}.
  \item \textbf{Simple objects} \ Quasiparticles whose
    degeneracy can not be lifted by local operators, namely of simple types. These simple anyons are
    labeled by lowercase letters $i,j,k,\dots$.
  \item {\textbf{Direct sum}} \ A generic (composite) quasiparticle
    excitation can be the degenerate states of several simple quasiparticles.
    Such degeneracy is ``accidental'' and can be lifted by local operators. We
    can not avoid such cases because the fusion of two (even simple) anyons is in general a
    direct sum of simple ones. This is denoted by $i\otimes j \cong \bigoplus_k
    N^{ij}_k k$, where $N^{ij}_k$ is a non-negative integer describing the
    multiplicity of $k$ in $i\ot j$. A classical example is that two spin 1/2 fuse into the direct sum
    of spin 0 and spin 1: $1/2\otimes 1/2=0\oplus 1$.
  \item {\textbf{Dual objects}} \ Anti-particle $X^*$ of anyon
    $X$. $X\ot X^*$ contains a single copy of the trivial anyon $\one$, $X\ot
    X^*\cong \one\oplus \cdots$. 
  \item {\textbf{Quantum Trace}} \ For an operator $f$ acting on
    $X$, its quantum trace $\Tr f$ is the expectation value of the following
    process up to proper normalization: create a pair $X,X^*$, act $f$ on $X$, and then annihilate the
    pair. We will pick the normalization such that $\Tr\id_\one=1$.
    Then the quantum dimension is $d_i=\Tr \id_i$, and total quantum dimension
    is
    $D^2=\sum_i d_i^2$. The topological \mbox{$S,T$-matrix} are given by the quantum trace
    $T_{ij}=\delta_{ij}\Tr c_{i,i}/d_i $, $S_{ij}=\Tr c_{j^*,i}c_{i,j^*}/D$.
\end{itemize}

\section{Local Excitations and Group Representations}
Next we study the local excitations and show their relations to the symmetry
group $G$.

Let the unitary braided fusion category of all quasiparticle excitations
be $\cC$. Take the category of all local excitations, they form a full subcategory
of $\cC$, denoted by $\cE$.
Local excitations can be created by acting
local operators $O$ on the ground state $|\psi\ket$.  For any group action
$U_g$, $U_g O|\psi\ket=U_g O U_g^\dag U_g|\psi\ket=U_g O U_g^\dag |\psi\ket$ is
an excited state with the same energy as $O|\psi\ket$. Since we assume the
symmetry to be on-site, $U_g OU_g^\dag$ is also a local operator.  Therefore,
$U_g OU_g^\dag|\psi\ket$ and $O|\psi\ket$ correspond to the degenerate local
excitations. Thus, we should group the states $\{U_gO|\psi\ket,\forall g\in G\}$ as a whole and
view them as the same type of excitation. The local Hilbert space spanned by
$\{U_gO|\psi\ket,\forall g\in G\}$, is then a representation of $G$. It is in
general a reducible representation but can be reduced to irreducible ones by
symmetric local operators.\footnote{Symmetric local operators commute with
  group actions $U_g$, thus are the \emph{intertwiners} between
group representations.} We see that local excitations ``locally'' carry group
representations. Simple types of local excitations are
labeled by irreducible representations of the symmetry group $G$. 

As a fusion category (forget the braidings) we have $\cE\cong\Rep(G)$. Further
considering the braidings we find that there are two possibilities, as fermions
braid with each other with an extra $-1$. For boson systems,
$\cE=\Rep(G)$ with the usual braiding for vector spaces, $c_{A,B}(a\ot_\C
b)=b\ot_\C
a$, $a\in A,b\in B$. For fermion
systems $\cE=\sRep(G^f)$, with a modified braiding, $c_{A,B}(a\ot_\C
b)=-b\ot_\C
a$ when $a\in A,b\in B$ are both fermionic (fermion parity $z$ acts as $-1$, $za=-a,zb=-b$).

The most important thing is that the category $\cE$ can recover the symmetry
group $G$ via Tannaka Duality. In other words, the information of morphisms in
$\cE$
alone (the data
of fusion and braiding of local excitations), are enough for us to recognize
the group. An easy example is that for an Abelian group, its irreducible
representations form the same group under tensor product (fusion). For more
general cases it is also true, by Deligne's theorem~\cite{Deligne02}.
These special categories, whose braidings are either bosonic or fermionic, are
named symmetric categories.

Instead of describing the symmetry by a group, it is
equivalent to say finite on-site symmetry is given by a symmetric fusion
category $\cE$. A huge advantage is that in the categorical viewpoint, boson
and fermion systems are treated equally.

\section{Trivial Mutual Statistics, Braiding Non-degeneracy}

Another important property of the local excitations is that they have
\emph{trivial mutual statistics} with all quasiparticles. By trivial mutual
statistics, between quasiparticles $A$ and $B$, we mean that moving $A$
along a whole loop around $B$ makes no difference. In terms of
braidings, trivial mutual statistics means that a double braiding is the same
as the identity, $c_{B,A}c_{A,B}=\id_{A\ot B}$.

It is easy to see the reason why local excitations always have trivial mutual
statistics. Assume that $B$ is a local excitation. Moving $A$ around $B$ is
the same as the following process: first annihilate $B$, second move $A$
around nothing, and then create $B$ again. This is because the operators that
hopping $A$ around, and the local operators annihilating and creating $B$, do
not overlap and thus commute. Moving $A$ around nothing surely makes no
difference, therefore, we know that any quasiparticle $A$ has trivial mutual
statistics with a local excitation $B$.

On the other hand, if a quasiparticle has trivial mutual statistics with all
quasiparticles, we claim that it must be a local excitation. This follows from
the idea of \emph{braiding non-degeneracy}, which means that for an
\emph{anomaly-free}\footnote{An anomaly-free $(n+1)D$ phase can be realized
  by an $n$ dimensional lattice model, without the help of an $n+1$
dimensional bulk.} topological phase, everything non-local must be detectable
via braidings. The subcategory of quasiparticles that have trivial mutual
statistics with all quasiparticles, which is called the M\"uger center~\cite{Muger9812040}, must
coincide with the subcategory of local excitations, which is determined by the
symmetry of the system, a symmetric fusion category $\cE$.
\section{Modular Extensions and the Full Classification}

By now we have outlined all the properties of the quasiparticle excitations in
a topological phase with symmetry $\cE$; they form a unitary braided fusion
category $\cC$ whose M\"uger center coincides with $\cE$. But is this enough to fully
characterize the phase?

Let us return to the invertible phases to check what is lacking.
A topological phase is invertible under stacking if and only if it has only
local excitations. So the question becomes whether $\cC=\cE$ describes all
invertible phases. The answer is of course ``No''.

We know that if there is no
symmetry, $\cC=\cE=\Ve$, invertible topological phases have a $\Z$
classification; they are generated by the $E_8$ state, with central charge
$c=8$, via stacking and
time-reversal. Since the central charge is added up under stacking, it is a good
quantity that complements the label of topological phases.

For boson systems with symmetry $G$, we have $\cC=\cE=\Rep(G)$. But they are classified
by $H^3(G,U(1))\xt 8\Z $, where the $8\Z$ part is the $E_8$ states as explained
above, and the $H^3(G,U(1))$ part is the classification of the {symmetry protected topological}
phases~\cite{CGLW1106.4772}. We need to seek for additional structures to recover such
$H^3(G,U(1))$ classification.

A good classification of invertible fermionic phases with
symmetry, however, is still lacking. Since our categorical viewpoint treats
bosonic and fermionic cases equally, it is promising to
make predictions on the classification of invertible fermionic phases.

So besides the central charge, what should be added to the categorical
description of topological phases with symmetry? Motivated by the idea of
gauging the symmetry~\cite{LG1202.3120}, we propose that the \emph{modular extension} $\cM$ of $\cC$
should be included to fully characterize the phase. When gauging the symmetry,
one adds extrinsic symmetry defects to the system and promotes the global
symmetry to a dynamical gauge group. Extrinsic symmetry defects then become dynamical
gauge flux excitations. Thus, extra quasiparticles (gauged symmetry defects) are added in the gauged theory,
which can detect the local excitations in the original theory via braidings.
We fully gauge the symmetry, such that the gauged theory is a topological phase with no symmetry (the M\"uger
center of $\cM$ is trivial). So, $\cM$ is a unitary \emph{modular} tensor
category (UMTC), and $\cC$ is a full subcategory of $\cM$, hence the name modular extension.
We need to avoid adding something unrelated to the original theory, therefore,
the modular extension is required to be a ``minimal'' one, in the sense that
the extra quasiparticles must have non-trivial mutual statistics with at least
one quasiparticle in $\cE$. In the following all modular extensions are assumed
minimal unless specified otherwise. It is possible that certain $\cC$ has no
modular extension, which means that the symmetry is not
gaugable. We consider it another anomaly-free condition that $\cC$ must have
modular extensions.

To conclude, for topological phases with a given symmetry $\cE$, we propose
that they are classified by the triple $(\cC,\cM,c)$,
where $\cE\subset\cC\subset\cM$. More precisely,
\begin{itemize}
  \item $\cE$ is the symmetric fusion category of local excitations. It is the
    representation category of the symmetry group, $\cE=\Rep(G)$ for boson
    systems and $\cE=\sRep(G^f)$ for fermion systems.
  \item $\cC$ is the unitary braided fusion category of all quasiparticle
    excitations in this phase. Its M\"uger center is $\cE$, equivalently, $\cC$
    is a non-degenerate UBFC over $\cE$, or a $\mce{\cE}$.
  \item $\cM$ is the unitary braider fusion category of quasiparticle
    excitations in the symmetry-fully-gauged theory of this phase. It is a UMTC, and a minimal modular extension of $\cC$.
  \item $c$ is the central charge, from which we can determine the layers of
    $E_8$ states.
\end{itemize}
We will give more strict definitions of the above in Section \ref{math}. 

By now, our proposal has no experimental support yet. Our confidence mainly
comes from the astonishing consistency between the mathematical framework,
known examples, and physical intuitions. As introduced above, topological
phases can be stacked, about which we have good physical intuitions. So it
should also be possible to define a stacking for the triple $(\cC,\cM,c)$.
Later we will seriously study such stacking, in strict mathematical
language. The known classification of invertible bosonic topological phases with
symmetries, obtained earlier by other approaches, is also recovered. 

\chapter{Universal Physical Observables}
We expect that UBFCs contain all the universal properties of quasiparticle
excitations, however, it is usually desirable to use only a subset of universal
properties to ``name'' topological orders. This subset
consists of only the \emph{physical observables} which do not depend on the
choice of basis or gauge. In fact, if we collect all the physical observables,
they should uniquely determine the topological phase.
We first introduce a ``physical representation'' of category theory.

\section{Fusion Space, Fusion Rules and Topological Spin}
Imagine we have a two dimensional manifold $M^2$ with $n$ punctures at
$x_1,x_2,\dots,x_n$. Physically, we put the system on $M^2$, and turn off local
Hamiltonian terms around the punctures $x_a$,\footnote{Set the local term $H_a$
to a constant which equals the ground state energy $\langle H_a\rangle$.}
such that any possible $n$-anyon
configurations becomes degenerate with the ground state. For simplicity, we
assume a boson system with no symmetry. We may further add
local perturbations $\delta H_a$ around $x_a$ to lift the degeneracy. In particular, we can
fix the anyon types $\xi_a$ at $x_a$ one by one, $x_1$ through $x_{n-1}$. The
last (composite) anyon $\xi_n$ at $x_n$ is then fixed by the topology of $M^2$
and the fusion of the other $n-1$ anyons, and we may use local operators around
$x_n$ to reduce $\xi_n$ to a simple type. After all these, we can still have a
degenerate space,
\begin{align*}
  \cV(M^2,\xi_1,\xi_2,\dots,\xi_n),
\end{align*}
which will be called a fusion space. In particular,
\begin{align}
  \cV(S^2,\xi_1,\xi_2,\dots,\xi_n)=\Hom(\one,\xi_1\ot\xi_2\ot\cdots\ot\xi_n).
\end{align}

We then extract the basis independent properties of the fusion spaces. First
consider their dimensions. The degeneracy of a sphere $S^2$ with no quasiparticle, or
with quasiparticle and anti-quasiparticle pair, should be 1. For a torus $T^2$,
as it can be obtained from gluing a cylinder, namely a sphere $S^2$ with two
punctures, we have
\begin{align}
  \cV(T^2)=\bigoplus_i \cV(S^2,i,i^*).
\end{align}
Thus the \emph{rank}, number of simple anyon types, is related to the ground state
degeneracy on the torus\footnote{The ground state degeneracy on the
  torus is always
  the same as the number of simple \emph{topological types}, which is the rank when
  there is no symmetry. But, in the presence of symmetry, since we
  count anyon types differently, there is no simple relation between the rank
and the degeneracy.}
\begin{align}
  N=\dim \cV(T^2).
\end{align}

The \emph{fusion rules} $i\ot j =\oplus_k N^{ij}_k k$ can also be determined in
terms of the dimension of a fusion space,
\begin{align}
  N^{ij}_k=\dim \Hom(k,i\ot j)=\dim \cV(S^2,i,j,k^*).
\end{align}

Another important observable is the \emph{topological spin}, or simply spin,
denoted by $s_i$. If we rotate an anyon $i$ by $2\pi$, the fusion space
$\cV(-,i,\dots)$ will acquire a phase factor $\ee^{2\pi\ii s_i}$ regardless of
the background manifold and anyons. Thus $s_i$ is just (the fractional part of) the
internal angular
momentum of anyon $i$.

In principle we should also explore other observables, while in practice it
seems that $(N^{ij}_k,s_i)$ are enough. They can distinguish all the known
examples of bosonic
topological orders with no symmetry, or UMTCs. When there is symmetry, we
essentially use subcategories of UMTCs. Therefore, by now we use
$(N^{ij}_k,s_i)$ as a short label or name for topological phases.
Alternatively, we can use the topological $S,T$ matrices (see next section).
If new examples come out whose $(N^{ij}_k,s_i)$ conflicts and reveal new
observables that we have missed, we can just append those new observables to
the short
label and make sure topological phases have a unique name.

In practice, the fusion rules $N^{ij}_k$ are still too huge. So we may use
the \emph{quantum dimension} $d_i$ to represent fusion. Physically $d_i$
measures the ``internal degrees of freedom'' of anyon $i$. It can be extracted
from the fusion space dimension via
\begin{align}
  \dim\cV(M^2, \underbrace{i,\dots,i}_{n \text{ copies of } i},\dots)\sim d_i^n,
  \quad n\to \infty,
\end{align}
or
\begin{align}
 \ln d_i= \lim_{n\to\infty} \frac{\ln\dim\cV(M^2,i^{\ot n},\dots)}{n}.
\end{align}
$d_i$ gives a one-dimensional representation of fusion
\begin{align}
  d_i d_j=\sum_k N^{ij}_k d_k.
\end{align}

\section{Conditions for Physical Observables}
In this section we list the conditions for the physical
observables~\cite{BK01,Wang10,RSW0712.1377,Witten89,GK9410089}.
They are only necessary conditions
derived from UBFCs, but in general not sufficient to describe a valid
UBFC. However, they are quite useful in
numerical searches for candidate topological
phases~\cite{Wen1506.05768,LKW1507.04673,LKW1602.05946}. In
Refs.~\cite{LKW1507.04673,LKW1602.05946} with these conditions and
$\cE\subset\cC\subset\cM$ we have successfully generated large tables for
topological phases with symmetries, from which we selected typical ones and
listed in
Appendix~\ref{tables}.
\begin{enumerate} 
\item \emph{Fusion ring}:\\
$N^{ij}_k$ for the UBFC $\cC$ are non-negative integers that satisfy 
\begin{align} 
\label{Ncnd} 
N^{ij}_k&=N_k^{ji}, \ \ N_j^{1i}=\delta_{ij}, \ \ \sum_{k=1}^N N_1^{i k}N_1^{
kj}=\delta_{ij},
\\
  \sum_{m} N_m^{ij}N_l^{m k} &= \sum_{n}
 N_l^{in}N_n^{j k} \text{ or } \sum_m N^{ij}_m N_m = N_i N_ j \text{ or } N_k
 N_i= N_i N_k. 
\nonumber 
\end{align} 
where the matrix $N_i$ is given by $(N_{i})_{ kj} = N^{ij}_k$, and the indices
$i,j,k$ run from 1 to $N$.  In fact $ N_1^{ij}$ defines a charge conjugation
$i\to i^*$:
\begin{align} 
N_1^{ij}=\delta_{ij^*}.  
\end{align}
$N^{ij}_k$ satisfying the above conditions define a fusion ring.

\item \emph{Rational condition}:\\
  $N^{ij}_k$ and $s_i$ for $\cC$ satisfy~\cite{BK01,Vafa88,AM88,Etingof0207007}
\begin{align}
\label{Vcnd}
\sum_r V_{ijkl}^r s_r =0 \text{ mod }1
\end{align} 
where
\begin{align}
\ \ \ \ \ \ 
V_{ijkl}^r &=  
N^{ij}_r N^{kl}_{ r^*}+
N^{il}_r N^{jk}_{ r^*}+
N^{ik}_r N^{jl}_{r^*}
\nonumber\\
&\ \ \ \
- ( \delta_{ir}+ \delta_{jr}+ \delta_{kr}+ \delta_{lr}) \sum_m N^{ij}_m N^{kl}_{
m^*}
\end{align}

\item \emph{Verlinde fusion characters}:\\
The topological $S$-matrix is given by [see eqn. (223) in \Ref{Kitaev0506438}] 
\begin{align} 
\label{SNsss}
S_{ij}&=\frac{1}{D}\sum_k N^{i j}_k \ee^{2\pi\ii(s_i+s_j-s_k)} d_k ,
\end{align}
where the quantum dimension $d_i$ is the largest eigenvalue of the matrix
$N_i$ and ${D=\sqrt{\sum_i d_i^2}}$ is the total quantum dimension.  
Then~\cite{Vafa88}:
\begin{align}
\label{Ver}
  \frac{ S_{il} S_{jl}}{ S_{1l} } =\sum_k N^{ij}_k S_{kl}.
\end{align} 

\item
\emph{Weak modularity}: \\
The topological $T$-matrix is given by
\begin{align} 
T_{ij}&=\delta_{ij}\ee^{2\pi\ii s_i}.
\end{align} 
Then [see eqn. (235) in Ref.~\cite{Kitaev0506438}]
  \begin{align}
    S^\dag T S=\Theta T^\dag S^\dag T^\dag,\ \ \
\Theta={D}^{-1}\sum_{i}\ee^{2\pi\ii s_i} d_i^2.
  \end{align}

\item 
\emph{Charge conjugation symmetry}:
\begin{align}
  S_{ij}=\ov{S_{i j^*}},\ s_i=s_{i^*}, \text{ or } S=S^\dag C,\ \ T=TC,
\end{align}
where the charge conjugation matrix $C$ is given by
$C_{ij}=N_1^{ij}=\delta_{ij^*}$.

\item
Let 
\begin{align}
\nu_i=\frac{1}{D^2} \sum_{jk} N_i^{jk}d_{j}d_{k}\ee^{\ii 4\pi(s_{j}-s_{k})}, 
\end{align}
then $\nu_i \in \Z$ if $i= i^*$~\cite{Bruillard1204.4836}.
\end{enumerate}

For a UMTC, we further have
\begin{enumerate}
  \item $S$ is a unitary matrix. In particular, this means that \eqref{Ver}
    can be rewritten as
    \begin{align}
      N^{ij}_k=\sum_k  \frac{ S_{il} S_{jl}\ov{S_{kl}}}{ S_{1l} },
    \end{align}
    which is the usual Verlinde formula. This way $N^{ij}_k,s_i$ and $S,T$
    determine each other.
    \item $\Theta=\exp(2\pi \ii\frac{c}{8})$, where $c$ is the chiral central
    charge. Thus, the UMTC determines the central charge modulo 8.
\item $\nu_i=0$ if $i\neq  i^*$, and $\nu_i=\pm 1$ if $i
= i^*$~\cite{RSW0712.1377,Wang10}.     
\end{enumerate}

\chapter{Stacking of Topological Phases}\label{stack}
The stacking is probably the most simple construction for topological phases. In
this chapter we discuss the stacking operation in detail.

\section{Stacking in terms of Observables}

We first consider the stacking of boson systems with no symmetry, which can be
easily described by the physical observables, namely $(N^{ij}_k,s_i,c)$.

Suppose that we have two UBFC's, $\cC$ and $\cD$, with simple anyons labeled by
${i\in\cC},\ {a\in\cD}$.
We can construct a new UBFC by simply
stacking $\cC$ and $\cD$, denoted by $\cC\boxtimes\cD$. 
The anyon labels of $\cC\boxtimes\cD$ are pairs $(i,a),i\in\cC,a\in\cD$, and
the observables are given by
\begin{gather}
  (N^{\cC\bt\cD})_{(k,c)}^{(i,a)(j,b)}=(N^\cC)_{k}^{ij} (N^\cD)^{ab}_{c},
\nonumber \\
  s_{(i,a)}^{\cC\bt\cD}=s_{i}^{\cC}+s_{a}^{\cD},\ \ \ c^{\cC\bt\cD}=c^{\cC}+
  c^{{\cD}},
\nonumber \\
    T^{\cC\bt\cD}=T^\cC\otimes_\C  T^{\cD}, 
\nonumber \\ 
S^{\cC\bt\cD}=S^\cC\otimes_\C  S^{\cD}.
\end{gather}
Note that if $\cC,\cD$ has symmetry $G$ ($\cC,\cD$ has M\"uger center $\cE$), $\cC\boxtimes\cD$ has symmetry
$G\times G$ ($\cC\bt\cD$ has M\"uger center $\cE\bt\cE$). Only when $G$ is trivial, i.e., $\cC,\cD$
are UMTCs describing bosonic topological orders with no symmetry, the above is
the stacking that ``preserves symmetry''. 

For topological phases with symmetry, we need to further take a ``quotient'' of
the observables, which corresponds to breaking the symmetry from $G\times G$ to
$G$. But such ``quotient'' is non-trivial at the level
of physical observables, especially for the modular extensions of $\cC,\cD$.
Therefore, from now on, we switch to strict categorical language.

\section{Mathematical Constructions}\label{math}
Below we give the strict mathematical definitions for
$\cE\subset\cC\subset\cM$, and constructions of stacking operations.
Readers who are not interested in the mathematical formulation may jump to the
next section for the main results. For readers who are interested in the
mathematics, a much more detailed discussion can be found in Ref.~\cite{LKW1602.05936}.

We start with the definition of UBFC. Physically it is just a generic and
consistent anyon
model, containing all the information on the fusion and braiding of anyons.
\begin{dfn}
  A unitary braided fusion category $\cC$ is
  \begin{itemize}
    \item $\C$-linear semisimple category:
      \begin{itemize}
	\item Simple objects $i$, $\Hom(i,i)=\C$.
	\item General objects are direct sums of simples, $X\cong\oplus_i
	  \dim\Hom(i,X)\,i$.
      \end{itemize}
    \item $\Hom(X,Y)$ are finite dimensional $\C$-vector spaces for any objects
      $X,Y$.
    \item Finitely many isomorphism classes of simple objects.
    \item Equipped with a monoidal (tensor) structure, $\ot:\cC \xt \cC\to
      \cC$, the associator ${(X\ot Y)\ot Z\simeq
      X\ot(Y\ot Z)}$, the tensor unit $\one$ and unit morphisms, satisfying pentagon and triangle
      equations.
    \item Rigid: every object $X$ has left and right duals $X^*$. (With unitary
      structure left and right duals are automatically isomorphic, so we do not
      distinguish their notations.)
    \item The tensor unit $\one$ is simple $\Hom(\one,\one)\simeq \C$.\\
      (--- The above is the definition of a fusion category)
    \item Equipped with a braiding $c_{X,Y}: X\ot Y\simeq Y\ot X$ satisfying
      the hexagon equations.
    \item Equipped with a unitary structure, an anti-linear ``hermitian conjugate'' map
      $\dag:\Hom(X,Y)\to \Hom(Y,X)$ for all $X,Y$ such that:
      \begin{itemize}
	\item $(gf)^\dag=f^\dag g^\dag, (\lambda f)^\dag=\ov\lambda f^\dag,
	  (f^\dag)^\dag=f$.
	\item $ff^\dag=0$ implies $f=0$.
	\item $\dag$ is compatible with the tensor product and braiding, i.e.,
	  $(f\ot g)^\dag=f^\dag\ot g^\dag$ and the associator, unit morphisms
	  and the braiding are unitary isomorphisms. (An isomorphism $f$ is
	  unitary if $f^\dag=f^{-1}$.)
      \end{itemize}
  \end{itemize}
\end{dfn}

For physical applications, all categories are assumed unitary
in this thesis.

\begin{eg}
  The category of Hilbert spaces, $\Hilb$. It corresponds to the trivial phase
  with no symmetry. It has only local excitations carrying no group
  representations.
\end{eg}
\begin{eg}
  The category of finite dimensional representations of a finite
  group $G$, $\Rep(G)$. It corresponds to the invertible bosonic phase with
  symmetry $G$.
\end{eg}
\begin{eg}
  Given a UBFC $\cC$, Let $\ov\cC$ be its mirror (parity) conjugate, namely
  the same fusion category $\cC$ with reversed
  braiding, and $\cC^{tr}$ be its time-reversal conjugate, corresponding to
    taking the same objects with all morphisms reversed,
    $\Hom_{\cC^{tr}}(X,Y)=\Hom_{\cC}(Y,X)$.
    Then $\ov\cC,\cC^{tr}$ are both also UBFCs.
    We can show that $\ov\cC$ is canonically braided equivalent to
    $\cC^{tr}$ by taking duals (charge conjugation) ${X\mapsto X^*},
    {f\mapsto f^*}$.
    This is the categorical version of $CP=T$.
  \end{eg}

In our classification we have local excitations being a subset, or subcategory, of all bulk
excitations, which is in turn a subset of all excitations in the gauged theory.
In this thesis by a subcategory or a subset of excitations, we always mean a \emph{full} subcategory in the following
sense
\begin{dfn}
  A full subcategory $\cB$ of $\cC$ means we take a subset of object
  $\ob(\cB)\subset \ob(\cC)$ but all the morphisms
  $\Hom_\cB(X,Y)=\Hom_\cC(X,Y)$.
  A fusion subcategory is a full subcategory that is a fusion category itself. 
\end{dfn}
\begin{eg}
  All UBFCs have a trivial fusion subcategory consists of multiples of the
  tensor unit $\one$, which is equivalent to $\Hilb$. Physically, this means
  that anyon models always contain all the trivial anyons.
\end{eg}
In a UBFC $\cC$, Frobenius-Perron dimension (defined by the largest eigenvalue
of the fusion matrix $N_i$) and quantum dimension (defined by the quantum
trace of the identity morphism $\id_i$) of object $i$
coincide, which is physically the ``number of internal degrees of freedom''
of anyon $i$, and will be denoted by $\dim(i)=d_i$. For a general object $X$, $\dim(X)=\sum_i \dim\Hom(i,X)\dim(i)$.
The total quantum dimension of a unitary fusion category is $\dim(\cC)=\sum_i
\dim(i)^2$, ($i,j,k\dots$ labels runs
over isomorphism classes of simple objects). We rely on the following lemma to
identify fusion categories.
\begin{lem}[EO~\cite{EO0301027}]\label{EOlem}
  If $\cB$ is a fusion subcategory of $\cC$ (or there is a fully faithful
  tensor embedding $\cB\hookrightarrow \cC$) then
  $\dim(\cB)\leq\dim(\cC)$, and the equality holds iff $\cB=\cC$.
\end{lem}

Below is the mathematical description of ``trivial mutual statistics''.
\begin{dfn}
  The objects $X,Y$ in a UBFC $\cC$ are said to \emph{centralize} each other if
  \begin{align*}
   c_{Y,X} c_{X,Y}=\mathrm{id}_{X\otimes Y},
  \end{align*}
  where $c_{X,Y}: X\otimes Y\cong Y\otimes X$ is the braiding in $\cC$.
  Equivalently, $i,j$ centralize each other if $S_{ij}=d_id_j/D$.
  Given a fusion subcategory $\cD\subset \cC$, its \emph{centralizer}
  $\cen{\cD}{\cC}$ in $\cC$
  is the full subcategory of objects in $\cC$ that centralize all the objects in
  $\cD$. The centralizer is a fusion subcategory. In particular,
  $\cen{\cC}{\cC}$ is called the M\"uger center of $\cC$.
\end{dfn}

\begin{dfn}
  A UBFC $\cC$ is a unitary modular tensor category (UMTC) if
  $\cen{\cC}{\cC}=\Hilb$.
\end{dfn}
\begin{lem}[DGNO~\cite{DGNO0906.0620}]
  Let $\cD$ be a
  fusion
  subcategory of a
  UMTC $\cC$, then
  $$\cen{(\cen{\cD}{\cC})}{\cC}=\cD,
  \quad\dim(\cD)\dim(\cen{\cD}{\cC})=\dim(\cC).$$
\end{lem}
\begin{dfn}
  A UBFC $\cE$ is a \emph{symmetric} fusion category if $\cen{\cE}{\cE}=\cE$.
\end{dfn}
UMTC and symmetric fusion category correspond to two extreme cases, 
\ie braiding is non-degenerate and maximally degenerate, respectively.
Symmetric fusion categories are closely related to bosonic and fermionic
symmetry groups, according to the following theorem
\begin{thm}[{Deligne}~\cite{Deligne02}]
  A symmetric fusion category is braided equivalent to $\Rep(G,z)$, where $G$
  is a finite group, and $z\in G$ is a central element such that $z^2=1$, and
  $\Rep(G,z)$ is the fusion category $\Rep(G)$ equipped with braiding $c^z$: 
  $$
  c^z_{X,Y}(x\otimes_\C y) = (-1)^{mn} y\otimes_\C x, \quad \forall x\in X, y\in Y, \quad zx=(-1)^mx, zy=(-1)^ny. 
  $$
\end{thm}
When $z=1$ it is $\Rep(G)$ with the usual braiding $x\ot_\C y\to y\ot_\C x$.
When $z\neq 1$ it is the fermion number parity. Fermions braid with each other
with an extra $-1$. We introduce $\sRep(G^f)=\Rep(G,z)$ for $z\neq 1$ to emphasize
its fermionic nature.
\begin{eg}
  $\sRep(\Z_2^f)$ is the category of super Hilbert spaces, $\sHilb$, that
  is, $\Z_2$-graded Hilbert spaces with $\Z_2$-graded braiding. It corresponds
  to invertible fermionic
  phases with no other symmetries. 
\end{eg}

In the following we keep $\cE$ as a (fixed) symmetric fusion category.
We give the strict definition for $\cE\subset\cC\subset\cM$.

\begin{dfn}
  A pair $(\cC,\iota)$, a UBFC $\cC$ with a fully faithful embedding $\iota:\cE\inj
  \cen{\cC}{\cC}$ is a UBFC \emph{over} $\cE$. Moreover, $\cC$ is said a non-degenerate UBFC
  over $\cE$, or $\mce{\cE}$, if $\cen{\cC}{\cC}=\cE$.
  Two UBFCs over $\cE$, $(\cC_1,\iota_1)$ and $(\cC_2,\iota_2)$ are equivalent if
  there is a braided monoidal equivalence $F:\cC_1\to\cC_2$ such that
  $F\iota_1=\iota_2$.
\end{dfn}
We recover the usual definition of UMTC when $\cE$ is trivial. In this case the
subscript is omitted.

\begin{dfn}
  Given a $\mce{\cE}$ $\cC$, its (minimal) \emph{modular extension} is a pair
  $(\cM,\iota_\cM)$, a UMTC
  $\cM$, together with a fully faithful embedding
  $\iota_\cM:\cC\hookrightarrow\cM$,
  such that $\cen{\cE}{\cM}=\cC$.
  Two modular extensions $({\cM_1},\iota_{\cM_1}),(\cM_2,\iota_{\cM_2})$ are equivalent if
  there is a braided monoidal equivalence $F:{\cM_1}\to\cM_2$ such that
  $F\iota_{\cM_1}=\iota_{\cM_2}$.
  We denote the set of equivalence classes of modular extensions of
  $\cC$ by $\mext(\cC)$.
\end{dfn}
\begin{rmk}
  Here the condition $\cen{\cE}{\cM}=\cC$ is equivalent to
  $\dim(\cM)=\dim(\cC)\dim(\cE)$, or $\cen{\cC}{\cM}=\cE$. Physically this
  means that the extra excitations in $\cM$ but not in $\cC$ all have
  non-trivial mutual statistics with at least one excitation in $\cE$.  
  \label{constructmtce}
  In fact, let $\cM$ be a UMTC that contains a symmetric fusion category $\cE$ as a
  full subcategory,
  and $\cD=\cen{\cE}{\cM}$. Then, $\cE$ is a full subcategory of $\cD$ ($\cE$
  centralizes itself) and
  $\cen{\cD}{\cM}=\cen{(\cen{\cE}\cM)}{\cM}=\cE$. We see that
  $\cen{\cD}{\cD}=\cD\cap(\cen{\cD}{\cM})=\cD\cap\cE=\cE$. This means that
  $\cD=\cen{\cE}{\cM}$ is automatically a $\mce{\cE}$, and $\cM$ is its
   modular extension. This will be a useful way
   to construct $\mce{\cE}$'s from UMTCs.
\end{rmk}

\begin{rmk}
  For a given $\mce{\cE}$ $\cC$, it is possible that there is no modular
  extension of $\cC$.
  An example was constructed by Drinfeld~\cite{Drinfeld}. It is a
  $\mce{\Rep(\Z_2\xt\Z_2)}$ with rank $N=5$ and $D^2=8$. The same example is
  also discussed in Ref.~\cite{CBVF1403.6491}.
\end{rmk}

It is important to note that counting modular extensions of a fixed $\cC$ is
different from counting topological phases.
\begin{dfn}
  Two topological phases with symmetry $\cE$,
  labeled by $((\cC_1,\iota_1),({\cM_1},\iota_{\cM_1}),c_1)$ and
  $((\cC_2,\iota_2),({\cM_2},\iota_{\cM_2}),c_2)$, are equivalent if $c_1=c_2$ and there
  are braided monoidal equivalences $F_\cC:\cC_1\to\cC_2$, $F_\cM:{\cM_1}\to{\cM_2}$
  such that $F_\cC\iota_1=\iota_2$, $\iota_{\cM_2} F_\cC=F_{\cM} \iota_{\cM_1}$.
\end{dfn}
Physically, when counting topological phases, we allow ``relabelling'' anyons in
$\cC$ and $\cM$ together in a compatible way. But we do not allow mixing
``excitations'' (anyons in $\cC$) with ``gauged symmetry defects'' (anyons
not in $\cC$). Also we do not allow
``relabelling'' local excitations in $\cE$, as they are related to the
symmetry group which has absolute meaning. For example spin-flip $\Z_2$ can not be
considered as the same as layer-exchange $\Z_2$, nor can their
representations be relabelled. On the other hand, when counting modular
extensions, we fix all the excitations in $\cC$ and only allow ``relabelling'' ``gauged symmetry defects'' (anyons
in $\cM$ but not in $\cC$). 

  The embeddings $\iota,\iota_\cM$ are important data. However, in the
  following constructions, the embeddings are naturally defined, as we
  construct $\cE,\cC$ as full subcategories of $\cM$. So we may omit the
  embeddings to simplify notations whenever there is no ambiguity.

Next we give the construction for the stacking of topological phases. First
consider the stacking
operation corresponding to the no-symmetry case. It is given by the Deligne
tensor product $\bt$, which defines a monoidal structure on the 2-category of
UBFCs (more generally, of Abelian categories).
For two UMTCs $\cC,\cD$, $\cC\bt\cD$ is still a UMTC.
(By construction, $\Hom_{\cC\bt\cD}(A\bt B,X\bt Y)=\Hom_\cC(A,X)\ot
\Hom_\cD(B,Y)$. All the structures follows component-wise.) There is a parallel story
for $\mce{\cE}$, a monoidal structure $\bt_\cE$ such that the ``stacking'' of
two $\mce{\cE}$s is still a $\mce{\cE}$. We introduce this construction and
generalize it to modular extensions. Such stacking
operation is for $\mce{\cE}$ together with their modular extensions, thus
physically the stacking operations for topological phases with symmetry $\cE$.

The basic idea is to first construct $\cC\bt\cD$ which has symmetry
$\cE\bt\cE$, and then break the symmetry down to $\cE$. We need to first
introduce the following important concept, which controls generic Bose
condensations in topological phases. Symmetry breaking $\cE\bt\cE\to\cE$ is
just a special case.
\begin{dfn}\label{alg}
  A \emph{condensable algebra} in a UBFC $\cC$ is a
  triple $(A,m,\eta)$, $A\in\cC$,
  $m:A\ot A\to A$, $\eta:\one\to A$ satisfying
  \begin{itemize}
    \item Associative: $m(\id_A\ot m)=m(m\ot \id_A)$
    \item Unit: $m(\eta\ot\id_A)=m(\id_A\ot\eta)=\id_A$
    \item Isometric: $m m^\dag=\id_A$
    \item Connected: $\Hom(\one,A)=\C$
    \item Commutative: $m c_{A,A}=m$
  \end{itemize}
\end{dfn}
\begin{rmk}
  This is an important notion that is widely studied. In the subfactor context
  it is called (irreducible local) \emph{$Q$-system}~\cite{LR9604008}. In
  category literature it is also known as
  connected \'etale algebra (connected commutative separable
  algebra)~\cite{DGNO0906.0620,DMNO1009.2117}, or commutative
  special symmetric $C^*$-Frobenius algebra~\cite{Muger0111204,FFRS0309465}. The latter two
  are more general; they do not require the category to be unitary. In the
  unitary case, they are equivalent notions~\cite{LR9604008}. We follow
  Ref.~\cite{Kong1307.8244} to call ``condensable algebra'' for its physical meaning
  and also simplicity.
\end{rmk}
\begin{dfn}
  A (left) \emph{module} over a condensable algebra $(A,m,\eta)$ in $\cC$ is a
  pair $(X,\rho)$, $X\in\cC$, $\rho:A\ot X\to X$ satisfying
  \begin{gather}
    \rho(\id_A\ot\rho)=\rho(m\ot \id_M),\nonumber\\
    \rho(\eta\ot\id_M)=\id_M.
  \end{gather}
  It is further a \emph{local} module if
  \begin{align*}
    \rho c_{M,A} c_{A,M}=\rho.
  \end{align*}
\end{dfn}
  We denote the category of left $A$-modules by $\cC_A$.
  A left module $(X,\rho)$ is turned into a right module via the braiding,
  $(X,\rho c_{X,A})$ or $(X,\rho c_{A,X}^{-1})$, and thus a $A$-$A$-bimodule.
  The relative tensor functor $\ot_A$ of bimodules then turns $\cC_A$ into a fusion category.
  (This is known as $\alpha$-induction in subfactor context.)
  In general there can be two monoidal structures on $\cC_A$, since there are
  two ways to turn a left module into a bimodule (usually we pick one for
  definiteness
  when considering $\cC_A$ as a fusion category).
  The two monoidal structures coincide for the fusion subcategory $\cC_A^0$ of
  local $A$-modules. Moreover, $\cC_A^0$ inherited the braiding from $\cC$ and
  is also a UBFC; it exactly describes the excitations in the topological phase
  after condensing $A$.
\begin{lem}[DMNO~\cite{DMNO1009.2117}]
  \[\dim(\cC_A)=\frac{\dim(\cC)}{\dim(A)}.\]
  If $\cC$ is a UMTC, then so is $\cC_A^0$, and
  \[\dim(\cC_A^0)=\frac{\dim(\cC)}{\dim(A)^2}.\]
\end{lem}

Below we construct a canonical condensable algebra $L_\cC$ in $\cC\bt\ov\cC$
for any UBFC $\cC$. In particular, $L_\cE$ is the algebra that corresponds to
the symmetry breaking $\cE\bt\cE\to\cE$.
\begin{dfn}
  The Drinfeld center $Z(\cA)$ of a monoidal category $\cA$ is a braided monoidal category with
  objects as pairs $(X\in\cA,b_{X,-})$, where $b_{X,-}: X\ot -\to -\ot X$ are
  half-braidings that satisfy similar conditions as braidings. Morphisms are
  those that commute with half-braidings. The tensor product is given by
  \[(X,b_{X,-})\ot(Y,b_{Y,-})=(X\ot Y,(b_{X,-}\ot\id_Y)(\id_X\ot b_{Y,-})),\] and
  the braiding is $c_{(X,b_{X,-}),(Y,b_{Y,-})}=b_{X,Y}$.
\end{dfn}
  It is known that $Z(\cA)$ is a UMTC if $\cA$ is a unitary fusion category~\cite{Muger0111204}.
  There is a forgetful tensor functor
  $\text{for}_\cA:Z(\cA)\to \cA$, $(X,b_{X,-})\mapsto X$ that forgets the
  half-braidings.
  Let $\cC$ be a braided fusion category and $\cA$ a fusion category, a tensor
  functor $F:\cC\to \cA$ is called a central functor if it factorizes through
  $Z(\cA)$, i.e., there exists a braided tensor functor $F':\cC\to Z(\cA)$ such
  that $F=\text{for}_\cA F'$.

\begin{lem}
  [DMNO~\cite{DMNO1009.2117}]
  Let $F:\cC\to\cA$ be a central functor, and $R:\cA\to\cC$ the right adjoint
functor of $F$.
Then the object $A=R(\one) \in\cC$ has a canonical structure of condensable algebra.
$\cC_A$ is monoidally equivalent to the image of
$F$, i.e. the smallest fusion subcategory of $\cA$ containing $F(\cC)$.
\end{lem}

  If $\cC$ is a UBFC, it is naturally embedded into
  $Z(\cC)$, by taking $X\mapsto (X,b_{X,-}=c_{X,-})$. So is $\ov\cC$. Therefore, we have a
  braided tensor functor $\cC\bt\ov\cC\to Z(\cC)$.
  Compose it with the forgetful functor $\text{for}_\cC:Z(\cC)\to\cC$ we
  get a central functor
\begin{align*}
  \ot : \cC\bt\ov\cC &\to \cC\\
  X\bt Y&\mapsto X\ot Y.
\end{align*}
Let $R$ be its right adjoint functor, we obtain a condensable algebra
$L_\cC:=R(\one)\cong { \oplus_i ( i\bt i^*)} \in \cC\bt\ov\cC$ and $\cC=(
\cC\bt\ov\cC)_{L_\cC}$, $\dim(L_\cC)=\dim(\cC)$.
In particular, for a symmetric category $\cE$, $L_\cE$ is a condensable algebra
in $\cE\bt\cE$, and $\cE=(\cE\bt\cE)_{L_\cE}=(\cE\bt\cE)_{L_\cE}^0$ for $\cE$
is symmetric, all $L_\cE$-modules are local.

Now, we are ready to define the stacking operation for $\mce{\cE}$'s as well
as their  modular extensions.
\begin{dfn}\label{stacking}
  Let $\cC,\cD$ be $\mce{\cE}$'s, and $\cM_\cC,\cM_\cD$ their 
  modular extensions. The stacking is defined by:
  \begin{align*}
    \cC\bt_\cE\cD:=(\cC\bt\cD)^0_{L_\cE},\quad \cM_\cC\bt_\cE \cM_\cD:=(\cM_\cC\bt
    \cM_\cD)_{L_\cE}^0
  \end{align*}
\end{dfn}

\begin{thm}
  $\cC\bt_\cE\cD$ is a $\mce{\cE}$, and
  $\cM_\cC\bt_\cE\cM_\cD$ is a  modular extension of $\cC\bt_\cE\cD$.
\end{thm}
\begin{proof}
  The embeddings
$\cE=(\cE\bt\cE)_{L_\cE}^0\hookrightarrow (\cC\bt\cD)^0_{L_\cE}=\cC\bt_\cE\cD
\hookrightarrow\cen{\cE}{\cM_\cC\bt_\cE\cM_\cD} 
\hookrightarrow (\cM_\cC\bt\cM_\cD)^0_{L_\cE}=\cM_\cC\bt_\cE\cM_\cD$
are obvious.
So $\cC\bt_\cE\cD$ is a UBFC over $\cE$. Also
\begin{align*}
  \dim(\cC\bt_\cE\cD)=\frac{\dim(\cC\bt\cD)}{\dim(L_\cE)}
  =\frac{\dim(\cC)\dim(\cD)}{\dim(\cE)},
\end{align*}
and $\cM_\cC\bt_\cE\cM_\cD$ is a UMTC,
\begin{align*}
  \dim(\cM_\cC\bt_\cE\cM_\cD)
  =\frac{\dim(\cM_\cC\bt\cM_\cD)}{\dim(L_\cE)^2}=\dim(\cC)\dim(\cD).
\end{align*}
Therefore, $\cC\bt_\cE\cD$ and $\cen{\cE}{\cM_\cC\bt_\cE\cM_\cD}$ have the same total
quantum dimension, thus by Lemma~\ref{EOlem} we know that they are the same.
  By Remark~\ref{constructmtce},
  $\cC\bt_\cE\cD$ is a $\mce{\cE}$, and
  $\cM_\cC\bt_\cE\cM_\cD$ is a
  modular extension of $\cC\bt_\cE\cD$.
\end{proof}

Note that $\cC\bt_\cE\cE=\cC$. Therefore,  for any
  modular extension $\cM_\cE$ of $\cE$, $\cM_\cC\bt_\cE\cM_\cE$ is still a
  modular extension of $\cC$. Physically this means that stacking with an
  invertible phase will not change the bulk excitations. In the following we want to show the inverse,
  that one can extract the ``difference'', a modular extension of $\cE$, or an
  invertible phase, between two modular extensions of $\cC$.

\begin{lem}\label{Lag}
  We have $(\cC\bt\ov\cC)_{L_\cC}^0=\cen{\cC}{\cC}$.
\end{lem}
\begin{proof}
  $(\cC\bt\ov\cC)_{L_\cC}$ is equivalent to $\cC$ (as a fusion
category). Moreover, for $X\in\cC$ the equivalence gives the free module $
L_\cC\ot(X\bt
\one )\cong L_\cC\ot(\one\bt X)$. $L_\cC\ot(X\bt\one )$ is a local $L_\cC$ module if and only
if $X\bt \one$ centralize $L_\cC$. This is the same as $X\in
\cen{\cC}{\cC}$. Therefore, we have $(\cC\bt \ov\cC)_{L_\cC}^0=\cen{\cC}{\cC}$.
\end{proof}
\begin{lem}[FFRS~\cite{FFRS0309465}]
  For a non-commutative algebra $A$, we have the left center
  $A_l$ of $A$, with algebra embedding $e_l:A_l\to A$, which is the maximal subalgebra
  such that $m {(\id_A\ot e_l)} c_{A_l,A}=m(e_l\ot \id_A)$. Similarly the right
  center $A_r$ with $e_r:A_r\to A$, is the maximal subalgebra such that $m
  ( e_r\ot\id_A) c_{A,A_r}=m(\id_A\ot e_r)$. $A_l$ and
  $A_r$ are commutative subalgebras, thus condensable.
  There is a canonical equivalence between the categories of local modules
  over the left and right centers, $\cC_{A_l}^0=\cC_{A_r}^0$.
\end{lem}

\begin{thm}\label{main}
 let $\cM$ and $\cM'$ be two modular extensions of the $\mce{\cE}$ $\cC$. There
 exists a unique $\cK\in\mext(\cE)$ such that $\cK\bt_\cE\cM=\cM'$. Such $\cK$
 is given by
   \begin{align*}
     \cK=(\cM'\bt \ov\cM)_{L_\cC}^0.
   \end{align*}
\end{thm}
\begin{proof}
  $\cK$ is a modular extension of $\cE$. This follows
  Lemma \ref{Lag}, that $\cE=\cen{\cC}{\cC}=(\cC\bt\ov \cC)^0_{L_\cC}$ is a full
  subcategory of $\cK$. $\cK$ is a UMTC by construction, and
  $\dim(\cK)=\frac{\dim(\cM)\dim(\cM')}{\dim(L_\cC)^2}=\dim(\cE)^2$.

  To show that $\cK=(\cM'\bt\ov\cM)_{L_\cC}$ satisfies
  $\cM'=\cK\bt_\cE\cM$, note that
  $\cM'=\cM'\bt\Hilb=\cM'\bt(\ov\cM\bt\cM)_{L_{\ov\cM}}^0$. It suffices that
  \begin{align*}
    (\cM'\bt\ov\cM\bt\cM)_{\one\bt
      L_{\ov\cM}}^0=[(\cM'\bt\ov\cM)_{L_\cC}^0\bt
      \cM]_{L_\cE}^0=(\cM'\bt\ov\cM\bt\cM)_{(L_\cC\bt\one)\ot(\one\bt
      L_\cE)}^0.
  \end{align*}
  While $\one\bt L_{\ov\cM} $ and $(L_\cC\bt\one)\ot(\one\bt
  L_\cE)$ turns out to be left and right centers of the algebra $(L_\cC\bt\one)\ot(\one\bt
  L_{\ov\cM})$.

  If $\cM'=\cK\bt_\cE\cM=(\cK\bt\cM)_{L_\cE}^0$,
  then
  \begin{align*}
    \cK= (\cK\bt\cM\bt\ov\cM)_{\one\bt
    L_\cM}^0=
    (\cK\bt\cM\bt\ov\cM)_{(L_\cE\bt\one)\ot(\one\bt 
    L_\cC)}^0=[(\cK\bt_\cE\cM)\bt\ov\cM]_{L_\cC}^0
    =(\cM'\bt\ov\cM)_{L_\cC}^0.
  \end{align*}
  It is similar here that $\one\bt L_{\cM} $ and
  $(L_\cE\bt\one)\ot(\one\bt
  L_\cC)$ are the left and right centers of the algebra
  $(L_\cE\bt\one)\ot(\one\bt
  L_{\cM})$. This proves the uniqueness of $\cK$.

  The above established the equivalences between UMTCs. To further show that they
  are equivalences between modular extensions, one need to check the
  embeddings of $\cE,\cC$.
Here the only non-trivial braided tensor equivalences are those between the categories of
local modules over left and right centers. By the detailed construction given
in Ref.~\cite{FFRS0309465}, one can check that they indeed preserve the embeddings of
$\cE,\cC$.
\end{proof}

Let us list several consequences of Theorem \ref{main}.
\begin{cor}\label{hegroup}
  $\mext(\cE)$ forms a finite Abelian group. The identity is $Z(\cE)$ and the inverse of
  $\cM$ is $\ov\cM$.
 \end{cor}
 \begin{proof}
   It is easy to verify that the stacking $\bt_\cE$ for modular extensions
   is associative and commutative. To show that they form a group we only need
   to find out the identity and inverse.
   In this case $\cK=(\cM'\bt \ov \cM)^0_{L_\cE}=\cM'\bt_\cE\ov \cM$,
   Theorem \ref{main} becomes $\cM'\bt_\cE\ov\cM\bt_\cE\cM=\cM'$, for any
   modular extensions $\cM,\cM'$ of $\cE$.
   Thus, $\ov{\cM'}\bt_\cE
   \cM'=\ov{\cM'}\bt_\cE
   \cM'\bt_\cE\ov\cM\bt_\cE\cM
   =\ov\cM\bt_\cE\cM$, i.e. $\cZ_\cE:=\ov\cM\bt_\cE\cM$ is the same category
   for any extension $\cM$, which is exactly the identity element. It is then
   obvious that the inverse of $\cM$ is $\ov\cM$. The finiteness follows from
  Ref.~\cite{BNRW1310.7050}.

   In fact, the identity $\cZ_\cE$ should be $Z(\cE)$, the Drinfeld center of
   $\cE$. (This is Theorem~\ref{cze}. The
   embedding $\cE\hookrightarrow Z(\cE)$ is given by the lift of the identity
   functor on $\cE$, i.e., $\cE\hookrightarrow Z(\cE)\to \cE$ equals $\id_\cE$.)
 \end{proof}

  \begin{eg}
    $\mext(\sRep(\Z_2^f))\cong\Z_{16}$, with central charge
    $c=n/2\text{ mod } 8,n=0,1,2,\dots,15$. This is the 16-fold way~\cite{Kitaev0506438}.
  \end{eg}
  \begin{eg}[LKW~\cite{LKW1602.05936}]
    $\mext(\Rep(G))\cong H^3(G,U(1))$, all with central charge
    $c=0\mod 8$. This agrees with the classification of bosonic
    SPT phases~\cite{CGLW1106.4772}.
  \end{eg}

 \begin{cor}\label{hetorsor}
   For a $\mce{\cE}$ $\cC$, $\mext(\cC)$, if exists, forms a
   $\mext(\cE)$-torsor. The action of $\mext(\cE)$ on $\mext(\cC)$ is given by
   the stacking $\bt_\cE$.
 \end{cor}

Below is a standalone theorem that fixes the unit element in the Abelian group
of modular extensions. 
\begin{thm}\label{cze}
  Let $\cM$ be a modular extension of a $\mce{\cE}$ $\cC$:
  \begin{align*}
    (\ov\cM\bt\cM)_{L_\cC}^0=Z(\cE).
  \end{align*}
  In particular, this means that $\cZ_\cE=Z(\cE)$.
\end{thm}
\begin{proof}
  There is a Lagrangian algebra $L_\cM$ in $\ov\cM\bt\cM$, such that the
  category of $L_\cM$-modules in $\ov\cM\bt\cM$ is
  $(\ov\cM\bt\cM)_{L_\cM}=\cM$, via the functor $L_\cM\ot(i\bt \one)\mapsto i$. $L_\cM$ is a condensable algebra over
$L_\cC$, and also a condensable algebra in
$(\ov\cM\bt\cM)_{L_\cC}^0$. We would like to show that
$[(\ov\cM\bt\cM)_{L_\cC}^0]_{L_\cM}=\cE$. To see this, note that
$\cE\hookrightarrow (\ov\cM\bt\cM)_{L_\cC}^0$, the image of $\cE$ identifies
with the free $L_\cC$-modules $ L_\cC\ot(i\bt \one) \cong  L_\cC\ot(\one\bt i),
i\in\cE$. Further check the free $L_\cM$-modules in
$(\ov\cM\bt\cM)_{L_\cC}^0$ generated by these objects, and we
find that $L_\cM\ot_{L_\cC} [L_\cC\ot (i\bt \one)]\cong
L_\cM\ot(i\bt\one)\mapsto i  $.
This means that $\cE\subset[(\ov\cM\bt\cM)_{L_\cC}^0]_{L_\cM}$. Since they
have the same total quantum dimension, we must
have $[(\ov\cM\bt\cM)_{L_\cC}^0]_{L_\cM}=\cE$. Since $L_\cM$ is Lagrangian in
$(\ov\cM\bt\cM)_{L_\cC}^0$,
$(\ov\cM\bt\cM)_{L_\cC}^0=Z([(\ov\cM\bt\cM)_{L_\cC}^0]_{L_\cM})=Z(\cE)$.
Moreover, $L_\cM\ot_{L_\cC}-:(\ov
\cM\bt\cM)_{L_\cC}^0\to[(\ov\cM\bt\cM)_{L_\cC}^0]_{L_\cM}$ coincides with the
forgetful functor $Z(\cE)\to\cE$. Thus the embedding $\cE\hookrightarrow
(\ov\cM\bt\cM)_{L_\cC}^0$ composed with the forgetful functor $Z(\cE)\to\cE$
gives the identity functor on $\cE$.
\end{proof}

\section{Main Results of the Stacking Operation}
We conclude the main results in the previous section. Topological phase with
symmetry $\cE$ are classified by the triple $(\cC,\cM,c)$. We
mathematically constructed the stacking operation between them,
\begin{align}
  (\cC_1,\cM_1,c_1)\bt_\cE(\cC_2,\cM_2,c_2)=(\cC_1\bt_\cE\cC_2,\cM_1\bt_\cE\cM_2,c_1+c_2).
\end{align}

In particular, the trivial phase with symmetry $\cE$ is given by
$(\cE,Z(\cE),c=0)$, and invertible topological phases with symmetry $\cE$ are described
by $(\cE,\cM,c)$, where $\cM$ is a modular extension of $\cE$,
$\cM\in\mext(\cE)$. They indeed form a Abelian group under the stacking operation defined
above. For boson systems, $\cE=\Rep(G)$, $\mext(\Rep(G))\cong H^3(G,U(1))$,
and they all have central charge $c=0\mod 8$. The
group structure $H^3(G,U(1))\times
8\Z$ is recovered. For fermion systems, we expect that $\mext(\sRep(G^f))$ gives
a full classification of invertible phases. We can obtain both the
fermionic SPT, namely the $c=0$ part in $\mext(\sRep(G^f))$, and the smallest
positive central charge $c_\text{min}$ of the chiral invertible phases. Thus,
invertible topological phases with symmetry $\cE$ are classified by
\begin{align}
  \text{SPT}\xt c_\text{min}\Z,\quad \text{SPT}\xt
  c_\text{min}\Z/8\Z\cong\mext(\cE).
\end{align}
By now we do not have a general formula for $\mext(\sRep(G^f))$, so we do
not know $c_\text{min}$ for generic $G^f$. Also we have checked the form
$\text{SPT}\xt
  c_\text{min}\Z/8\Z\cong\mext(\sRep(G^f))$ only for $G^f=G_b\times \Z_2^f$,
  or small $G^f$ that is not of the form $G_b\times \Z_2^f$, but
  not for generic $G^f$; it
  remains a conjecture to be proven.

Also if we stack an invertible phase $(\cE,\cM_\cE,c_1)$ onto $(\cC,\cM,c_2)$, it
only changes the modular extension part,
\begin{align}
  (\cE,\cM_\cE,c_1)\bt_\cE(\cC,\cM,c_2)=(\cC,\cM_\cE\bt_\cE\cM,c_1+c_2).
\end{align}
By stacking all invertible phases (all modular extensions of $\cE$), all modular extensions of $\cC$ can be
generated. Moreover, the ``difference'' between two modular extensions is a
unique invertible phase (unique modular extension of $\cE$). In short, the
modular extensions of $\mce{\cE}$ $\cC$ form a
torsor over the Abelian group $\mext(\cE)$.

Therefore, a $\mce{\cE}$ $\cC$, if its modular extension exists, already fixed
the topological phase up to invertible ones. Appending the modular extension to
the label further fixes the invertible ones up to $E_8$
states\footnote{UMTC fixes central
charge $c$ modulo 8.}, and appending
the central charge $c$ totally fixes the topological phase.
On the other hand, if a $\mce{\cE}$ $\cC$ has no modular extension, namely the
symmetry can not be gauged, it is anomalous and can only be realized on
the boundary of (3+1)D topological phases~\cite{CBVF1403.6491}.

\chapter{Anyon Condensation}

In this chapter we discuss other constructions that relate topological phases.
The general pattern is starting from a phase $\cC$, condensing certain anyons
into certain states, and driving a transition into a new phase $\cD$. The other
anyons that are not condensed become anyons in the new phase $\cD$.

To do
this, let us first consider constructing an effective theory for anyons. For
simplicity, assume that we are going to condense only one type of anyon,
denoted by $A$. We then try to turn off the energy cost of $A$. In terms of the
Hamiltonian $H_\cC$ of the old phase $\cC$, we
assume that there is certain parameter $g$ controlling the interactions of the
underlying system, such that at $H_\cC(g=1)$ we have the original phase $\cC$
and at $H_\cC(g=0)$ the anyon $A$ becomes gapless.
In other words, when $g=0$ we arrive at a critical point, where any
many-$A$-anyon state becomes degenerate with the ground state. 
Next we further tune the Hamiltonian, by adding an effective Hamiltonian $H_{A}$
on many-$A$-anyon states. $H_A$ describes the state that we want the anyon $A$ to
condense into.
The new phase $\cD$ is then
\begin{align}
  H_\cD=H_\cC(g=0)+H_A.
\end{align}
We can also describe this state in terms of effective
many-$A$-anyon wavefunction $\langle \{z_a\}|\Psi\ket$, where $\{z_a\}$ denote
the positions of $A$ anyons, and $\Psi$ is the effective ground state of $H_A$.
Such effective wavefunction allows us to write down the ground state wavefunction of the new
phase $\cD$.
The ground
state of phase $\cD$ is given by
\begin{align}
  |0_\cD\>=\sum_{\{z_a\}}|\{z_a\}\>\langle \{z_a\}|\Psi\ket.
\end{align}
Let $\{z_i\}$ denote the degrees of freedom of the underlying
system (such as spins, electron positions), we have the following wavefunction:
\begin{align}
  \<\{z_i\}|0_\cD\>=\sum_{\{z_a\}}\<\{z_i\}|\{z_a\}\>\langle \{z_a\}|\Psi\ket,
\end{align}
where $\<\{z_i\}|\{z_a\}\>$ is the
many-$A$-anyon wavefunction in the old phase $\cC$.\footnote{For example in the $\nu=1/m$
Laughlin state:
\begin{align*}
  \<\{z_i\}|\{z_a\}\>=\prod (z_a-z_i)\prod_{i<j} (z_i-z_j)^m \times
   \ee^{-\frac14\sum |z_i|^2},
\end{align*}
where $\{z_a\}$ are the positions of quasi-hole excitations.}

In order to perform such condensation, it is obvious that
the properties of anyon $A$ and the target state $H_A,\Psi$ must satisfy
non-trivial consistent conditions.
We analyse two variants in the following.
In concrete physical systems, the above ideas may not be easy to be realized
precisely. But for the two variants to be discuss below, we do have precise
mathematical (categorical) constructions from phase $\cC$ to phase $\cD$.

\section{Type I: Bose Condensation}\label{bosecon}
The first variant is condensing a boson $A$ into a trivial state, or a
$A$-condensate. This means that we want the effective wavefunction to be
$\langle \{z_a\}|\Psi\ket=1$ for any configurations $\{z_a\}$. The underlying
mathematics is introduced in the last chapter. The condensable algebra
$(A,m,\eta)$ is a self boson to be condensed, and the morphisms $m,\eta$ exactly describes the
condensation process.
Roughly speaking,
recall that $m$ is a ``multiplication morphism'', an operator mapping
from two copies of $A$ to a single $A$. The isometric condition
$mm^\dag=\id_A$ means that $m^\dag m$ is a projector acting on two copies of
$A$. We can consider that $m^\dag m$ projects a pair of $A$ onto a
``singlet'' state. The unit, associative and commutative conditions ensures that
such projectors can be consistently applied to any numbers of $A$-anyons and
leads to a singlet state, which is exactly the $A$-condensate. So $m$ is related
to the effective theory $H_A\sim -\sum m^\dag m$ for the new phase $\cD$.

The special algebra $L_\cE$ discussed in the last chapter
corresponds to breaking the symmetry of the two-layer system from $\cE\bt\cE$ to $\cE$. It is global
symmetry for $\mce{\cE}$'s, but gauge symmetry for the gauged theory, or modular
extensions. Other algebras can be considered as inducing general ``topological symmetry breaking''. 

In general, given a topological phase $(\cC,\cM,c)$ with symmetry $\cE$,
condensing a condensable algebra $A\in\cC$, gives rise to a
new topological phase described by $(\cC_A^0,\cM_A^0,c)$. If $A\cap \cE$ (the
largest subalgebra of $A$ in $\cE$) is non-trivial, such condensation will
break the symmetry to a smaller one, otherwise the symmetry is preserved. At the same time, the
condensation creates a gapped domain wall between the old and the new phases.
Point-like excitations on the domain wall should be described by $\cC_A$ (and
$\cM_A$ for the gauged theories).

Recall that $\cC_A^0, \cC_A$ are just categories of (local) $A$-modules in
$\cC$, so the Bose condensation is essentially the representation theory of
algebras in unitary braided fusion categories, generalizing that in usual
vector spaces. It is possible to spell out all the data in terms of tensors and
write down a concrete representation theory~\cite{ERB1310.6001}. However, this
way is not efficient when $A$ is a large algebra.

Again we can study the Bose condensation at the level of physical
observables. This will lead to some necessary conditions.

First we restrict to bosonic topological orders with no symmetry. Assume that
by condensing $A$ in phase $\cC$, we obtain a new phase $\cD=\cC_A^0$, and a
gapped domain wall $\cW=\cC_A$. Let the topological $S,T$-matrices for $\cC,\cD$
be $(S^\cC,T^\cC),(S^\cD,T^\cD)$. We consider the following fusion space:
Put phase $\cC$ and phase $\cD$  on a sphere $S^2$, separated by the gapped
domain wall $\cW$, and an anyon $a^*$ in
phase $\cC$, an anyon $i$ in phase $\cD$. We denote the corresponding fusion
space by
$\cV(S^2,i,\cW,a^*)$. Its dimension
\begin{align}
  W_{ia} :={\dim}[\cV(S^2,i,W,a^*)],
\end{align}
is an important physical observable.
The matrix $W$ satisfies the following necessary conditions~\cite{LWW1408.6514,Kawahigashi1504.01088},
\begin{align}
  S^\cD W&=W S^\cC, T^\cD W= W T^\cC,\nonumber\\
  W_{ia}W_{jb}&\leq\sum_{kc} (N^\cD)_{ij}^k W_{kc} (N^\cC)_{ab}^c.\label{tunnel}
\end{align}

We can compute the dimension of the fusion space $\mathcal V(S^2,i,W,a^*)$ by first creating a
pair $aa^*$ in phase $\cC$, then tunneling $a$ through the domain wall.
In the channel where the tunneling does not leave any topological quasiparticle
on the domain wall, $a$ in phase $\cC$ will become
a composite anyon $q_{\cW,a}$ in phase $\cD$,
\begin{align}
  q_{\cW,a}=\oplus_i W_{ia} i.
\end{align}
Thus the fusion-space dimension $W_{ia}$ is also the \emph{number} of tunneling channels from,
$a$ of phase $\cC$, to, $i$ of phase $\cD$. So we also refer to $W$ as
the ``tunneling matrix''.

We may as well create a pair $ii^*$ in phase $\cD$ and tunnel $i^*$ to $a^*$.
$W^\dag$ describes such tunneling in the opposite direction (i.e.,
$W:A\to B,~ W^\dag: B\to A$). $W^\dag$ and $W$ contains the same physical
data. To be consistent, tunneling $i^*$ to $a^*$ should give rise to the same
fusion-space dimension, $(W^\dag)_{a^*i^*}=W_{i^*a^*}=W_{ia}$.
This is guaranteed by $W (S^\cC)^2=(S^\cD)^2W$.

In particular, since the algebra $A$ condenses, it becomes the vacuum in phase
$\cD$, thus tunnelling the trivial particle $\one$ in $\cD$ to $\cC$ should give
the algebra $A$,
\begin{align}
  A\cong \oplus_a W_{\one a} a.
\end{align}
And $W_{ia}$ itself gives a Lagrangian condensable algebra $L$ in the folded phase
$\cC\bt\ov\cD$,
\begin{align}
  L\cong \oplus_{a\in\cC,i\in\cD}\, W_{ia}\, a\bt i^*.
\end{align}

We conclude that the tunnelling matrix $W$ satisfying \eqref{tunnel} can fix the
object in the triple $(A,m,\eta)$ of an condensable algebra, but the data of
$m,\eta$ are missing. Indeed, there are known examples that $W$ does not
correspond to any valid algebra~\cite{Davydov1412.8505}. However, the conditions
\eqref{tunnel} are good enough to exclude impossible Bose condensation and pick
up a few candidates of condensable algebras.

For a topological phase $(\cC,\cM,c)$ with symmetry $\cE$, the above still works
for the modular extension $\cM$. We require that the condensed boson $A$ is
in $\cC$,
\begin{align}
  W_{\one a}=0, \quad\text{if } a\notin \cC.
\end{align}
And if $W_{\one a}=\delta_{\one a}$ for $a\in \cE$, equivalently $A\cap
\cE=\one$, the symmetry is preserved;
otherwise the symmetry is broken to a smaller one.

Note that in the presence of symmetry, a condensable algebra can be
\emph{anomalous}.
Let $(\cC,\cM,c)$ be a topological phase with symmetry $\cE$, and $A$ a
condensable algebra in $\cC$ such that $A \cap \cE=\one$ ($A$ does not break symmetry).
Let $(\cD,\cN,c)=(\cC_A^0,\cM_A^0,c)$ be the phase after Bose-condensing $A$.
The corresponding domain wall is $(\cC_A, \cM_A)$ (before gauging $\cC_A$ and
after gauging $\cM_A$).

The following is true, which can be thought as boundary-bulk duality (the
domain wall is the boundary of the folded two-layer phase),
\begin{enumerate}
  \item  $Z(\cM_A)=\cM \bt \ov\cN$.
   \item  Similarly, $Z(\cC_A)$ is a modular extension of $\cC \bt_\cE
     \ov\cD$. Here the embedding is determined as follows. Firstly, $\cE$ is a
     full subcategory of $\cD=\cC^0_A\subset\cC_A$ with the embedding
     $\cE\hookrightarrow \cC_A$ as first embedding $\cE$ into $\cC$ and then
     take the free modules. Thus $\cE\hookrightarrow \cC_A$ is a central
     functor and lifts to an embedding $\cE\hookrightarrow Z(\cC_A)$. Then
     $\cC\bt_\cE\ov\cD=\cen{\cE}{Z(\cC_A)}$.
\end{enumerate}

However, it may not be true that $Z(\cC_A)=\cM \bt_\cE \ov\cN$.  
The difference between $Z(\cC_A)$ and $\cM \bt_\cE \ov\cN$ is a modular
extension $\cK$ of $\cE$, namely a (2+1)D
$\cE$-SPT phase, $(\cE,\cK,c=0)$, $Z(\cC_A)=\cM \bt_\cE \ov\cN \bt_\cE \cK$.
Its physical meaning is that the domain wall $(\cC_A,\cM_A)$ must have a
"bulk": a (2+1)D domain wall in the (3+1)D bulk, and the (2+1)D domain wall
hosts the corresponding SPT phase.
When $\cK$ is non-trivial, $\cK \neq Z(\cE)$, we say that the algebra $A$ is
anomalous. We will discuss an example of this in Appendix~\ref{mirrorSET}.

If two topological phases with symmetry $\cE$ are related by anomaly-free Bose condensations
that preserve the symmetry, we say that they are \emph{Witt equivalent} over
$\cE$.
The equivalence classes are called Witt classes, denoted by
$\text{Witt}_\cE$. Taking Witt classes is compatible with the stacking
$\bt_\cE$, namely Witt classes still form a commutative monoid under the
stacking $\bt_\cE$. Moreover, due to Lemma~\ref{Lag} and Theorem~\ref{cze}, the
inverse Witt class always exists, given by the mirror conjugate. Thus
Witt classes $\text{Witt}_\cE$ actually form an Abelian group, called the Witt
group. Note that we take into account modular extensions and central charges in our definition
  of Witt classes. The Witt group defined in Ref.~\cite{DNO1109.5558}, the equivalence classes of
  $\mce{\cE}$'s alone under Bose condensations, which does not exclude
  anomalous $\mce{\cE}$'s that have no modular extension, or anomalous Bose
  condensations discussed above, can be different from our definition; when
  all these anomalies vanish, it is $(\text{Witt}_\cE/8\Z)/\mext(\cE) $.
Various constructions, such as symmetry breaking $\cE\to\cE'$ or stacking
$-\bt\cE'$, can induce group homomorphisms from $\text{Witt}_\cE$ to
$\text{Witt}_{\cE'}$ or $\text{Witt}_{\cE\bt\cE'}$.

Topological phases in the same Witt class have the same central charge,
similar topological spins and mutual statistics as their $S,T$-matrices are
related via \eqref{tunnel}. This is one way to ``group'' topological phases.

\section{Type II: Abelian Condensation}

The second variant is condensing Abelian anyons\footnote{Abelian anyons are
anyons with quantum dimension $1$. Here ``Abelian'' means that the braiding
processes
between Abelian anyons commute with each other, as they are just phases factors. On
the contrary, braiding processes between non-Abelian anyons in general do not
commute, and must be represented by matrices.} into a Laughlin-like state. This
idea dates back to Haldane and Halperin, known as ``hierarchy''
construction~\cite{Haldane83,Halperin84}.
But below we discuss it at a more general level.

We start with a topological phase $\cC$. The anyons in $\cC$ are labeled by
$i,j,k,\cdots$.  Let $a_c$ be an Abelian anyon in $\cC$ with spin $s_c$.  We
try to condense $a_c$ into the Laughlin state,
\begin{align}
  \langle\{z_a\}|\Psi\rangle=\prod_{a<b}(z_a-z_b)^{M_c}\times\ee^{-\frac14\sum|z_a|^2}.
\end{align}
The resulting topological
phase is described by $\cD$, determined by $\cC$, $a_c$ and $M_c$.  
    Here $z_a,z_b$ are the positions of $a_c$ anyons. $M_c$ must be consistent with
    anyon statistics.
    Consider
    exchanging two $a_c$ anyons, we obtain: a phase factor $\ee^{2\pi\ii
      \frac{M_c}{2}}$ from the wave function and a phase factor $\ee^{2\pi\ii s_{a_c}}$ from anyonic statistics.
     To be consistent, total phase factor must be 1:
     \begin{align}
     \frac{M_c}{2}+s_{a_c}\in \Z.
   \end{align}
	  So we need to take {$M_c=m_c-2s_{a_c}$}, where
	  $m_c$ is an even
	integer.

  Anyon $i$ in the phase $\cC$ may be dressed with a flux $M_i$ in the new phase
  $\cD$.
  \begin{align}
    \Psi(i,M_i)=\prod_b(\xi_i-z_b)^{M_i}\prod_{a<b}(z_a-z_b)^{M_c}\times\ee^{-\frac14\sum|z_a|^2}.
  \end{align}
$\xi_i$ is the position of anyon $i$.
Thus an anyon in the new phase is represented by a
pair {$(i,M_i)$.
Again, $M_i$ can not be arbitrary. If $a_c$ has trivial mutual statistics with
$i$, $M_i$ can be any integer. Otherwise, consider moving $a_c$ around
$(i,M_i)$ and we obtain: a phase factor $\ee^{2\pi\ii M_i}$ from the flux $M_i$
and a phase factor $\ee^{2\pi\ii {t_i}}$ from the mutual statistics between
    $a_c$ and $i$. The mutual statistics can be extracted from the $S$ matrix,
    $\ee^{2\pi\ii {t_i}}=DS_{i a_c^*}/d_i$, $t_{a_c}=2s_{a_c}$.
To be consistent, total phase factor must be 1:
\begin{align}
  M_i+t_i\in\Z.
\end{align}

  Next we compute the fusion rules and spins in the new phase $\cD$. The spin of $(i,M_i)$ is given by the spin of $i$ plus the ``spin'' of  the flux $M_i$:
  \begin{align}
    s_{(i,M_i)}=s_i+\frac{M_i^2}{2M_c}.\label{spinim}
  \end{align}
  To fuse anyons $(i,M_i),(j,M_j)$ in the new phase, just fuse $i,j$ as in
  the old phase, and add up the flux:
  \begin{align}
    (i,M_i)\otimes (j,M_j)=\bigoplus_k N^{ij}_k (k,M_i+M_j).
  \end{align}
  But note that this is not the final fusion rules, because anyons $(i,M_i)$ in
  the new phase are subject to the equivalence
  relation
  \begin{align}
    (i,M_i)\sim (i\otimes a_c,M_i+M_c). \label{equivim}
    \end{align}
  This is because the anyon $a_c$ dressed with a flux $M_c$ is a ``trivial excitation'' in the
  new phase:
  \begin{gather}
    \Psi(a_c,M_c)\sim\prod_b^n(\xi_{a_c}-z_b)^{M_c}\prod_{a<b}^n(z_a-z_b)^{M_c}=\prod_{a<b}^{n+1}(z_a-z_b)^{M_c},\nonumber\\
    (a_c,M_c)\sim (\one,0).
  \end{gather}
  The anyon types in $\cD$ actually correspond to the equivalence classes.
  After imposing the equivalence relation one obtains the final fusion rules in
  the new phase.

Applying the equivalence relation \eqref{equivim} $q$ times, we obtain
\begin{align}\label{equivq}
  (i,M_i)\sim (i\otimes a_c^{\otimes q},M_i+qM_c).
\end{align}
Let $q_c$ be the ``period'' of $a_c$, i.e., the smallest positive integer such
that $a_c^{\otimes q_c}=\one$. We see that
\begin{align}
  (i,M_i)\sim (i,M_i+q_cM_c).
\end{align}
Thus, we can focus on the reduced range of $M_i+t_i\in\{0,1,2,\cdots,q_c|M_c|-1\}$.
Let $N^\cC, N^\cD$ denote the rank of $\cC, \cD$ respectively.
Within the
reduced range, we have $q_c|M_c|N^\cC$ different labels, and we want to
show that the orbits generated by the equivalence relation \eqref{equivq} all
have the same length, which is $q_c$. To see this, just note that for
$0<q<q_c$, either $i\neq i\otimes a_c^{\otimes q}$, or if $i=i\otimes
a_c^{\otimes q}$, $M_i\neq M_i+qM_c$; in
other words, the labels $(i,M_i)$ are all different within $q_c$ steps. It
follows that the rank of $\cD$ is $N^\cD=|M_c|N^\cC$.  

The above enables us to extend the construction to categorical level, which goes down to the level of $F,R$ matrices.

The first step is to construct a unitary braided fusion category $\tilde\cD$, based on the
observation that the range of the second flux label can be reduced to
$q_c|M_c|$. Such $\tilde\cD$ can be viewed as an 
``extension'' of $\cC$ by $\Z_{q_c|M_c|}$. The anyons are labeled by the pair
$(i,M_i)$ where $i\in\cC$
and $M_i+t_i \in \Z_{q_c|M_c|}$. Fusion is then given by addition
\begin{align}
  (i,M_i)\otimes (j,M_j)=\oplus_k N^{ij}_k (k,[M_i+M_j]_{q_c|M_c|}),
\end{align}
where $[\cdots]_{q_c|M_c|}$ denotes the residue modulo $q_c|M_c|$. The $F,R$-matrices in
$\tilde{\cD}$ are given by those in $\cC$ modified by appropriate phase
factors. More precisely, let $F^{i_1i_2i_3}_{i_4}$ 
and $R^{i_1i_2}_{i_3}$ be the $F,R$-matrices in $\cC$; then in $\tilde\cD$ we take
\begin{align}
  \tilde F^{(i_1,M_1)(i_2,M_2)(i_3,M_3)}_{(i_4,M_4)}
  &=F^{i_1i_2i_3}_{i_4}\ee^{
    \frac{\pi\ii}{M_c}M_1(M_2+M_3-[M_2+M_3]_{q_c|M_c|})},\nonumber\\
    \tilde
    R^{(i_1,M_1)(i_2,M_2)}_{(i_3,M_3)}&=R^{i_1i_2}_{i_3}\ee^{\frac{\pi\ii}{M_c}M_1M_2}.
\end{align}
It is straightforward to check that they satisfy the pentagon and hexagon
equations, and $\tilde \cD$ is a valid unitary braided fusion category. Moreover, the
modified $R$ matrices do give us the desired modified spin.  
The $S$ matrix is
\begin{align}
  S^{\tilde \cD}_{(i,M_i),(j,M_j)}&=\sum_k \frac{N^{ij}_k}{D_{\tilde \cD}} d_k \ee^{2\pi\ii [s_{(i,M_i)}+s_{(j,M_j)}-s_{(k,M_i+M_j)}]}\nonumber\\
    &=\sqrt{\frac{q_c}{|M_c|}}S^{\cC}_{ij}\ee^{-2\pi\ii\frac{M_iM_j}{M_c}}.
\label{StildeD}
\end{align}

The second step is to reduce $\tilde\cD$ to $\cD$.  Categorically, just note
that $\{(a_c,M_c)^{\otimes q},q=0,\dots,q_c-1\}$ forms a symmetric fusion
subcategory of $\cen{\tilde\cD}{\tilde\cD}$, which can be identified with $\Rep(\Z_{q_c})$; by condensing this
$\Rep(\Z_{q_c})$, \ie condensing the regular algebra $\Fun(\Z_{q_c})$ in
$\Rep(\Z_{q_c})$, we obtain the desired $\cD$.  Put it simply,
we just further impose the equivalence relation \eqref{equivq} in $\tilde\cD$,
such that one orbit of length $q_c$ is viewed as one type of anyon instead of
$q_c$ different types.  This way we complete the construction of Abelian anyon
condensation at full categorical level.

Below we will study the properties of $\cD$ in detail.
It would be more convenient to use $(i,M_i)$ directly, which is the same as
working in a $\tilde\cD$. Then we can further impose the equivalence relation. For example, when we need to
sum over anyons in $\cD$, we can instead do
\begin{align}
  \sum_{I\in\cD}\to\frac{1}{q_c}\sum_{i\in\cC}\sum_{m=0}^{q_c|M_c|-1}.
\end{align}
Now we are ready to calculate other quantities of the new phase
$\cD$. First, it is easy to see that the quantum dimensions remain the same
$  d_{(i,m)}=d_i$.
The total quantum dimension is then
\begin{align}
  D_\cD^2=\frac{1}{q_c}\sum_{i\in\cC}\sum_{m=0}^{q_c|M_c|-1}
  d_{(i,m)}^2=|M_c|D_\cC^2.
\end{align}
The $S$ matrix is
\begin{align}
\label{Sd}
S^{\cD}_{(i,M_i),(j,M_j)}&=\frac{1}{\sqrt{q_c}}S^{\tilde\cD}_{(i,M_i),(j,M_j)}=\frac{1}{\sqrt{|M_c|}}S^{\cC}_{ij}\ee^{-2\pi\ii\frac{M_iM_j}{M_c}}.
\end{align}
From the above it is easy to check that $\cen{\cC}{\cC}=\cen{\cD}{\cD}$. This means that the symmetry $\cE=\cen{\cC}{\cC}=\cen{\cD}{\cD}$ is preserved.

If both $\cC,\cD$ are UMTCs, the new $S^\cD,T^\cD$ matrices,
as well as $S^\cC,T^\cC$, should
both obey the modular relations $STS=\ee^{2\pi\ii \frac{c}{8}}T^\dag S T^\dag$,
from which we can extract the central charge of $\cD$.
Firstly, using the
modular relation for both $\cC$ and $\cD$, we find that 
\begin{align}
  \frac{1}{q_c\sqrt{|M_c|}}&\sum_{i,j,k\in\cC}\sum_{p=0}^{q_c|M_c|-1}\left\{\overline{S^\cC_{xi}}S^\cC_{ik}T^\cC_{kk}S^\cC_{kj}\overline{S^\cC_{jy}}\exp\left[\frac{2\pi\ii}{2M_c}(t_i+t_j-t_k+p)^2\right]\right\}
  \nonumber\\
  &=\exp\left( 2\pi\ii\frac{c^\cD-c^\cC}{8} \right)T^\cC_{xx}\delta_{xy}.
  \label{difc}
\end{align}
We can show that $  c^\cD-c^\cC=\sgn(M_c) \mod 8$. Using the reciprocity theorem for generalized
Gauss sums~\cite{BEW98}:
\begin{align}
  \label{reci}
  \sum_{n=0}^{|c|-1}\ee^{\pi\ii\frac{an^2+bn}{c}}=\sqrt{|c/a|}\ee^{\frac{\pi\ii}{4}[\sgn(ac)-\frac{b^2}{ac}]}\sum_{n=0}^{|a|-1}\ee^{-\pi\ii\frac{cn^2+bn}{a}},
\end{align}
where $a,b,c$ are integers, $ac\neq 0$ and $ac+b$ even. Thus,
\begin{align}
  &\ \ \ \
\sum_{p=0}^{q_c|M_c|-1}\exp\left[\frac{2\pi\ii}{2M_c}(t_i+t_j-t_k+p)^2\right]
  \\
  &=\frac{1}{q_c}\ee^{\frac{\pi\ii(t_i+t_j-t_k)^2}{M_c}}\sum_{p=0}^{q_c^2|M_c|-1}\ee^{\frac{\pi\ii}{M_cq_c^2}\left[q_c^2p^2+2q_c^2(t_i+t_j-t_k)p\right]}
  =\frac{\sqrt{|M_c|}}{q_c}\ee^{\frac{\pi\ii}{4}\sgn(M_c)}\sum_{p=0}^{q_c^2-1}\ee^{-\pi\ii[M_cp^2+2(t_i+t_j-t_k)p]}
  \nonumber\\
  &=\frac{\sqrt{|M_c|}}{q_c}\ee^{\frac{\pi\ii}{4}\sgn(M_c)}\sum_{p=0}^{q_c^2-1}\ee^{-\pi\ii
  (m_c-2s_c)p^2}\frac{S^\cC_{ia_c^{\otimes p}}}{S^\cC_{i\one}}
\frac{S^\cC_{ja_c^{\otimes
p}}}{S^\cC_{j\one}}\frac{\overline{S^\cC_{ka_c^{\otimes p}}}}{S^\cC_{k\one}}
\nonumber\\
  &=\frac{\sqrt{|M_c|}}{q_c}\ee^{\frac{\pi\ii}{4}\sgn(M_c)}\sum_{p=0}^{q_c^2-1}T^\cC_{a_c^{\otimes
  p},a_c^{\otimes p}}\frac{S^\cC_{ia_c^{\otimes p}}}{S^\cC_{i\one}}
\frac{S^\cC_{ja_c^{\otimes
p}}}{S^\cC_{j\one}}\frac{\overline{S^\cC_{ka_c^{\otimes p}}}}{S^\cC_{k\one}}.
\nonumber 
\end{align}
Substituting the above result into \eqref{difc}, we have
\begin{align}
  &\frac{1}{q_c\sqrt{|M_c|}}\sum_{i,j,k\in\cC}\sum_{p=0}^{q_c|M_c|-1}\left\{\overline{S^\cC_{xi}}S^\cC_{ik}T^\cC_{kk}S^\cC_{kj}\overline{S^\cC_{jy}}\exp\left[\frac{2\pi\ii}{2M_c}(t_i+t_j-t_k+p)^2\right]\right\}
  \nonumber\\
  &=\frac{1}{q_c^2}\ee^{\frac{\pi\ii}{4}\sgn(M_c)}\sum_{p=0}^{q_c^2-1}
  \sum_{k}T_{kk}^\cC T^\cC_{a_c^{\otimes p},a_c^{\otimes p}}
  \frac{\overline{S^\cC_{ka_c^{\otimes p}}}}{S^\cC_{k\one}}
  \sum_i \frac{\overline{S^\cC_{xi}}S^\cC_{ik}S^\cC_{ia_c^{\otimes
  p}}}{S^\cC_{i\one}}\sum_j
  \frac{S^\cC_{kj}\overline{S^\cC_{jy}}S^\cC_{ja_c^{\otimes p}}}{S^\cC_{j\one}}
  \nonumber\\
  &=\frac{1}{q_c^2}\ee^{\frac{\pi\ii}{4}\sgn(M_c)}\sum_{p=0}^{q_c^2-1}
  \sum_{k}T_{k\otimes a_c^{\otimes p},k\otimes a_c^{\otimes p}}^\cC 
  N^{k,a_c^{\otimes p}}_x N^{k,a_c^{\otimes p}}_y
  \nonumber\\
  &=\frac{1}{q_c^2}\ee^{\frac{\pi\ii}{4}\sgn(M_c)}\sum_{p=0}^{q_c^2-1}
  \sum_{k}T_{k\otimes a_c^{\otimes p},k\otimes a_c^{\otimes p}}^\cC 
  \delta_{k\otimes a_c^{\otimes p},x}\delta_{xy}
  \nonumber\\
  &=\ee^{\frac{\pi\ii}{4}\sgn(M_c)}T^\cC_{xx}\delta_{xy},
\end{align}
as desired.
In fact, based on the physical picture, we have a stronger result
\begin{align}
  c^\cD=c^\cC+\sgn(M_c) .\label{diffc}
\end{align}

In the following we refer to the above construction from $\cC$ to $\cD$ as the
one-step condensation. It is always reversible. In $\cD$,
choosing $a_c'=(\one,1),\ s_c'=\frac{1}{2M_c},\ m_c'=0,\ M_c'=-1/M_c$, and repeating
the construction, we will go back to $\cC$.  
To see this we perform the construction for a $\tilde\cD$. Taking
$(j,M_j)=(a_c')^*=(\one,-1)$ in \eqref{Sd} we find that the mutual
statistics between $(i,M_i)$ and $a_c'=(\one,1)$ is
$t'_{(i,M_i)}=\frac{M_i}{M_c}$.
Let $((i,M_i),P_i)$ label the anyons after the above one-step
condensation. We have two equivalence relations
\begin{align}
((i,M_i),P_i)\sim ((i\ot a_c,M_i+M_c),P_i),
\end{align}
which reduces $\tilde\cD$ to $\cD$ and 
\begin{align}
  ((i,M_i),P_i)\sim ((i,M_i)\ot(\one,1),P_i+M_c')=((i,M_i+1),P_i-1/M_c),
\end{align}
which arises from the second one-step condensation. Combining them we can
eliminate the flux labels such that every label is equivalent to a
representative of the following canonical form
\begin{align}
  ((i,-t_i),t_i/M_c),
\end{align}
which can then be identified with the anyon $i$ in $\cC$. It is easy to check
that the $F,R$-matrices for the representatives are the same as the original ones in
$\cC$. We also need to show that it is true for the whole equivalence class.
Note that imposing the equivalence relations is nothing but condensing
$\Rep(\Z_{q_c})$ and $\Rep(\Z_{q_c|M_c|})$, and the equivalence classes
correspond to the free modules over the regular algebras. Since taking free
modules is a braided tensor functor, we know that the resulting $F,R$-matrices
are equivalent to those of the representatives. Thus, we indeed come back to
the original phase $\cC$.

Therefore, Abelian anyon condensations are reversible, which defines an
equivalence relation between topological phases. We call the corresponding
equivalence classes the ``non-Abelian families''.

Now we examine the important quantity $M_c=m_c-2s_c$ which relates the ranks
before and after the one-step condensation, $N^\cD=|M_c|N^\cC$. Since $m_c$ is
a freely chosen even integer, when $a_c$ is not a boson or fermion ($s_c\neq 0$
or $1/2 \mod 1$), we can always make $0<|M_c|<1$, which means that the rank is
reduced after one-step condensation. We then have the important
conclusion:
  Each non-Abelian family have ``root'' phases with
  the smallest rank.
  Abelian anyons in the root phases must be bosons or
  fermions.

We can further show that the Abelian bosons or
fermions in the root phases have trivial mutual statistics among them.
To see this, assuming that $a,b$ are Abelian anyons in a root phase. Since
the mutual statistics is given by $DS_{ab}=\exp[2\pi\ii(s_a+s_b-s_{a\otimes
b})]$, and $a,b,a\otimes b$ are all bosons or fermions, non-trivial mutual
statistics can only be $DS_{ab}=-1$. Now consider two cases: (1) one of $a,b$,
say $a$, is a fermion, then by condensing $a$ (choosing $a_c=a$,
$m_c=2$, $s_c=1/2$, $t_b=1/2$), in the new phase, the rank remains the same but
$s_{(b,0)}=s_b+\frac{t_b^2}{2M_c}=s_b+1/8$, which means $(b,0)$ is an Abelian
anyon but neither a boson nor a fermion. By condensing $(b,0)$ again we can
reduce the rank, which conflicts with the root phase assumption. (2) $a,b$
are all bosons. Still we condense $a$ with $m_c=2,s_c=0,t_b=1/2$. In the new
phase the rank is doubled but
$s_{(b,0)}=s_b+\frac{t_b^2}{2M_c}=1/16$, which means further condensing $(b,0)$
with $m_c'=0$ the rank is reduced to $1/8$, which is again, smaller than the
rank of the beginning root phase, thus contradictory.

Therefore, in the root phases, Abelian anyons are bosons or fermions with
  trivial mutual statistics among them.
If we denote the full subcategory of Abelian anyons in a UBFC $\cC$ by $\Ab{\cC}$,
the above means that in a root phase $\cC$, $\Ab{\cC}$ is a symmetric fusion
category.
We also have a straightforward corollary:
  All Abelian topological orders
  have the same unique root, which is the trivial topological order.
In
other words, all Abelian topological orders are in the same trivial non-Abelian
family.

To easily determine if two phases belong to the same non-Abelian family, 
it is very helpful to introduce some \emph{non-Abelian invariants}:
\begin{enumerate}
  \item The fractional part of
    the central charge, $c$ mod 1.  Since the one-step condensation changes the central charge
    by $\sgn(M_c)$ (see \eqref{diffc}), we know that central charges in the same
    non-Abelian family can only differ by integers.
  \item It is not hard to check that, in
    the one-step condensation, the number of simple anyon types with the same quantum
    dimension $d$, denoted by $N(d_i=d)$,
    is also multiplied by $|M_c|$. Thus the ratio $N(d_i=d)/N$ is a constant
    within one non-Abelian family.
  \item  The third invariant is a bit involved. Note that
    in the one-step condensation, if $i$ has trivial mutual statistics with $a_c$,
    $t_i=0$, then $(i,0)$ in $\cD$ have the same spin as $i$ in $\cC$ and the same
    mutual statistics with $(j,M_j),\forall M_j$ as that between $i$ and $j$ in $\cC$.
    Therefore, the centralizer of Abelian anyons, $\cen{(\Ab{\cC})}{\cC}$, namely, the subset of anyons that
    have trivial mutual statistics with all Abelian anyons, is the same within one
    non-Abelian family.
\end{enumerate}

With these we can show the condition that $\Ab{\cC}$ is
symmetric fusion category (Abelian anyons are bosons or fermions with
  trivial mutual statistics among them) is also 
sufficient for a topological phase $\cC$ to be a root state with the
smallest rank among a non-Abelian family. First note that the rank of
$\Ab{\cC}$, $N^{\Ab{\cC}}$ is just the number of simple anyon types with quantum dimension
1, thus $N^{\Ab{\cC}}/N^\cC=N(d_i=1)/N$ is a constant. $\cC$ has the smallest rank if and
only if
$\Ab{\cC}$ has the smallest rank. On the other hand,
$\cN:=\cen{(\Ab{\cC})}{\cC}$
is also an invariant. We then have $\cen{(\Ab{\cC})}{\Ab{\cC}}
=\cen{(\Ab{\cC})}{\cC}\cap \Ab{\cC}=\Ab{\cN}\subset \Ab{\cC}$. As
$\Ab{\cN}$ is invariant, when $\Ab{\cC}$ is symmetric,
$\cen{(\Ab{\cC})}{\Ab{\cC}}=\Ab{\cC}$, it has the smallest rank which is the
same as $\Ab{\cN}$.

The non-Abelian family is yet another way to ``group'' topological phases. We
see that its invariants are quite different from those of Witt classes. The Bose
condensation preserves central charges and spins, but changes quantum
dimensions, while Abelian condensation
changes central charges and spins, but preserves quantum dimensions. So far we know that some Abelian condensation can be mimicked by
stacking with an auxiliary state and then perform Bose condensation. For
example, Abelian-condensing a $\Z_2$ fermion with $M_c=1$, is the same as
stacking with a $\Z_4$-fusion-rule state whose $s_i=0,1/8,1/2,1/8$, $c=1$, and
then Bose-condensing the fermion pair. However, it is not clear if every $a_c,M_c$
has such an auxiliary state. By now we consider the two types of anyon
condensations to be independent.

Similar to the roots in a non-Abelian family, in a Witt class there are
topological phases that no longer admit Bose condensations; they are good
representatives of the family/class. However, unlike the Abelian condensation,
Bose condensation are \emph{not} reversible. So a huge advantage of non-Abelian
families over Witt classes is that from a root one can reconstruct the whole
family. Classifying the root phases is the same as classifying all topological
phases. We have listed the low rank roots and non-Abelian families of
topological phases with no symmetry in Ref.~\cite{LW1701.07820}.

\chapter{Examples}
In this chapter we introduce several simple examples with the toric code UMTC
and Ising UMTCs. Besides directly describing bosonic topological orders, they can also be
viewed as the gauged theories and describe topological phases with $\Z_2^f$ or
$\Z_2$ symmetries. There are also non-trivial anyon condensations between these
phases. This way with two simple UMTCs we can illustrate the general structures
discussed above.

We want to mention that there are far more examples studied in the literature
than mentioned in this thesis. There are several systematic approaches to
realize intrinsic topological orders, for example, the $K$-matrix formulation
for all Abelian topological orders~\cite{WZ92} (see Appendix~\ref{Hab} for a
brief introduction), the
Levin-Wen string-net model~\cite{LW0404617} for non-chiral topological orders with gapped
boundaries. Besides, conformal field theory and Kac-Moody algebras are also
very powerful in constructing chiral topological phases, but less
systematic than the previous two approaches.
They all give rise to concrete wavefunctions or lattice models for topological
orders.
However, most of them are limited to the
realization of intrinsic topological orders, not considering the symmetries. So
far only the string-net models are systematically extended to include bosonic
symmetries~\cite{HBFL1606.07816,CGJQ1606.08482}.

\section{Toric Code UMTC}
As the first example we describe the toric code~\cite{Kitaev9707021} UMTC. There
are 4 types of anyons, labeled by $\one,e,m,f.$ Their fusion rules and spins are given in
Table \ref{tcumtc}.
\begin{table}[h]
  \caption{Fusion rules and topological spins of toric code UMTC}
  \label{tcumtc}
  \medskip

  \centering
  \begin{tabular}{c|cccc}
   $i\ot j$ & $\one$ & $e$ & $m$ & $f$\\
   \hline
   $\one$&$\one$&$ e$&$m$&$f$\\
   $e$&$e$&$\one$&$f$&$m$\\
   $m$&$m$&$f$&$\one$&$e$\\
   $f$&$f$&$m$&$e$&$\one$\\
   \hline\hline
   $s_i$& 0&0&0&1/2
  \end{tabular}
\end{table}

For convenience, we also list its $S,T$-matrices
\begin{align}
  T=
  \begin{pmatrix}
    1&&&\\&1&&\\&&1&\\&&&-1
  \end{pmatrix},\quad
  S=\frac12
  \begin{pmatrix}
    1&1&1&1\\1&1&-1&-1\\1&-1&1&-1\\1&-1&-1&1
  \end{pmatrix}.
\end{align}

It can be realized by the toric code model, $\Z_2$ quantum double model,
$\Z_2$ gauge theory, or Levin-Wen string-net
model~\cite{LW0404617} with $\Rep(\Z_2)$ as the input fusion category. On a lattice of spin 1/2 (on
links), the fixed-point Hamiltonian reads
\begin{align}
  H=-\sum_\text{vertices} A_v -\sum_\text{plaquettes} B_p,
\end{align}
where $A_v$ is the product of $\sigma_z$ on the links around the vertex, and
$B_p$ is the product of $\sigma_x$ on the links around the plaquette. In the
string-net picture, we interpret $\sigma_z=-1$ as the presence of
(non-trivial) string, and
$\sigma_z=1$ as the absence of string (or presence of the trivial string). Thus
$A_v$ enforces that $\Z_2$ fusion rules of string (non-trivial strings
fuse to the trivial one; in other words, string cannot break at the vertex), and
$B_p$ creates a closed string loop in the plaquette and fuse it to the edges
of the plaquette.

The ground state is the equal weight superposition of all
closed loop configurations.
The $e$ anyons are created/annihilated/hopped by the string operators $W_e=\prod \sigma_x$,
flipping spins along the links. The $m$ anyons are created/annihilated/hopped
by the string operators $W_m=\prod \sigma_z$ acting on the dual lattice, along
paths that cross links. The $f$ anyons are created/annihilated/hopped by the
product of $W_e,W_m$.

\section{Ising UMTC}
The Ising fusion rules is given in Table \ref{Isingfus}, with three types of
anyons $\one,\sigma,\psi$.
\begin{table}[h]
  \caption{The Ising fusion rules}
  \label{Isingfus}
  \medskip

  \centering
  \begin{tabular}{c|ccc}
    $i\ot j$&$\one$&$\sigma$&$\psi$\\
    \hline
    $\one$&$ \one $&$\sigma$&$\psi$\\
    $\sigma$&$ \sigma $&$\one\oplus\psi$&$ \sigma$\\
    $\psi$&$\psi$&$\sigma$&$\one$
  \end{tabular}
\end{table}

Such fusion rules allow 8 different sets of solutions for pentagon and hexagon
equations, corresponding to 8 Ising-type UMTCs. They have spins
$s_\one=0,s_\psi=1/2,s_\sigma=1/16+n/8$, and central charge $c=1/2+n$, where
$n=0,1,\dots,7$. Usually by the Ising UMTC we mean the one with $n=0$,
$s_\sigma=1/16$, $c=1/2$, and its mirror conjugate $\Isb$ the one with
$n=7$, $s_\sigma=-1/16$, $c=-1/2$. The one with $n=1$, $c=3/2$ can be realized by the
$SU(2)_2$ Chern-Simons theory, the Moore-Read (or Pfaffian) state~\cite{MR91}
\begin{align}
  \Psi_\text{MR}=\mathrm{Pf}(\frac1{z_a-z_b})\prod_{a<b} (z_a-z_b)\times\ee^{-\frac14\sum|z_a|^2}.
\end{align}

All the 8 Ising-type UMTCs have the same $S$-matrix
\begin{align}
  S=\frac12
  \begin{pmatrix}
    1&\sqrt{2}&1\\
    \sqrt{2}&0&-\sqrt{2}\\
    1& -\sqrt{2}&1
  \end{pmatrix}.
\end{align}

Since the Ising UMTC is chiral, it cannot be realized by a commuting projector
lattice model. But $\Is\bt\Isb$ can be realized by the Levin-Wen model with
$\Is$ itself as the input fusion category. Now there are three types of strings,
$\one,\sigma,\psi$ (One may imagine a lattice with spin 1 on the links). Again
the Hamiltonian has the following form,
\begin{align}
  H=-\sum_\text{vertices} A_v -\sum_\text{plaquettes} B_p.
\end{align}
$A_v$ enforces the Ising fusion rule, such that the preferred string
configuration is: $\sigma$ strings form closed loops; $\psi$ strings either form
closed loops, or end on $\sigma$ strings. $B_p$ still creates $\sigma,\psi$
loops in the plaquette and fuse them to the edges. The detailed fusion process
is described by the data of Ising fusion category (mainly the $F$-matrix). The
quasiparticle excitations are described by the $\Is\bt\Isb$ UMTC.

\section{Bose Condensation}
In the above examples, Ising-type UMTCs do not allow any Bose condensation, but
there are several possible Bose condensations in the toric code UMTC or
$\Is\bt\Isb$ UMTC.

In the toric code UMTC, one can take the condensable algebra to be either
$\one\oplus e$ or $\one \oplus m$. After Bose condensation, the trivial
topological phase is obtained, also there is a gapped boundary whose excitations
are described by the $\Rep(\Z_2)$ fusion category. Namely there is only one non-trivial type of
particle-like excitation on the gapped boundary that has a $\Z_2$ fusion rule.

In the $\Is\bt\Isb$ UMTC, one can Bose-condenses the algebra $\one\oplus
\sigma\ov\sigma\oplus \psi\ov\psi$ and obtain the trivial phase. The
corresponding gapped boundary is described by the Ising fusion category.

The above Bose condensation is a general feature of topological phase $\cC$ that can
be realized by Levin-Wen models. There must a Lagrangian algebra $A$ in $\cC$
such that $\cC_A^0$ is the trivial phase, and $\cC_A$ is the fusion category
describing the corresponding gapped boundary. Also $\cC_A$ is a input fusion
category of the Levin-Wen model and $\cC=Z(\cC_A)$.

The other Bose condensation in the $\Is\bt\Isb$ UMTC is more interesting.
Condense the algebra $\one\oplus \psi\ov\psi$\footnote{In the literature this is
usually called condensing the fermion pair $\psi\ov\psi$.} and we will obtain exactly the
toric code UMTC.
If we order the anyons in the $\Is\bt\Isb$ UMTC as
$\one\ov{\one},\one\ov{\sigma},\one\ov{\psi},\sigma\ov{\one},\sigma\ov{\sigma},\sigma\ov{\psi},\psi\ov{\one},\psi\ov{\sigma},\psi\ov{\psi}$,
such Bose condensation corresponds to the following tunneling matrix
\begin{align}
W=
\left(
\begin{array}{ccccccccc}
 1 & 0 & 0 & 0 & 0 & 0 & 0 & 0 & 1 \\
 0 & 0 & 0 & 0 & 1 & 0 & 0 & 0 & 0 \\
 0 & 0 & 0 & 0 & 1 & 0 & 0 & 0 & 0 \\
 0 & 0 & 1 & 0 & 0 & 0 & 1 & 0 & 0 \\
\end{array}
\right),\quad
\begin{array}{c}
  \one\ov{\one}\to \one,\quad \psi\ov{\psi} \to \one,\\
  \one\ov{\psi}\to f,\quad \psi\ov \one \to f,\\
  \sigma\ov\sigma\to e\oplus m.
\end{array}
\label{dItoTc}
\end{align}

\section{As Invertible Fermionic Phases}
Let's consider the invertible fermionic phases with no other
symmetry. We have $\cE=\cC=\sRep(\Z_2^f)=\{\one,f\}$. It is easy to see that the
above toric code UMTC and 8 Ising-type UMTCs are all modular extensions of
$\sRep(\Z_2^f)$. In fact in $\mext(\sRep(\Z_2^f))$ there are also 7 other
Abelian rank 4 UMTCs with central charge $c=1,2,\dots,7$, constituting the
16-fold way~\cite{Kitaev0506438}.

Viewed as fermionic topological orders, $(\sRep(\Z_2^f),\text{toric code
UMTC}, c=0)$ is the trivial fermion product state.
$(\sRep(\Z_2^f),\Is,c=1/2)$, 
$(\sRep(\Z_2^f),\Isb,c=-1/2)$ correspond to $p\pm\ii p$ superconductors, where
the Ising anyon $\sigma$ corresponds to the vortex in the $p\pm \ii p$
superconductors.

The Bose condensation \eqref{dItoTc} introduced in the last section is also the
stacking $\bt_{\sRep(\Z_2^f)}$ for modular extensions. Physically it means that
stacking $p+\ii p$ with $p-\ii p$ produces the trivial fermion product state.

\section{As Topological Phases with $\Z_2$ symmetry}
First consider the invertible phases with $\Z_2$ symmetry, $\cE=\cC=\Rep(\Z_2)$.
It is easy to see the toric code UMTC is a modular extension of $\Rep(\Z_2)$.
The other modular extension of $\Rep(\Z_2)$ is the double-semion UMTC. This is
consistent with the fact that (2+1)D SPT phases with on-site unitary $\Z_2$
symmetry is classifies by $H^3(\Z_2,U(1))\cong \Z_2$.

Here we consider a non-trivial example, the toric code model with $\Z_2$
symmetry that exchanges $e,m$ anyons. In this case the original description of
toric code UMTC is no longer symmetric. The correct $\mce{\Rep(\Z_2)}$ $\cC$
turns out to have 5 types of anyons $\one_+,\one_-,f_+,f_-,\tau$. The first four
are the original anyons $\one,f$ carrying even/odd $\Z_2$ charge. The last one
$\tau$ is the composite of $e$ and $m$, $\tau\sim e\oplus m$. As the $\Z_2$
symmetry exchanges $e,m$, they together form a new anyon type $\tau$ with
quantum dimension $d_\tau=2$. This degeneracy cannot be lifted by symmetric
local perturbations.

One of its modular extension is $\Is\bt\Isb$, with the embedding
\begin{align}
  \one_+\mapsto \one\ov\one,\quad \one_-\mapsto \psi\ov\psi,
  \quad f_+\mapsto \psi\ov\one,\quad 
  f_-\mapsto \one\ov\psi,\quad \tau\mapsto\sigma\ov\sigma.
\end{align}
The other modular extension is then the stacking $\bt_{\Rep(\Z_2)}$ of
$\Is\bt\Isb$ with the double-semion UMTC, which turns out to be $SU(2)_2\bt
\ov{SU(2)_2}$.

As $\Is\bt\Isb$ is the gauged theory, the toric code model with on-site $e,m$ exchange symmetry can be realized by
``ungauging'' the Ising string-net model. Roughly speaking, this is done by
making the $\sigma$ strings in the string-net model into $\Z_2$ symmetry defects rather than fluctuating strings~\cite{HBFL1606.07816,CGJQ1606.08482}.

\chapter*{Conclusion and Outlook}
\addcontentsline{toc}{chapter}{Conclusion and Outlook}

In this thesis, we gave a classification of (2+1)D bosonic or fermionic topological phases with
finite on-site symmetries.
We first introduced the underlying mathematics, the theory of unitary
braided fusion categories, which describes the fusion and braiding of
quasiparticle
excitations. Then topological phases with symmetry are classified by a sequence of
UBFCs, $\cE\subset\cC\subset\cM$, plus a central charge $c$. Here $\cE$ is the
symmetric fusion category describing the local excitations, which carry
representations of the symmetry group. Thus, $\cE$ is also the categorical
description of the symmetry. $\cE=\Rep(G)$ corresponds to bosonic phases while
$\cE=\sRep(G^f)$ corresponds to fermionic phases. $\cC$ is the UBFC describing all the excitations,
whose M\"uger center coincides with $\cE$. $\cM$ is a minimal modular extension
of $\cC$ that describes the excitations in the gauged theory. $\cM$ encodes some
information of the invertible topological phases, in particular the SPT phases.
In the end, as $\cM$ only fixes the central charge $c$ modulo 8, the $E_8$ state
which has no symmetry, no bulk excitations, but edge state with central charge
8, is totally undetectable by the categorical approach. One can stack multiple layers of
$E_8$ states or its time-reversals without changing $\cE\subset\cC\subset\cM$.
To fix this ambiguity we appended the total central charge $c$ to
$\cE\subset\cC\subset\cM$ to complete the classification.

We also studied the stacking of topological phases and two types of anyon
condensations. They allow us to construct new topological
phases from known ones, and ``group'' them into suitable
equivalence classes or families, which simplifies the classification of
topological phases.

We have been focused on finite on-site unitary symmetries in the thesis. This
is mainly due to the technical difficulty dealing with the case where there
are infinitely many irreducible representations of the symmetry group. In 
Appendix~\ref{mirrorSET} we briefly discussed how to include anti-unitary
symmetries. It should be possible to overcome the technical difficulties, and extend the basic idea of the thesis to include also
continuous and space-time symmetries in the future.

Combined with previous results, a complete classification of topological phases
with symmetry in below 2+1D is almost at hand. It is then interesting
to investigate topological phases in 3+1D. The first step is to figure out 3+1D
topological orders, which may
require higher category theory as the natural underlying mathematical language.
Next we also need to combine topological orders with symmetries. But recall
that topological order appears as a new mechanism for phases of matter
starting from 2+1D, in 3+1D similarly we can have new mechanisms that are even
beyond topological order. Some examples are Haah's code~\cite{Haah1101.1962}, and
the stacking of infinite layers (extending to the 3rd dimension) of 2+1D
topological orders. A clear understanding of the new mechanisms in 3+1D is
still beyond our scope and will be an intriguing future project.


\bibliographystyle{unsrt}
\cleardoublepage 
\phantomsection  
\renewcommand*{\bibname}{References}

\addcontentsline{toc}{chapter}{\textbf{References}}

\bibliography{../library}


\appendix

\chapter*{APPENDICES}
\addcontentsline{toc}{chapter}{APPENDICES}
\chapter{Relation to the $G$-crossed UMTC approach}
In this appendix we discuss the relation between our approach and the
$G$-crossed UMTC approach for bosonic symmetry enriched
topological (SET) phases~\cite{BBCW1410.4540}. The latter
may be a bit more familiar to physicists. It fixes the underlying intrinsic
topological order, or a UMTC, and try to define the action of a symmetry group
$G$ on it. Besides, one also needs to consider the symmetry $G$-defects. The
$G$-defects can not be freely braided like the quasiparticles; they leave
defect lines behind. But, there is a ``$G$-crossed braiding'' for them. The
UMTC plus the $G$-defects together with the action of the symmetry group $G$,
forms the so called $G$-crossed UMTC.

Comparing to the $\mce{\Rep(G)}$ (UMTC over $\cE=\Rep(G)$) approach introduced in the main text,
this is just an equivalent perspective. $\mce{\Rep(G)}$ is the ``symmetric''
perspective while $G$-crossed UMTC is the ``symmetry-broken'' perspective. From a $G$-crossed UMTC, by taking representations
(equivariantization) we obtain $\Rep(G)=\cE\subset\cC\subset\cM$. More
precisely,
\begin{itemize}
  \item The tensor unit $\one$ (which spans the category of
    Hilbert spaces $\Hilb$) becomes the representation category $\Rep(G)$.
    In other words, local excitations acquire symmetry charges.
  \item The UMTC (trivial component in the $G$-crossed UMTC, trivial $G$-defects)
    becomes the $\mce{\Rep(G)}$ $\cC$. The topological excitations can carry usual 
    group representations or projective representations when the group actions
    do not permute topological charges, but more general ``representations''
    when the group actions permute topological charges.
  \item The $G$-crossed UMTC becomes the modular extension $\cM$. $G$-defects
    are promoted to gauge fluxes, dynamical excitations in the gauged theory.
\end{itemize}
On the other hand, from $\Rep(G)=\cE\subset\cC\subset\cM$, by breaking the
symmetry (condensing $\Rep(G)$ or the regular algebra $\Fun(G)$ in $\Rep(G)$,
de-equivariantization), we go back to the $G$-crossed UMTC and explicit
$G$-actions are recovered.

To illustrate this idea, let's consider the example, trivial topological order
with $\Z_2$ symmetry. In the $G$-crossed UMTC perspective, we consider all
local Hilbert spaces, the UMTC $\Hilb$ with a $\Z_2$ action. In particular we
allow the $\Z_2$ action to change local quantum states
\begin{align}
  |0\rangle\to|1\rangle.
\end{align}
We are not forbidden from describing the symmetry with its action on this
$|0\rangle,|1\rangle$ basis. But in a real physical system with $\Z_2$
symmetry, $|0\rangle$ alone can not be stable, i.e., can not be an energy
eigenstate, since it is not a representation of $\Z_2$. To really
observe the state $|0\rangle$, our probe has to somehow break the $\Z_2$
symmetry. On the other hand, the even/odd irreducible representations
\begin{align}
  |+\rangle=\frac{1}{\sqrt{2}}(|0\rangle+|1\rangle),\quad
  |-\rangle=\frac{1}{\sqrt{2}}(|0\rangle-|1\rangle),
\end{align}
can be stable excited states without breaking the $\Z_2$ symmetry and correspond
to the categorical $\cE=\Rep(G)$ way to describe the symmetry. 

Through the main text we use the term ``topological phases with symmetry''
instead of the term ``symmetry enriched topological phases'', or SET, which is
more common in the literature. This also reflects the different philosophies in
the two perspectives. For us, we fix the symmetry at first, and try to classify
all topological phases with this symmetry and study the stacking that preserves
the symmetry. In the $G$-crossed approach, one at first fixes a bosonic topological
order with no symmetry, and tries to add consistent $G$-actions and
$G$-defects. Thus the topological order is ``enriched'' by the symmetry.

There are two main differences between the two approaches. The first is that in
$G$-crossed UMTC approach, when trying to define the $G$-action on the
underlying UMTC $\cC$, since they are many layers of structures, such as anyon
types, local operators, fusion, braiding and so on, there can be obstructions in
$H^3(G,\Ab{\cC})$
for certain choice of the $G$-action. Only when the obstruction vanishes, one
can consistently define the $G$-action on the UMTC. However, a $\mce{\Rep(G)}$ is
automatically free of such obstructions. Similarly, when trying to consistently
add $G$-defects, there can be obstructions in $H^4(G,U(1))$, which prevent
from extending a UMTC with $G$-action to a $G$-crossed UMTC. Again, existence
of modular extensions implies that such obstructions vanish. But there are
indeed examples that certain UBFCs have no minimal modular
extensions~\cite{Drinfeld,CBVF1403.6491}. Non-vanishing $H^4(G,U(1))$ means that the
corresponding topological phase is anomalous and can only exist on the (2+1)D
surface of a (3+1)D SPT phase described by the obstruction in
$H^4(G,U(1))$~\cite{CBVF1403.6491}.
These obstructions are explicit in the $G$-crossed UMTC approach, but implicit
in our approach. It is not clear how to read out the obstructions directly from
$\mce{\cE}$'s without using the mathematical equivalence with
$G$-crossed UMTCs.

The second difference is more fundamental which forces us to take our new
perspective. Although the two approaches are equivalent for boson systems, the
$G$-crossed approach can not be applied to $\Z_2^f$, the fermion number parity.
We think that the underlying physical reason is that $\Z_2^f$ can not be broken, not only for the system but also for all our probes. Thus, only the ``symmetric'' perspective works. Again let's
use the example of trivial fermion topological order with no other symmetry, to
illustrate this. The even irreducible representation of $\Z_2^f$,
$z|+\rangle=|+\rangle$, is now physically a boson, and the odd irreducible
representation $z|-\rangle=-|-\rangle$ is a fermion. To ``observe'' the $\Z_2^f$
action
\begin{align}
  z|0\rangle=|1\rangle,
\end{align}
we must have the states
\begin{align}
  |0\rangle=\frac{1}{\sqrt{2}}(|+\rangle+|-\rangle),\quad
  |1\rangle=\frac{1}{\sqrt{2}}(|+\rangle-|-\rangle),
\end{align}
which are the superpositions of bosons with fermions. This is impossible.
Therefore, we
can only have ``representations'' but no ``actions'' of $\Z_2^f$. For fermionic
topological phases, we have to use $\sRep(G^f)=\cE\subset\cC\subset\cM$.
Surely one can break the bosonic part $G_b=G^f/\Z^f_2$, and obtain, similar to
$G$-crossed UMTCs, a theory of $G_b$-crossed $\mce{\sRep(\Z_2^f)}$'s. This is
possible but has not been very well developed comparing to its bosonic
companion. Nonetheless, the $\sRep(\Z_2^f)$ part can never be broken.

In the end, we want to mention that partially breaking the symmetry is also
of interest. This leads a theory that mixes the $G$-crossed part and the
over $\cE$ part, where $G_b$-crossed $\mce{\sRep(\Z_2^f)}$ is just a special
case. It may reveal more structures and provide us with explicit formulas on
group-cohomological classifications and obstructions for fermionic SPT and SET
phases.

\chapter{Mirror and Time-reversal Symmetry}\label{mirrorSET}
In this Appendix we briefly discuss the categorical description for mirror and
time-reversal symmetries.
Recall that the mirror conjugate $\ov\cC$ of a UBFC $\cC$ is canonically
braided equivalent to the time-reversal conjugate $\cC^{tr}$.
Therefore, the classification of time-reversal SETs should be the same as mirror
SETs.
We start by considering UMTCs with mirror symmetry action:

\begin{dfn}
  A topological phase (UMTC) $\cC$ has mirror symmetry (potentially anomalous)
  if there is a braided tensor equivalence $T:\cC\to\ov\cC$. More precisely, there
  are two braided tensor functors ${T:\cC\to\ov\cC}$, $\ov
  T:\ov\cC\to\ov{\ov\cC}=\cC$, such that
  $\ov T T\cong \id_\cC$, $T \ov T\cong \id_{\ov\cC}$.
\end{dfn}

Here the data $\ov T T\cong \id_\cC$, $T \ov T\cong \id_{\ov\cC}$ encode the mirror
symmetry fractionalization.

The existence of $T:\cC\cong\ov\cC$ implies that the central charge of $\cC$ is
$c=0$ or $4$ mod 8. A first anomaly-free condition is that $\cC$ has exactly zero central charge.
To study the other anomalies, since such equivalence $T$ is \emph{not} a braided tensor equivalence
(automorphism) from $\cC$ to $\cC$
itself, we can not directly apply the techniques developed for on-site symmetries
($G$-crossed UMTC or $\mce{\cE}$ with modular extensions). So we use the
folding trick~\cite{SHFH1604.08151,Lake1608.02736} to turn mirror symmetry into an on-site $\Z_2$ symmetry. Folding
the topological phase $\cC$ along the mirror axis, and we obtain a double-layer
phase $\cC\bt\ov\cC$ together with a canonical gapped boundary $\cC$. Alternatively,
such folding can be encoded in a canonical Lagrangian condensable algebra
$L_\cC\cong\bigoplus_{i\in\cC} i\bt i^*$ in $\cC\bt\ov\cC$. Condensing $L_\cC$
one obtains
the trivial phase $(\cC\bt\ov\cC)_{L_\cC}^0=\Hilb$ and the gapped
boundary $(\cC\bt\ov\cC)_{L_\cC}=\cC$ as a fusion category (forget braidings
on $\cC$).

Now the mirror symmetry is turned into the on-site $\Z_2$ action on
$\cC\bt\ov\cC$, on the gapped boundary $\cC$ and on $L_\cC$.
\begin{enumerate}
  \item \textbf{Action on $\cC\bt\ov\cC$}

Let $\tilde T:\cC\bt\ov\cC\xrightarrow{T\bt \ov T} \ov\cC\bt\cC\cong \cC\bt\ov\cC$, where
the second equivalence is just exchanging two layers. It is clear $\tilde T^2\cong
\id_{\cC\bt\ov\cC}.$ Moreover, it ``fractionalization''
$H^2(\Z_2,(\cC\bt\ov\cC)_{Ab})$ is always trivial.  This means that we
can take $\tilde T^2=\id_{\cC\bt\ov\cC}$.

\item \textbf{Action on the gapped boundary $\cC$}

We have the bulk to boundary map
\begin{align}
  \cC\bt\ov\cC &\to \cC\\
  X\bt Y&\mapsto X\ot Y\nonumber
\end{align}
As $T,\ov T$ are essentially the same tensor functor, the action on the gapped
boundary is just $ T:\cC\cong\cC$ viewed as a tensor functor.
(Note that braided tensor functor and tensor functor contain the \emph{same}
data. Being braided or not is just a \emph{property} of a tensor functor.)
$\tilde T$ is naturally the induced action in the bulk on $Z(\cC)=\cC\bt\ov\cC$
from the action $T$ on the boundary $\cC$.

\item \textbf{Action on $L_\cC$}

We see that the data $\ov T T\cong \id_\cC, T\ov T\cong \id_{\ov\cC}$ are
combined together in $\tilde T^2\cong \id_{\cC\bt\ov\cC}$. To see the mirror symmetry
 fractionalization, we consider the action on $L_\cC$, which is an \emph{algebra isomorphism} $\alpha:\tilde T(L_\cC)\cong L_\cC$ satisfying 
\begin{align}
  \xymatrix{\tilde T^2(L_\cC)\ar[r]^{\tilde T(\alpha)}\ar[d]^{\cong} &\tilde T(L_\cC)\ar[d]^{\alpha}\\
L_\cC\ar[r]^= &L_\cC}\label{act}
\end{align}
Here algebra isomorphism means that
\begin{align}
  \xymatrix{\tilde T(L_\cC)\ot \tilde T(L_\cC) \ar[rr]^-{\alpha\ot\alpha} \ar[d]^\cong
  &&L_\cC\ot L_\cC \ar[dd]^m \\
  \tilde T(L_\cC\ot L_\cC)\ar[d]^{\tilde T(m)}&&\\
  \tilde T(L_\cC)\ar[rr]^{\alpha}&&L_\cC}\label{algiso}
\end{align}
where $m$ is the multiplication morphism.

Note that the action $\alpha$ in general is \emph{not} an automorphism of
$L_\cC$, but an isomorphism from $L_\cC$ to $\tilde T (L_\cC)$.
We like to prove that $\alpha$ is the same data as $\ov T T\cong
\id_\cC$. Say $\alpha=\oplus\alpha_i$ where $\alpha_i:\tilde T(i\bt i^*)=\ov
T(i^*)\bt T(i) \to j^*\bt j$, $j=T(i)$. Set $\ov T
T(i)\xrightarrow{\alpha_i\id_i}
i$. To be consistent with $T\ov T T(i)\xrightarrow{\alpha_i\id_{T(i)}}
T(i)$, we should set $T\ov T(i)\xrightarrow{\alpha_{T(i)}\id_i}i$.  Under such
identification, \eqref{act} which now reads ``$\alpha_i\alpha_{T(i^*)}$ equals
$\tilde T^2(i\bt i^*)\to i\bt i^*$'', is the same as ``$\ov T T\bt T\ov T(i\bt i^*)\to i\bt
i^*$ equals $\tilde T^2(i\bt i^*)\to i\bt i^*$''. Also after writing the
multiplication $m$ out
explicitly, we can show that ``$\alpha$ is an algebra isomorphism
\eqref{algiso}'' is
equivalent to ``$\ov T T(i)\xrightarrow{\alpha_i\id_i} i$ is a monoidal natural
isomorphism''.
\end{enumerate}

Now we have a two-layer system with an on-site $\Z_2$ symmetry. We turn to the
``symmetric perspective'', $\mce{\cE}$'s with modular extensions.
We first ``take representations'' (equivariantization) of the on-site
$\Z_2$ action $\tilde T$, denoted by ${(\cC\bt\ov\cC)}^{\Z_2}$,
which is a $\mce{\Rep(\Z_2)}$.
Similar to representations in
the category of vector spaces, which is a pair $(V,\rho)$, a vector space $V$ with
actions $\rho:G\to \aute(V)$, an object (``representation'') in
$(\cC\bt\ov\cC)^{\Z_2}$ is a pair $(X,\mu_{\tilde T})$, an object $X\in\cC\bt\ov\cC$ with an
action of $\tilde T$, $\mu_{\tilde T}:X\cong X$ satisfying similar condition as \eqref{act}.

Thus, $(L_\cC,\alpha)$ is an object in $(\cC\bt\ov\cC)^{\Z_2}$. In
$(\cC\bt\ov\cC)^{\Z_2}$, choosing the action $\alpha$ is just assigning
``symmetry charges'' to components of $L_\cC$.   Since $\alpha$ is an
algebra isomorphism, $(L_\cC,\alpha)$ is also an algebra in $(\cC\bt\ov\cC)^{\Z_2}$.
This means that the assignment of ``symmetry charges'' must make the new object
in $(\cC\bt\ov\cC)^{\Z_2}$ still an algebra.
Condensing $(L_\cC,\alpha)$ we get
$[(\cC\bt\ov\cC)^{\Z_2}]_{(L_\cC,\alpha)}=\cC^{\Z_2}$ and
$[(\cC\bt\ov\cC)^{\Z_2}]_{(L_\cC,\alpha)}^0=\Rep(\Z_2)$, where $\cC^{\Z_2}$ is
the equivariantization of $\cC$ as fusion categories (i.e., the gapped
boundary).

$\cC^{\Z_2}$ is almost the same construction as $(\cC \bt \ov\cC)^{\Z_2}$.
The only thing need to be noted here is that we view $\cC$ as only a fusion
category without braidings, so the ``symmetry fractionalization on the
boundary'' is given by
$H^2( \Z_2, Z(\cC)_{Ab})$, valued in Abelian anyons in the center of $\cC$,
$Z(\cC)=\cC \bt \ov \cC$, rather than $\cC$ itself. This makes $H^2( \Z_2,
Z(\cC)_{Ab})$ trivial so we can always set $T^2=\id_\cC$ as tensor functors.
This is consistent with the fact that symmetry fractionalization of $\tilde T$ for the two-layer
system is trivial.
On the
other hand $\ov T T \cong \id_\cC, T \ov T\cong \id_{\ov C}$ as braided tensor
functors is independent data, which is ``mirror symmetry fractionalization''
of the single-layer system.

Now we are ready to discuss the anomaly of mirror symmetry fractionalization.
As it is now encoded in the action $\alpha$ on the algebra $L_\cC$, the anomaly
of mirror symmetry fractionalization is the same as the anomaly of the
condensable algebra $(L_\cC,\alpha)$. Recall the discussions in
Section~\ref{bosecon}. We can just take $Z(\cC^{\Z_2})$ as a modular extension of
$(\cC\bt\ov\cC)^{\Z_2}$. Condensing $(L_\cC,\alpha)$ in the gauged theory
$Z(\cC^{\Z_2})$ we get the gauged theory of a $\Z_2$-SPT phase, which reveals the anomaly of
the action $\alpha$ (or mirror symmetry fractionalization).
If $Z(\cC^{\Z_2})_{(L_\cC,\alpha)}^0=Z(\Rep(\Z_2))$, $\alpha$ is anomaly free; if $Z(\cC^{\Z_2})_{(L_\cC,\alpha)}^0\neq Z(\Rep(\Z_2))$ (and it
can only be the double-semion phase), $\alpha$ is
anomalous.

Next we discuss an example. In the toric code model, $\cC=Z(\Rep(\Z_2))$.
One possible mirror action is that the functor $T$ that acts as
identity on objects but maps to reversed braiding. Then $\cC\bt\ov\cC$ is just the double-layer
toric code, and $Z(\cC^{\Z_2})$ should be $D(D_4)\equiv
Z(\Ve_{D_4})=Z(\Rep(D_4))$, where
$\Ve_{G}$ denotes the category of $G$-graded vector spaces. The algebra in $D(D_4)$
corresponding to $eTmT$ (both $e,m$ has $T^2=-1$, or
$\alpha_\one=\alpha_f=1,\alpha_e=\alpha_m=-1$) gives a condensation to the
double-semion phase~\cite{QJW1710.09391}. Thus $eTmT$ is anomalous.

It would be beneficial to go through the equivariantization process in detail
and show the similarities and differences between equivariantization and taking
usual representations.
Let's calculate how $\cC^{\Z_2}$ gives $\Ve_{D_4}$ with
$\cC=Z(\Rep(\Z_2))=\{\one,e,m,f\}$ and $T$ the functor that does not permute
any objects.  Naively we have 8 simple objects in $\cC^{\Z_2}$, of the form $(i,x)$,
$i\in \cC$, $x=\pm 1$ corresponding to the morphism $T(i)\xrightarrow{x
\id_i} i$. But now one can not simply add up the $\Z_2$ charge $x$. Fusion in
$\cC^{\Z_2}$ is given by $(i,x)\ot (j,y)= (i\ot j, w)$, where $w$ is the morphism
$T(i\ot j) \cong T(i) \ot T(j) \xrightarrow{x \id_i\ot y \id_j} i\ot j$.

In order for $T$ to be a braided functor between $\cC$ and $\ov \cC$, we
can not take $T(i\ot j) \cong T(i) \ot T(j)$ to be simply identity morphisms.
More precisely, $T$ is braided if
\begin{align}
  \xymatrix{T(i\ot j) \ar[r]^-\cong \ar[d]^{T(c_{ij})} 
  &T(i)\ot T(j)\ar[d]^{c_{T(i),T(j)}}\\
T(j\ot i)\ar[r]^-\cong  &T(j)\ot T(i)}
\end{align}
where $c_{ij}$ is the braiding. For toric code, in certain gauge we have for example $c_{e,m}=1, c_{m,e}=-1$,
so in mirror conjugate $\ov\cC$ we have $c_{\ov e,\ov m}=c_{m,e}^{-1}=-1, c_{\ov
m, \ov e}=1$. It is then clear that the difference between $T(e\ot m) \cong
T(e) \ot T(m)$ and $T(m\ot e) \cong T(m) \ot T(e)$ must be $-1$. A good choice
happens to be $T(i\ot j)\xrightarrow{c_{ij}} T(i)\ot T(j)$.
We take
\begin{align}
  c_{ee}=c_{mm}=c_{em}=c_{fm}=c_{ef}=1,\quad c_{ff}=c_{me}=c_{mf}=c_{fe}=-1.
\end{align}
Thus the fusion rules are
\begin{align}
  (i,x)\ot (j,y)= (i\ot j, c_{ij}xy).
\end{align}
which is just the extension of $\Z_2\xt \Z_2$ by $\Z_2$ with ``2-cocycle''
$c_{ij}$. To see that it is the same as $D_4$, just check the following
\begin{itemize}
  \item $\Z_4$ subgroup $\{(\one,1),(f,1),(\one,-1),(f,-1)\}$.
  \item $\Z_2$ subgroup $\{(\one,1),(e,1)\}$.
  \item $(e,1),(f,1)$ generate the group, and $(e,1)\ot (f,1) \ot (e,1)=(f,-1)$.
\end{itemize}
Thus the fusion is isomorphic to $\Z_4\rtimes \Z_2=D_4$. Since
$\cC=Z(\Rep(\Z_2))$
has trivial associator or $F$-matrices, $\cC^{\Z_2}$ also has trivial $F$-matrices.
Thus $\cC^{\Z_2}$ is the fusion category with $D_4$ fusion rules and no
additional 3-cocycle twists, which means that $\cC^{\Z_2}=\Ve_{D_4}$,
$Z(\cC^{\Z_2})=D(D_4)$.

It is also interesting to calculate the case when $T$ permutes $e,m$. We have 5 simple objects in $\cC^{\Z_2}$,
$\{(\one,1),(\one,-1),(f,1),(f,-1),(e\oplus m,1)\},$ where $(e\oplus m,1)$ is
of quantum dimension 2. In this case $T(i\ot j) \cong T(i) \ot T(j)$ can be
chosen to be just identify morphisms (for $T:\cC\to\ov\cC$, but for unitary
on-site $e,m$ exchange they can not be identities). Its fusion rules are the same as $\Rep(D_4)$ or
$\Rep(Q_8)$. By calculating the $F$-matrices in $\cC^{\Z_2}$ it should be possible to
explicitly identify $\cC^{\Z_2}$ with $\Rep(D_4)$, but this way is too involved. We
use a result in Ref.~\cite{GNN0905.3117} to bypass it. The result says that
$Z(\cC\rtimes G)=Z(\cC^G)$. As the only non-trivial structure of $T$ acting on
$\cC$ is exchanging $e,m$, the corresponding $\cC\rtimes
\Z_2=\Ve_{\Z_2\xt\Z_2} \rtimes \Z_2=\Ve_{(\Z_2\xt\Z_2) \rtimes
\Z_2}=\Ve_{D_4}$. Thus we also have
$Z(\cC^{\Z_2})=Z(\Ve_{D_4})=D(D_4)=Z(\Rep(D_4))$. This also implies that $\cC^{\Z_2}$ must be
$\Rep(D_4)$.
(Note that when $T$ does not permute $e,m$, we showed that $T$ has other
non-trivial structures. In this case $\cC\rtimes
\Z_2=\Ve_{\Z_2\xt\Z_2} \rtimes
\Z_2=\Ve_{\Z_2\xt\Z_2\xt\Z_2}^{\omega_3}$ with some 3-cocycle $\omega_3$
twist.)

\chapter{Abelian Topological Orders, $K$-matrix and Abelian Condensation}
\label{Hab}

Consider a bosonic Abelian topological order, which can always be described by
an even $K$-matrix $K_0$ of dimension $\kappa$.  Anyons are labeled by
$\kappa$-dimensional integer vectors $\bm l_0$.  Two
integer vectors $\bm l_0$ and $\bm l'_0$ are equivalent (\ie describe the same
type of topological excitation) if they are related by
\begin{align}
  \bm l'_0 = \bm l_0 + K_0 \bm k,\label{eq0}
\end{align}
where $\bm k$ is an arbitrary integer vector. Fusion of anyons are done by
first adding up vectors and then imposing the above equivalence relation.
The mutual statistical angle between two
anyons, $\bm l_0$ and $\bm k_0$, is given by
\begin{align}
 \theta_{\bm l_0,\bm k_0} =2\pi \bm k_0^T K_0^{-1} \bm l_0.
\end{align}
The spin of the anyon $\bm l_0$ is given by
\begin{align}
 s_{\bm l_0} = \frac12 \bm l_0^T K_0^{-1} \bm l_0.
\end{align}

Next we discuss Abelian condensation in Abelian topological orders in the
$K$-matrix formulation.
Let us construct a new topological order from the $K_0$
topological order by assuming Abelian anyons labeled by $\bm l_c$ condense.
Here we treat the anyon as a bound state between a boson and flux.  We then
smear the flux such that it behaves like an additional uniform magnetic field,
and condense the boson into $\nu=1/m_c$ Laughlin state (where $m_c=$ even).
The resulting new topological order is described by the $(\kappa+1)$-dimensional
$K$-matrix
\begin{align}
K_1= \bpm K_0 & \bm l_c \\
     \bm l_c^T &  m_c \\
\epm
\end{align}

In the following, we are going to show that, to describe the result of the $\bm
l_c$ anyon condensation, we do not need to know $K_0$ directly.  We only need
to know the spin of the condensing particle $\bm l_c$
\begin{align}
s_c=\frac12 \bm l_c^T K_0^{-1} \bm l_c,
\end{align}
and the mutual statistics 
\begin{align}
\theta_{\bm l_0,\bm l_c}\equiv 2\pi t_{\bm l_0},\ \ \  t_{\bm l_0} = \bm l_c^T K_0^{-1} \bm l_0
\end{align}
between $\bm l_0$ and  $\bm l_c$.

First, we find that, as long as $m_c-2s_c\neq 0$, $K_1$ is invertible with
\begin{align}
 K_1^{-1} = 
\bpm
 K_0^{-1} +\frac{ K_0^{-1} \bm l_c \bm l_c^T K_0^{-1}}{ m_c - 2s_c}
& - \frac{K_0^{-1} \bm l_c}{ m_c - 2s_c} \\
- \frac{\bm l_c^T K_0^{-1} }{ m_c - 2s_c} &  \frac{1}{ m_c - 2s_c}
\epm
\end{align}
The anyons in the new $K_1$ topological order are labeled by
$\kappa+1$-dimensional integer vectors $\bm l^T = (\bm l_0^T, m)$.  The spin of $\bm
l$ is
\begin{align}
\label{slsl0}
 s_{\bm l}  &= \frac12 \bm l^T K_1^{-1} \bm l
 = \frac12 \Big(2s_0 + \frac{m^2+t_{\bm l_0}^2 -2 m t_{\bm l_0} }{m_c-2s_c}  \Big) 
\nonumber\\
&= s_{\bm l_0} +\frac12 \frac{(m-t_{\bm l_0})^2 }{m_c-2s_c}
\end{align}
The vectors 
$\bm l^T = (\bm l_0^T, m)$
and
$\bm l^{\prime T} = (\bm l_0^{\prime T}, m')$
are equivalent if they are related by
\begin{align}
  \label{equivk0k}
\bm l_0^{\prime } - \bm l_0 = K_0 \bm k_0 + k \bm l_c, \ \ \
m' - m = \bm l_c^T \cdot \bm k_0+m_c k,
\end{align}
for any $\kappa$-dimensional integer vector $\bm k_0$ and integer $k$. 
To avoid the gauge ambiguity, for the
integer vectors $\bm l_0$, we pick a representative for each equivalence class
(by \eqref{eq0}, fixing the gauge).
Taking $k=1$ and appropriate $\bm k_0$ such that $\bm l_0'$ and $\bm l_0$ are
the pre-fixed representatives, we see that 
\begin{align}
\label{equiv}
(\bm l_0^T, m) \sim ({\bm{l}'}_0^T\sim\bm l_0^{T}+\bm l_c^T, m+t_{\bm l'_0}-t_{\bm
l_0}+m_c-2s_c).
\end{align}

We also want to express the fusion in the new phase in terms of the pre-fixed
representatives $\bm l_1,\bm l_2,\bm l_3$. Assuming that $(\bm l_3^T,m_3)\sim (\bm
l_1^T+\bm l_2^T,m_1+m_2)$, and taking $k=0$ and appropriate $\bm k_0$ in
\eqref{equivk0k} (the cases of non-zero $k$ can be generated via
\eqref{equiv}), we find that
\begin{align}
  \label{fusion}
&(\bm l_1^T,m_1)+(\bm l_2^T,m_2)
\\
&\sim (\bm l_3^T\sim \bm l_1^T+\bm
l_2^T,m_3=m_1+m_2+t_{\bm l_3}-t_{\bm l_1}-t_{\bm l_2}).
\nonumber 
\end{align}

We can easily calculate the determinant of $K_1$ whose absolute value is the
rank of the new phase:
\begin{align}
  \det(K_1)&=\det\bpm K_0 & \bm l_c \\
     \bm l_c^T &  m_c \\
     \epm = \det(K_0)(m_c-\bm l_c^T K_0^{-1} \bm l_c)
     \nonumber\\
&=(m_c-2s_c)\det(K_0)
\end{align}
Let $M_c=m_c-2s_c$. It is an important gauge invariant quantity relating the
ranks of the two phases. If we perform the condensation with a different anyon $\bm l_c'$
and a different even integer $m_c'$, but make sure that $\bm l_c'\sim \bm l_c$
and $M_c'=m_c'-2s_c'=m_c-2s_c=M_c$, the new topological order will be the same.

It is worth mentioning that such construction is reversible: for the $K_1$
state, take $\bm l_c'^T=(\bm 0^T,1),m_c'=0$, and repeat the construction:
\begin{align}
  K_2=\bpm K_0&\bm l_c & 0\\
  \bm l_c^T& m_c&1\\
  0&1&0\epm \sim 
  \bpm K_0& 0 & 0\\
  0& 0&1\\
  0&1&0\epm
  \sim K_0.
\end{align}
We return to the original $K_0$ state.
\chapter{Selected Tables for Topological Phases with Symmetries}
\label{tables}
In this Appendix we give several tables of $\mce{\cE}$'s for various $\cE$, in
terms of quantum dimensions $d_i$ and topological spins
$s_i$~\cite{LKW1507.04673,LKW1602.05946}. Since we
are not able to calculate modular extensions for all entries, some of them may
be invalid. However, within certain numerical search bound, we have given all
possible candidates. Also as we showed in Chapter~\ref{stack}, as long as the $\mce{\cE}$ is
valid and has modular extensions, it determines the topological phases up to
invertible ones. So these tables can be viewed as listing candidates for topological
phases with symmetries, up to invertible ones.
As the classification of invertible bosonic phases is clear, we will only
mention the classification of invertible fermionic phases in the following
examples.
The numerical search was done by my supervisor Xiao-Gang Wen based on his algorithm
searching for bosonic topological orders (UMTCs)~\cite{Wen1506.05768}.

Table~\ref{SETZ2-34} lists all bosonic topological phases with $\Z_2$ symmetry
for $N=3,4$ and $D^2\leq 100$.
All the topological orders in this list have modular extensions, and are
realizable by (2+1)D boson systems.
We use $N^{|\Th|}_{c}$ to label $\mce{\cE}$'s, where
$\Theta ={D}^{-1}\sum_{i}\ee^{2\pi\ii s_i} d_i^2=
|\Th|\ee^{2\pi \ii c/8}$ and $D^2=\sum_id_i^2$. In the ``comment/$K$-matrix'' column, SB
means the state after symmetry breaking, $\bt^t$ indicates a twisted fusion
rule (from $\Z_2\xt\Z_2$ to $\Z_4$), $N^B_c$ means a bosonic topological
order with central charge $c$ (see Refs.~\cite{Wen1506.05768,LKW1507.04673}) and
$K$-matrix describes an Abelian topological order (see Appendix~\ref{Hab} for
a brief introduction).
Here $\zeta_n^m=\frac{\sin[\pi(m+1)/(n+2)]}{\sin[\pi/(n+2)]}$.  
\begin{table}[h] 
\caption{ Bosonic topological phases with $\Z_2$ symmetry.}
\label{SETZ2-34} 
\medskip

\centering
\begin{tabular}{ |c|c|l|l|l| } 
\hline 
$N^{|\Th|}_{c}$ & $D^2$ & $d_1,d_2,\cdots$ & $s_1,s_2,\cdots$ & comment/$K$-matrix \\
\hline 
$2^{\zeta_{2}^{1}}_{ 0}$ & $2$ & $1, 1$ & $0, 0$ & $\cE=\text{Rep}(\Z_2)$\\
\hline
$3^{\zeta_{2}^{1}}_{ 2}$ & $6$ & $1, 1, 2$ & $0, 0, \frac{1}{3}$ & 
SB:$K=$ $\begin{pmatrix}
 2 & -1 \\
 -1 & 2 \\
\end{pmatrix}
$
\\
$3^{\zeta_{2}^{1}}_{-2}$ & $6$ & $1, 1, 2$ & $0, 0, \frac{2}{3}$ & SB:$K=$ $\begin{pmatrix}
 -2 & 1 \\
 1 & -2 \\
\end{pmatrix}
$
\\
\hline
$4^{\zeta_{2}^{1}}_{ 1}$ & $4$ & $1, 1, 1, 1$ & $0, 0, \frac{1}{4}, \frac{1}{4}$ & $2^B_1\boxtimes \text{Rep}(\Z_2)$\\
 $4^{\zeta_{2}^{1}}_{ 1}$ & $4$ & $1, 1, 1, 1$ & $0, 0, \frac{1}{4}, \frac{1}{4}$ & $2^B_1\boxtimes^t \text{Rep}(\Z_2)$ \\
$4^{\zeta_{2}^{1}}_{-1}$ & $4$ & $1, 1, 1, 1$ & $0, 0, \frac{3}{4}, \frac{3}{4}$ & $2^B_{-1}\boxtimes \text{Rep}(\Z_2)$\\
 $4^{\zeta_{2}^{1}}_{-1}$ & $4$ & $1, 1, 1, 1$ & $0, 0, \frac{3}{4}, \frac{3}{4}$ & $2^B_{-1}\boxtimes^t \text{Rep}(\Z_2)$ \\
$4^{\zeta_{2}^{1}}_{ 14/5}$ & $7.2360$ & $1, 1,\zeta_{3}^{1},\zeta_{3}^{1}$ & $0, 0, \frac{2}{5}, \frac{2}{5}$ & $2^B_{ 14/5}\boxtimes \text{Rep}(\Z_2)$ \\
$4^{\zeta_{2}^{1}}_{-14/5}$ & $7.2360$ & $1, 1,\zeta_{3}^{1},\zeta_{3}^{1}$ & $0, 0, \frac{3}{5}, \frac{3}{5}$ & $2^B_{-14/5}\boxtimes \text{Rep}(\Z_2)$ \\
$4^{\zeta_{2}^{1}}_{ 0}$ & $10$ & $1, 1, 2, 2$ & $0, 0, \frac{1}{5}, \frac{4}{5}$ & 
SB:$K=$ $\begin{pmatrix}
 2 & -3 \\
 -3 & 2 \\
\end{pmatrix}
$
\\
$4^{\zeta_{2}^{1}}_{ 4}$ & $10$ & $1, 1, 2, 2$ & $0, 0, \frac{2}{5}, \frac{3}{5}$ & SB:$K=$ $
\begin{pmatrix}
2 & 1 & 0 & 0 \\ 
1 & 2 & 0 & 1 \\ 
0 & 0 & 2 & 1\\
0 & 1 & 1 & 2\\
\end{pmatrix}
$\\
 \hline 
\end{tabular} 
\end{table}

Table~\ref{SETZ3} lists all bosonic topological phases with $\Z_3$ symmetry, for
$N=4,5,6$ and $D^2\leq 100$,
$N=7$ and $D^2\leq 60$,
$N=8$ and $D^2\leq 40$.
\begin{table*}[h] 
\caption{Bosonic topological phases with $\Z_3$ symmetry.  } 
\label{SETZ3} 
\medskip

\centering
\begin{tabular}{ |c|c|l|l|l| } 
\hline 
$N^{|\Th|}_{c}$ & $D^2$ & $d_1,d_2,\cdots$ & $s_1,s_2,\cdots$ & comment/$K$-matrix \\
\hline 
$3^{\zeta_{4}^{1}}_{ 0}$ & $3$ & $1, 1, 1$ & $0, 0, 0$ & $\cE=\Rp(\Z_3)$\\
\hline
$4^{\zeta_{4}^{1}}_{ 4}$ & $12$ & $1, 1, 1, 3$ & $0, 0, 0, \frac{1}{2}$ & 
SB:$K=$ $\begin{pmatrix}
2 & 1 & 1 & 1 \\ 
1 & 2 & 0 & 0 \\ 
1 & 0 & 2 & 0\\
1 & 0 & 0 & 2\\
\end{pmatrix}
$
\\
\hline
$6^{\zeta_{4}^{1}}_{ 1}$ & $6$ & $1, 1, 1, 1, 1, 1$ & $0, 0, 0, \frac{1}{4}, \frac{1}{4}, \frac{1}{4}$ & $2^B_1\boxtimes \text{Rep}(\Z_3)$ \\
$6^{\zeta_{4}^{1}}_{-1}$ & $6$ & $1, 1, 1, 1, 1, 1$ & $0, 0, 0, \frac{3}{4}, \frac{3}{4}, \frac{3}{4}$ & $2^B_{-1}\boxtimes \text{Rep}(\Z_3)$ \\
$6^{\zeta_{4}^{1}}_{ 14/5}$ & $10.854$ & $1, 1, 1,\zeta_{3}^{1},\zeta_{3}^{1},\zeta_{3}^{1}$ & $0, 0, 0, \frac{2}{5}, \frac{2}{5}, \frac{2}{5}$ &  $2^B_{ 14/5}\boxtimes \text{Rep}(\Z_3)$\\
$6^{\zeta_{4}^{1}}_{-14/5}$ & $10.854$ & $1, 1, 1,\zeta_{3}^{1},\zeta_{3}^{1},\zeta_{3}^{1}$ & $0, 0, 0, \frac{3}{5}, \frac{3}{5}, \frac{3}{5}$ &  $2^B_{-14/5}\boxtimes \text{Rep}(\Z_3)$\\
\hline
$8^{\zeta_{4}^{1}}_{ 3}$ & $24$ & $1, 1, 1, 1, 1, 1, 3, 3$ & $0, 0, 0, \frac{3}{4}, \frac{3}{4}, \frac{3}{4}, \frac{1}{4}, \frac{1}{2}$ & $2^{ B}_{-1}\boxtimes 4^{\zeta_{4}^{1}}_{ 4}$\\
$8^{\zeta_{4}^{1}}_{-3}$ & $24$ & $1, 1, 1, 1, 1, 1, 3, 3$ & $0, 0, 0, \frac{1}{4}, \frac{1}{4}, \frac{1}{4}, \frac{1}{2}, \frac{3}{4}$ & $2^{ B}_{ 1}\boxtimes 4^{\zeta_{4}^{1}}_{ 4}$\\
$8^{\zeta_{4}^{1}}_{ 6/5}$ & $43.416$ & $1, 1, 1,\zeta_{3}^{1},\zeta_{3}^{1},\zeta_{3}^{1}, 3,\frac{3+\sqrt{45}}{2}$ & $0, 0, 0, \frac{3}{5}, \frac{3}{5}, \frac{3}{5}, \frac{1}{2}, \frac{1}{10}$ & $2^{ B}_{-14/5}\boxtimes 4^{\zeta_{4}^{1}}_{ 4}$\\
$8^{\zeta_{4}^{1}}_{-6/5}$ & $43.416$ & $1, 1, 1,\zeta_{3}^{1},\zeta_{3}^{1},\zeta_{3}^{1}, 3,\frac{3+\sqrt{45}}{2}$ & $0, 0, 0, \frac{2}{5}, \frac{2}{5}, \frac{2}{5}, \frac{1}{2}, \frac{9}{10}$ & $2^{ B}_{ 14/5}\boxtimes 4^{\zeta_{4}^{1}}_{ 4}$\\
 \hline 
\end{tabular} 
\end{table*}

Table~\ref{SETS3-567} lists all bosonic topological phases with $S_3$ symmetry,
for $N=4,5,6$ and $D^2\leq 100$,
$N=7$ and $D^2\leq 60$,
$N=8$ and $D^2\leq 40$.
\begin{table*}[h] 
\caption{Bosonic topological phases with $S_3$ symmetry.  } 
\label{SETS3-567} 
\medskip

\centering
\begin{tabular}{ |c|c|l|l|l| } 
\hline 
$N^{|\Th|}_{c}$ & $D^2$ & $d_1,d_2,\cdots$ & $s_1,s_2,\cdots$ & comment/$K$-matrix \\
\hline 
$3^{\sqrt{6}}_{ 0}$ & $6$ & $1, 1, 2$ & $0, 0, 0$ & $\cE=\Rp(S_3)$\\
\hline
$5^{\sqrt{6}}_{ 4}$ & $24$ & $1, 1, 2, 3, 3$ & $0, 0, 0, \frac{1}{2}, \frac{1}{2}$ & 
SB:$4^{ B}_{ 4}$ 
\\
$5^{\sqrt{6}}_{ 4}$ & $24$ & $1, 1, 2, 3, 3$ & $0, 0, 0, \frac{1}{2}, \frac{1}{2}$ & SB:$4^{ B}_{ 4}$ 
 $K=
\begin{pmatrix}
2 & 1 & 1 & 1 \\ 
1 & 2 & 0 & 0 \\ 
1 & 0 & 2 & 0\\
1 & 0 & 0 & 2\\
\end{pmatrix}
$
\\
$5^{\sqrt{6}}_{ 4}$ & $24$ & $1, 1, 2, 3, 3$ & $0, 0, 0, \frac{1}{2}, \frac{1}{2}$ & SB:$4^{ B}_{ 4}$ \\
\hline
$6^{\sqrt{6}}_{ 1}$ & $12$ & $1, 1, 2, 1, 1, 2$ & $0, 0, 0, \frac{1}{4}, \frac{1}{4}, \frac{1}{4}$ & $2^{ B}_{ 1}\boxtimes \Rp(S_3)$\\
$6^{\sqrt{6}}_{ 1}$ & $12$ & $1, 1, 2, 1, 1, 2$ & $0, 0, 0, \frac{1}{4}, \frac{1}{4}, \frac{1}{4}$ & SB:$2^{ B}_{ 1}$ \\
$6^{\sqrt{6}}_{-1}$ & $12$ & $1, 1, 2, 1, 1, 2$ & $0, 0, 0, \frac{3}{4}, \frac{3}{4}, \frac{3}{4}$ & $2^{ B}_{-1}\boxtimes \Rp(S_3)$\\
$6^{\sqrt{6}}_{-1}$ & $12$ & $1, 1, 2, 1, 1, 2$ & $0, 0, 0, \frac{3}{4}, \frac{3}{4}, \frac{3}{4}$ & SB:$2^{ B}_{-1}$ \\
$6^{\sqrt{6}}_{ 2}$ & $18$ & $1, 1, 2, 2, 2, 2$ & $0, 0, 0, \frac{1}{3}, \frac{1}{3}, \frac{1}{3}$ & SB:$3^{ B}_{ 2}$ \\
$6^{\sqrt{6}}_{ 2}$ & $18$ & $1, 1, 2, 2, 2, 2$ & $0, 0, 0, \frac{1}{3}, \frac{1}{3}, \frac{1}{3}$ & SB:$3^{ B}_{ 2}$ \\
$6^{\sqrt{6}}_{-2}$ & $18$ & $1, 1, 2, 2, 2, 2$ & $0, 0, 0, \frac{2}{3}, \frac{2}{3}, \frac{2}{3}$ & SB:$3^{ B}_{-2}$ \\
$6^{\sqrt{6}}_{-2}$ & $18$ & $1, 1, 2, 2, 2, 2$ & $0, 0, 0, \frac{2}{3}, \frac{2}{3}, \frac{2}{3}$ & SB:$3^{ B}_{-2}$ \\
$6^{\sqrt{6}}_{ 14/5}$ & $21.708$ & $1, 1, 2,\zeta_{3}^{1},\zeta_{3}^{1},\zeta_{8}^{4}$ & $0, 0, 0, \frac{2}{5}, \frac{2}{5}, \frac{2}{5}$ & $2^{ B}_{ 14/5}\boxtimes \Rp(S_3)$\\
$6^{\sqrt{6}}_{-14/5}$ & $21.708$ & $1, 1, 2,\zeta_{3}^{1},\zeta_{3}^{1},\zeta_{8}^{4}$ & $0, 0, 0, \frac{3}{5}, \frac{3}{5}, \frac{3}{5}$ & $2^{ B}_{-14/5}\boxtimes \Rp(S_3)$\\
 \hline 
\end{tabular} 
\end{table*}

Table~\ref{toplst} lists fermionic topological phases with $\Z_2^f$ symmetry.
The corresponding SPT is trivial and $c_\text{min}=1/2$. In the labels we use
$c$ mod $c_\text{min}$.
 For fermionic phases we
always have $\Theta=0$. Thus we add additional labels to distinguish entries.
$\Th_2$ in the table is defined as
$\Th_2 \equiv D^{-1}\sum_i \ee^{\ii 4\pi s_i} d_i^2$.
Also $\angle\Th_2 := \text{Im}\ln\Th_2$.
The table contains all fermionic topological orders with
$N=2$, $N=4$ and $D^2 \leq 600$, $N=6$ and $D^2 \leq 400$. They all have modular
extensions and are all realizable by (2+1)D fermion systems.
\begin{table*}[h] 
\caption{Fermionic topological phases with $\Z_2^f$ symmetry.  } 
\label{toplst} 
\medskip

\centering
\begin{tabular}{ |c|c|l|l|l| } 
\hline 
$N^F_c(\stck{|\Th_2|}{\angle\Th_2/2\pi})$ & $D^2$ & $d_1,d_2,\cdots$ & $s_1,s_2,\cdots$ & comment/$K$-matrix \\
 \hline 
$2_{ 0}^F(\stck{\zeta_2^1}{ 0})$ & $ 2$ & $1, 1$ & $0, \frac{1}{2}$ & trivial
$\cF_0=\sRep(\Z_2^f)$\\
\hline
$4_{ 0}^F(\stck{ 0}{0})$ & $ 4$ & $1, 1, 1, 1$ & $0, \frac{1}{2}, \frac{1}{4},-\frac{1}{4}$ & $\cF_0 \boxtimes 2_{ 1}^B$, $ K=\bpm 2\ 2\\2\ 1\epm$\\
$4_{ 1/5}^F(\stck{\zeta_2^1\zeta_3^1}{ 3/20})$ & $7.2360$ & $1, 1,\zeta_{3}^{1},\zeta_{3}^{1}$ & $0, \frac{1}{2}, \frac{1}{10},-\frac{2}{5}$ & $\cF_0 \boxtimes 2_{-14/5}^B$\\
$4_{-1/5}^F(\stck{\zeta_2^1\zeta_3^1}{-3/20})$ & $7.2360$ & $1, 1,\zeta_{3}^{1},\zeta_{3}^{1}$ & $0, \frac{1}{2},-\frac{1}{10}, \frac{2}{5}$ & $\cF_0 \boxtimes 2_{ 14/5}^B$\\
$4_{1/4}^F(\stck{\zeta_{6}^{3}}{ 1/2})$ & $13.656$ & $1,
1,\zeta_{6}^{2},\zeta_{6}^{2}=1+\sqrt{2}$ & $0, \frac{1}{2}, \frac{1}{4},-\frac{1}{4}$ &
Modular extension $SU(2)_6$ \\
\hline
$6_{ 0}^F(\stck{\zeta_2^1}{ 1/4})$ & $ 6$ & $1, 1, 1, 1, 1, 1$ & $0, \frac{1}{2}, \frac{1}{6},-\frac{1}{3}, \frac{1}{6},-\frac{1}{3}$ & $\cF_0 \boxtimes 3_{-2}^B$, $K=(3)$\\
$6_{ 0}^F(\stck{\zeta_2^1}{-1/4})$ & $ 6$ & $1, 1, 1, 1, 1, 1$ & $0, \frac{1}{2},-\frac{1}{6}, \frac{1}{3},-\frac{1}{6}, \frac{1}{3}$ & $\cF_0 \boxtimes 3_{ 2}^B$, $K=(-3)$\\
$6_{ 0}^F(\stck{\zeta_{6}^{3}}{ 1/16})$ & $ 8$ & $1, 1, 1, 1,\zeta_2^1,\zeta_2^1$ & $0, \frac{1}{2}, 0, \frac{1}{2}, \frac{1}{16},-\frac{7}{16}$ & $\cF_0 \boxtimes 3_{ 1/2}^B$\\
$6_{ 0}^F(\stck{\zeta_{6}^{3}}{-1/16})$ & $ 8$ & $1, 1, 1, 1,\zeta_2^1,\zeta_2^1$ & $0, \frac{1}{2}, 0, \frac{1}{2},-\frac{1}{16}, \frac{7}{16}$ & $\cF_0 \boxtimes 3_{-1/2}^B$\\
$6_{ 0}^F(\stck{1.0823}{ 3/16})$ & $ 8$ & $1, 1, 1, 1,\zeta_2^1,\zeta_2^1$ & $0, \frac{1}{2}, 0, \frac{1}{2}, \frac{3}{16},-\frac{5}{16}$ & $\cF_0 \boxtimes 3_{ 3/2}^B$\\
$6_{ 0}^F(\stck{1.0823}{-3/16})$ & $ 8$ & $1, 1, 1, 1,\zeta_2^1,\zeta_2^1$ & $0, \frac{1}{2}, 0, \frac{1}{2},-\frac{3}{16}, \frac{5}{16}$ & $\cF_0 \boxtimes 3_{-3/2}^B$\\
$6_{ 1/7}^F(\stck{\zeta_2^1\zeta_5^2}{-5/14})$ & $18.591$ & $1, 1,\zeta_{5}^{1},\zeta_{5}^{1},\zeta_{5}^{2},\zeta_{5}^{2}$ & $0, \frac{1}{2}, \frac{5}{14},-\frac{1}{7},-\frac{3}{14}, \frac{2}{7}$ & $\cF_0 \boxtimes 3_{ 8/7}^B$\\
$6_{-1/7}^F(\stck{\zeta_2^1\zeta_5^2}{ 5/14})$ & $18.591$ & $1, 1,\zeta_{5}^{1},\zeta_{5}^{1},\zeta_{5}^{2},\zeta_{5}^{2}$ & $0, \frac{1}{2},-\frac{5}{14}, \frac{1}{7}, \frac{3}{14},-\frac{2}{7}$ & $\cF_0 \boxtimes 3_{-8/7}^B$\\
$6_{0}^F(\stck{ \zeta_{10}^{5}}{-1/12})$ & $ 44.784$ & $1,
1,\zeta_{10}^{2},\zeta_{10}^{2},\zeta_{10}^{4},\zeta_{10}^{4}$ & $0,
\frac{1}{2}, \frac{1}{3},-\frac{1}{6}, 0, \frac{1}{2}$ & Modular extension
$\ov{SU(2)_{10}}$ \\
$6_{0}^F(\stck{ \zeta_{10}^{5}}{ 1/12})$ & $ 44.784$ & $1,
1,\zeta_{10}^{2},\zeta_{10}^{2},\zeta_{10}^{4},\zeta_{10}^{4}$ & $0,
\frac{1}{2},-\frac{1}{3}, \frac{1}{6}, 0, \frac{1}{2}$ & Modular extension $SU(2)_{10}$ \\
 \hline 
\end{tabular} 
\end{table*}

Table~\ref{SETZ2Z2f} lists fermionic topological phases with $\Z_2\times
\Z_2^f$ symmetry. The corresponding SPT is classified by $\Z_8$ and
$c_\text{min}=1/2$.
The list contains all topological orders with 
$N=6$ and $D^2\leq 300$,
$N=8$ and $D^2\leq 60$,
$N=10$ and $D^2\leq 20$.
\begin{table*}[h] 
\caption{Fermionic topological phases with $\Z_2\times \Z_2^f$ symmetry.  } 
\label{SETZ2Z2f} 
\medskip

\centering
\begin{tabular}{ |c|c|l|l|l| } 
\hline 
$N^{|\Th|}_{c}$ & $D^2$ & $d_1,d_2,\cdots$ & $s_1,s_2,\cdots$ & comment/$K$-matrix \\
\hline 
$4^{ 0}_{0}({ 2\atop  0})$ & $4$ & $1, 1, 1, 1$ & $0, 0, \frac{1}{2}, \frac{1}{2}$ & $\cE=\sRp(\Z_2\times \Z_2^f)$\\
\hline
$6^{ 0}_{0}$ & $12$ & $1, 1, 1, 1, 2, 2$ & $0, 0, \frac{1}{2}, \frac{1}{2}, \frac{1}{6}, \frac{2}{3}$ & 
SB:$K=$ $\bpm -1 & -2 \\ -2&-1\epm$\\
$6^{ 0}_{0}$ & $12$ & $1, 1, 1, 1, 2, 2$ & $0, 0, \frac{1}{2}, \frac{1}{2}, \frac{1}{3}, \frac{5}{6}$ & SB:$K=$ $\bpm 1 & 2 \\ 2&1\epm$\\
\hline
$8^{ 0}_{0}({ 0\atop 0})$ & $8$ & $1, 1, 1, 1, 1, 1, 1, 1$ & $0, 0, \frac{1}{2}, \frac{1}{2}, \frac{1}{4}, \frac{1}{4}, \frac{3}{4}, \frac{3}{4}$ & $2^{ B}_{ 1}\boxtimes \sRp(\Z_2\times \Z_2^f)$\\
$8^{ 0}_{0}({ 0\atop 0})$ & $8$ & $1, 1, 1, 1, 1, 1, 1, 1$ & $0, 0, \frac{1}{2}, \frac{1}{2}, \frac{1}{4}, \frac{1}{4}, \frac{3}{4}, \frac{3}{4}$ & SB:$4^{ F}_{0}({ 0\atop 0})$ \\
$8^{ 0}_{-14/5}({\zeta_{8}^{4}\atop  3/20})$ & $14.472$ & $1, 1, 1, 1,\zeta_{3}^{1},\zeta_{3}^{1},\zeta_{3}^{1},\zeta_{3}^{1}$ & $0, 0, \frac{1}{2}, \frac{1}{2}, \frac{1}{10}, \frac{1}{10}, \frac{3}{5}, \frac{3}{5}$ & $2^{ B}_{-14/5}\boxtimes \sRp(\Z_2\times \Z_2^f)$\\
$8^{ 0}_{14/5}({\zeta_{8}^{4}\atop -3/20})$ & $14.472$ & $1, 1, 1, 1,\zeta_{3}^{1},\zeta_{3}^{1},\zeta_{3}^{1},\zeta_{3}^{1}$ & $0, 0, \frac{1}{2}, \frac{1}{2}, \frac{2}{5}, \frac{2}{5}, \frac{9}{10}, \frac{9}{10}$ & $2^{ B}_{ 14/5}\boxtimes \sRp(\Z_2\times \Z_2^f)$\\
$8^{ 0}_{0}({ 2\atop  0})$ & $20$ & $1, 1, 1, 1, 2, 2, 2, 2$ & $0, 0, \frac{1}{2}, \frac{1}{2}, \frac{1}{10}, \frac{2}{5}, \frac{3}{5}, \frac{9}{10}$ & SB:$10^{ F}_{0}({\zeta_{2}^{1}\atop  0})$ \\
$8^{ 0}_{0}({ 2\atop  1/2})$ & $20$ & $1, 1, 1, 1, 2, 2, 2, 2$ & $0, 0, \frac{1}{2}, \frac{1}{2}, \frac{1}{5}, \frac{3}{10}, \frac{7}{10}, \frac{4}{5}$ & SB:$10^{ F}_{0}({\zeta_{2}^{1}\atop  1/2})$ \\
$8^{ 0}_{1/4}({\zeta_{2}^{1}\zeta_{6}^{3} \atop  1/2})$ & $27.313$ & $1, 1, 1, 1,\zeta_{6}^{2},\zeta_{6}^{2},\zeta_{6}^{2},\zeta_{6}^{2}$ & $0, 0, \frac{1}{2}, \frac{1}{2}, \frac{1}{4}, \frac{1}{4}, \frac{3}{4}, \frac{3}{4}$ & SB:$4^{ F}_{1/4}({\zeta_{6}^{3}\atop  1/2})$ \\
$8^{ 0}_{1/4}({\zeta_{2}^{1}\zeta_{6}^{3}\atop  1/2})$ & $27.313$ & $1, 1, 1, 1,\zeta_{6}^{2},\zeta_{6}^{2},\zeta_{6}^{2},\zeta_{6}^{2}$ & $0, 0, \frac{1}{2}, \frac{1}{2}, \frac{1}{4}, \frac{1}{4}, \frac{3}{4}, \frac{3}{4}$ & SB:$4^{ F}_{1/4}({\zeta_{6}^{3}\atop  1/2})$ \\
\hline
$10^{ 0}_{0}({ 4\atop  0})$ & $16$ & $1, 1, 1, 1, 1, 1, 1, 1, 2, 2$ & $0, 0, \frac{1}{2}, \frac{1}{2}, 0, 0, \frac{1}{2}, \frac{1}{2}, 0, \frac{1}{2}$ & SB:$8^{ F}_{0}({\sqrt{8}\atop  0})$ \\
$10^{ 0}_{0}({ 4\atop  0})$ & $16$ & $1, 1, 1, 1, 1, 1, 1, 1, 2, 2$ & $0, 0, \frac{1}{2}, \frac{1}{2}, 0, 0, \frac{1}{2}, \frac{1}{2}, 0, \frac{1}{2}$ & SB:$8^{ F}_{0}({\sqrt{8}\atop  0})$ \\
$10^{ 0}_{0}({\sqrt{8}\atop  1/8})$ & $16$ & $1, 1, 1, 1, 1, 1, 1, 1, 2, 2$ & $0, 0, \frac{1}{2}, \frac{1}{2}, 0, 0, \frac{1}{2}, \frac{1}{2}, \frac{1}{8}, \frac{5}{8}$ & SB:$8^{ F}_{0}({ 2\atop  1/8})$ \\
$10^{ 0}_{0}({\sqrt{8}\atop  1/8})$ & $16$ & $1, 1, 1, 1, 1, 1, 1, 1, 2, 2$ & $0, 0, \frac{1}{2}, \frac{1}{2}, 0, 0, \frac{1}{2}, \frac{1}{2}, \frac{1}{8}, \frac{5}{8}$ & SB:$8^{ F}_{0}({ 2\atop  1/8})$ \\
$10^{ 0}_{0}({ 0\atop 0})$ & $16$ & $1, 1, 1, 1, 1, 1, 1, 1, 2, 2$ & $0, 0, \frac{1}{2}, \frac{1}{2}, 0, 0, \frac{1}{2}, \frac{1}{2}, \frac{1}{4}, \frac{3}{4}$ & SB:$8^{ F}_{0}({ 0\atop 0})$ \\
$10^{ 0}_{0}({ 0\atop 0})$ & $16$ & $1, 1, 1, 1, 1, 1, 1, 1, 2, 2$ & $0, 0, \frac{1}{2}, \frac{1}{2}, 0, 0, \frac{1}{2}, \frac{1}{2}, \frac{1}{4}, \frac{3}{4}$ & SB:$8^{ F}_{0}({ 0\atop 0})$ \\
$10^{ 0}_{0}({\sqrt{8}\atop -1/8})$ & $16$ & $1, 1, 1, 1, 1, 1, 1, 1, 2, 2$ & $0, 0, \frac{1}{2}, \frac{1}{2}, 0, 0, \frac{1}{2}, \frac{1}{2}, \frac{3}{8}, \frac{7}{8}$ & SB:$8^{ F}_{0}({ 2\atop -1/8})$ \\
$10^{ 0}_{0}({\sqrt{8}\atop -1/8})$ & $16$ & $1, 1, 1, 1, 1, 1, 1, 1, 2, 2$ & $0, 0, \frac{1}{2}, \frac{1}{2}, 0, 0, \frac{1}{2}, \frac{1}{2}, \frac{3}{8}, \frac{7}{8}$ & SB:$8^{ F}_{0}({ 2\atop -1/8})$ \\
 \hline 
\end{tabular} 
\end{table*}

Table~\ref{SETZ4f} lists fermionic topological phases with $\Z_4^f$ symmetry.
The corresponding SPT is trivial and $c_\text{min}=1$.
The list contains all topological orders with 
$N=6$ and $D^2\leq 100$,
$N=8$ and $D^2\leq 60$,
$N=10$ and $D^2\leq 20$.
\begin{table*}[h] 
\caption{Fermionic topological phases with $ \Z_4^f$ symmetry.  } 
\label{SETZ4f} 
\medskip

\centering
\begin{tabular}{ |c|c|l|l|l| } 
\hline 
$N^{|\Th|}_{c}$ & $D^2$ & $d_1,d_2,\cdots$ & $s_1,s_2,\cdots$ & comment/$K$-matrix \\
\hline 
$4^{ 0}_{0}$ & $4$ & $1, 1, 1, 1$ & $0, 0, \frac{1}{2}, \frac{1}{2}$ & $\cE=\sRp(\Z_4^f)$\\
\hline
$6^{ 0}_{0}$ & $12$ & $1, 1, 1, 1, 2, 2$ & $0, 0, \frac{1}{2}, \frac{1}{2}, \frac{1}{6}, \frac{2}{3}$ & 
 $K=
-\begin{pmatrix}
1&2 \\
2&1\\ 
\end{pmatrix}
$
\\
$6^{ 0}_{0}$ & $12$ & $1, 1, 1, 1, 2, 2$ & $0, 0, \frac{1}{2}, \frac{1}{2}, \frac{1}{3}, \frac{5}{6}$ & 
 $K=
\begin{pmatrix}
1&2 \\
2&1\\ 
\end{pmatrix}
$
\\
\hline
$8^{ 0}_{0}$ & $8$ & $1, 1, 1, 1, 1, 1, 1, 1$ & $0, 0, \frac{1}{2}, \frac{1}{2}, \frac{1}{4}, \frac{1}{4}, \frac{3}{4}, \frac{3}{4}$ & $2^{ B}_{-1}\boxtimes \sRp(\Z_4^f)$\\
$8^{ 0}_{0}$ & $8$ & $1, 1, 1, 1, 1, 1, 1, 1$ & $0, 0, \frac{1}{2}, \frac{1}{2}, \frac{1}{4}, \frac{1}{4}, \frac{3}{4}, \frac{3}{4}$ & $2^{ B}_{ 1}\boxtimes \sRp(\Z_4^f)$\\
$8^{ 0}_{-14/5}$ & $14.472$ & $1, 1, 1, 1,\zeta_{3}^{1},\zeta_{3}^{1},\zeta_{3}^{1},\zeta_{3}^{1}$ & $0, 0, \frac{1}{2}, \frac{1}{2}, \frac{1}{10}, \frac{1}{10}, \frac{3}{5}, \frac{3}{5}$ & $2^{ B}_{-14/5}\boxtimes \sRp(\Z_4^f)$\\
$8^{ 0}_{14/5}$ & $14.472$ & $1, 1, 1, 1,\zeta_{3}^{1},\zeta_{3}^{1},\zeta_{3}^{1},\zeta_{3}^{1}$ & $0, 0, \frac{1}{2}, \frac{1}{2}, \frac{2}{5}, \frac{2}{5}, \frac{9}{10}, \frac{9}{10}$ & $2^{ B}_{ 14/5}\boxtimes \sRp(\Z_4^f)$\\
$8^{ 0}_{0}$ & $20$ & $1, 1, 1, 1, 2, 2, 2, 2$ & $0, 0, \frac{1}{2}, \frac{1}{2}, \frac{1}{10}, \frac{2}{5}, \frac{3}{5}, \frac{9}{10}$ &  SB:$10^{ F}_{0}({\zeta_{2}^{1}\atop  0})$ \\
$8^{ 0}_{0}$ & $20$ & $1, 1, 1, 1, 2, 2, 2, 2$ & $0, 0, \frac{1}{2}, \frac{1}{2}, \frac{1}{5}, \frac{3}{10}, \frac{7}{10}, \frac{4}{5}$ & SB:$10^{ F}_{0}({\zeta_{2}^{1}\atop  1/2})$ \\
\hline
$10^{ 0}_{0}({ 4\atop  0})$ & $16$ & $1, 1, 1, 1, 1, 1, 1, 1, 2, 2$ & $0, 0, \frac{1}{2}, \frac{1}{2}, 0, 0, \frac{1}{2}, \frac{1}{2}, 0, \frac{1}{2}$ & SB:$8^{ F}_{0}({\sqrt{8}\atop  0})$ \\
$10^{ 0}_{0}({ 4\atop  0})$ & $16$ & $1, 1, 1, 1, 1, 1, 1, 1, 2, 2$ & $0, 0, \frac{1}{2}, \frac{1}{2}, 0, 0, \frac{1}{2}, \frac{1}{2}, 0, \frac{1}{2}$ & SB:$8^{ F}_{0}({\sqrt{8}\atop  0})$ \\
$10^{ 0}_{0}({\sqrt{8}\atop  1/8})$ & $16$ & $1, 1, 1, 1, 1, 1, 1, 1, 2, 2$ & $0, 0, \frac{1}{2}, \frac{1}{2}, 0, 0, \frac{1}{2}, \frac{1}{2}, \frac{1}{8}, \frac{5}{8}$ & SB:$8^{ F}_{0}({ 2\atop  1/8})$ \\
$10^{ 0}_{0}({\sqrt{8}\atop  1/8})$ & $16$ & $1, 1, 1, 1, 1, 1, 1, 1, 2, 2$ & $0, 0, \frac{1}{2}, \frac{1}{2}, 0, 0, \frac{1}{2}, \frac{1}{2}, \frac{1}{8}, \frac{5}{8}$ & SB:$8^{ F}_{0}({ 2\atop  1/8})$ \\
$10^{ 0}_{0}({ 0\atop 0})$ & $16$ & $1, 1, 1, 1, 1, 1, 1, 1, 2, 2$ & $0, 0, \frac{1}{2}, \frac{1}{2}, 0, 0, \frac{1}{2}, \frac{1}{2}, \frac{1}{4}, \frac{3}{4}$ & SB:$8^{ F}_{0}({ 0\atop 0})$ \\
$10^{ 0}_{0}({ 0\atop 0})$ & $16$ & $1, 1, 1, 1, 1, 1, 1, 1, 2, 2$ & $0, 0, \frac{1}{2}, \frac{1}{2}, 0, 0, \frac{1}{2}, \frac{1}{2}, \frac{1}{4}, \frac{3}{4}$ & SB:$8^{ F}_{0}({ 0\atop 0})$ \\
$10^{ 0}_{0}({\sqrt{8}\atop -1/8})$ & $16$ & $1, 1, 1, 1, 1, 1, 1, 1, 2, 2$ & $0, 0, \frac{1}{2}, \frac{1}{2}, 0, 0, \frac{1}{2}, \frac{1}{2}, \frac{3}{8}, \frac{7}{8}$ & SB:$8^{ F}_{0}({ 2\atop -1/8})$ \\
$10^{ 0}_{0}({\sqrt{8}\atop -1/8})$ & $16$ & $1, 1, 1, 1, 1, 1, 1, 1, 2, 2$ & $0, 0, \frac{1}{2}, \frac{1}{2}, 0, 0, \frac{1}{2}, \frac{1}{2}, \frac{3}{8}, \frac{7}{8}$ & SB:$8^{ F}_{0}({ 2\atop -1/8})$ \\
 \hline 
\end{tabular} 
\end{table*}



\end{document}